\documentclass[article]{IEEEtran}
\usepackage{mathrsfs}
\usepackage{psfrag}
\usepackage{graphicx}
\usepackage{fancyhdr,epsfig}
\usepackage{mathtools}
\usepackage{research5IEEE}
\usepackage{amsmath}
\usepackage{amsfonts}
\usepackage{amssymb}
\usepackage{amsthm}
\usepackage{ booktabs}
\usepackage{array}  
\allowdisplaybreaks[4]
\newcommand{\FB}{\textnormal{Fb}}

\newcommand{\En}{{(\textnormal{Enh})}}

\usepackage{hyperref}
\newcommand{\M}{{(\textnormal{M})}}
\newcommand{\Epsilon}{\mathcal{E}}
\newtheorem{Theorem}{Theorem}

\newtheorem{Remark}{Remark}
\newtheorem{Example}{Example}
\newtheorem{Corollary}{Corollary}
\newtheorem{Definition}{Definition}
\newtheorem{Proposition}{Proposition}

\begin{document}
\title{Coding Schemes with Rate-Limited Feedback that Improve over the Nofeedback Capacity for a Large Class of Broadcast Channels}

\author{Youlong Wu, Mich\`ele Wigger%
\thanks{This paper was in part presented at the \emph{IEEE Information Theory Workshop}, in Sevilla, Spain, Sep. 2013, and at the \emph{IEEE International Symposium on Information Theory}, in Honolulu, HI, USA, June--July 2014. }
\thanks{Y. Wu was with Telecom ParisTech, CNRS LTCI, Université Paris-Saclay, 75013 Paris, France. He is now with the Institute for Communications Engineering, Technische Universit\"at M\"unchen, Munich, Germany (e-mail: youlong.wu@tum.de). M. Wigger is with the Department of Communications and Electronics, Telecom ParisTech, CNRS LTCI, Université Paris-Saclay, 75013 Paris, France. (e-mail: michele.wigger@telecom-paristech.fr). }%
\thanks{M. Wigger was supported by the city of Paris under program ``Emergences``.}
}



\maketitle

{\begin{abstract}
We propose two coding schemes for the two-receiver discrete memoryless broadcast channel (BC) with rate-limited feedback from one or both receivers. They
improve over the nofeedback capacity region for a large class of channels, including the class of \emph{strictly essentially less-noisy BCs} that we introduce in this article. Examples of strictly essentially less-noisy BCs are the binary symmetric BC (BSBC) or the binary erasure BC (BEBC) with unequal  cross-over or erasure probabilities {at} the two receivers.
When the feedback rates are sufficiently large,  our schemes recover all previously known capacity results for discrete memoryless BCs with feedback. 

In both our schemes, we let the receivers feed back quantization  messages about their receive signals. In the first scheme, the transmitter simply \emph{relays} the quantization information obtained from Receiver~1 to Receiver~2, and vice versa. {This provides each receiver} with a second observation of the input signal and can thus improve its decoding performance unless the BC is physically degraded. {Moreover, each receiver uses its knowledge of the quantization message describing its own outputs so as to attain  the same performance as if this message had not been transmitted at all.}

In our second scheme the transmitter first \emph{reconstructs and processes} the quantized output signals, and then sends the outcome as a common update information to both receivers. A special case of our second scheme applies also to memoryless BCs without feedback but with strictly-causal state-information at the transmitter and causal state-information at the receivers. It recovers all previous achievable regions also for this setup with state-information.

\end{abstract}}

\begin{IEEEkeywords}
Broadcast channel, channel capacity,  rate-limited feedback
\end{IEEEkeywords}

\section{Introduction}

For most discrete  memoryless broadcast channels (DMBC), it is not known whether feedback can increase the capacity region, even when the feedback links are noise-free and of infinite rate.
There are some exceptions. For {example, for} all physically degraded DMBCs the capacity regions with and without feedback coincide \cite{gamal'78}.  The first simple example DMBC where feedback  increases capacity was presented by Dueck \cite{dueck}. His example and coding scheme were generalized by Shayevitz and Wigger \cite{wigger} who proposed a general scheme and achievable region for DMBCs with generalized feedback. In the generalized feedback model, the feedback to the transmitter is modeled as an additional output of the DMBC that can depend on the input and the receivers' outputs in an arbitrary manner. 

Other achievable regions for 
 general DMBCs with perfect or noisy feedback have been proposed by Kramer \cite{kramer'03} and by Venkataramanan and Pradhan \cite{venkataramananpradhan11}. Kramer's achievable region {can be used to show} that feedback  improves  capacity for some specific \emph{binary symmetric BCs}  (BSBC). {Comparing the general achievable regions in \cite{wigger, kramer'03,  venkataramananpradhan11} to each other is hard because of their complex form which involves several auxiliary random variables.}

A different line of works has concentrated on the 
memoryless {Gaussian} broadcast channel (BC) \cite{ozarow'84, Wu05,Elia04,minero,belhadj'14, bhaskaran, wigger'14, BenyishaiShayevitz, pillai,wuminerowigger14}. The best achievable region for perfect feedback and when the noises at the two receivers are independent  is given in \cite{belhadj'14}  and is based on a MAC-BC duality approach. In \cite{wigger'14}, the  asymptotic high-SNR sum-capacity for arbitrary noise correlation is derived. 

 {In this paper, we consider DMBCs with \emph{rate-limited feedback}, where the feedback links from the receivers to the transmitter are assumed to be instantaneous and noiseless but {rate-limited}. We present two {types of} coding schemes that use Marton coding \cite{marton'79}  in a block-Markov framework (similar to \cite{wigger}), and where in both types the receivers feed back compression information about their channel outputs of the previous block.}

In our type-I  schemes, (Schemes IA--IC), the encoder simply \emph{relays} the obtained compression informations as part of the cloud center of the Marton code employed in the next-following block. {Each receiver reconstructs the compressed version of the other receiver's outputs, and decodes its intended data and compression information based on this compressed signal and its own channel outputs. {The key novelty of our scheme is that in this decoding each receiver cleverly uses its knowledge of the compression messages describing its own previous outputs in a way as to attain the same performance as if this message had not been transmitted at all. }}


In our type-II  scheme, (Scheme~II), the encoder uses the feedback messages to  \emph{reconstruct}  compressed versions of both receivers' channel outputs, and then \emph{processes} these compressed signals together with the previously sent codewords to generate  update information for the two receivers. This update information is sent  as part of the cloud center of the Marton code employed in the next-following block. Each receiver uses backward decoding to simultaneously reconstruct the encoder's compressed signal and  decode its intended messages sent in the cloud center and satellite. 

{
Our coding schemes exhibit the following features: 
\begin{itemize}
\item Unlike previous schemes, our new schemes  strictly improve over the nofeedback capacity region for a large class of memoryelss BCs, which includes:
\begin{enumerate}
\item \emph{Strictly essentially less-noisy {DMBCs}}, which we define in this paper and which represent a subclass of Nair's essentially less-noisy DMBCs \cite{CNair'10}. They include as special cases the  BSBC and the  \emph{binary erasure BC} (BEBC) with unequal cross-over probabilities or unequal erasure probabilities at the two receivers, and the {\emph{binary symmetric/binary erasure channel BC} (BSC/BEC-BC)} for a large range of parameters. 
\item The {BSC/BEC-BC} for all parameters where this DMBC is more capable \cite{gamal'79} and the BSC and the BEC have different capacities.  
\item The memoryless Gaussian BC with unequal noise variances at the two receivers.\footnote{Interestingly, our result hinges on the fact that the receivers are allowed to code over the feedback links: A recent result by Pillai and  Prabhakaran \cite{pillai} shows that when the feedback links are additive Gaussian noise channels of noise variance exceeding a certain threshold, then one cannot improve over the nofeedback capacity if the receivers simply feed back their channel outputs.}
\item An instance of the semideterministic BC as is proved in \cite{annina}.
\end{enumerate}
\item 
When the feedback-rates are sufficiently large,  our new schemes recover all previously known  capacity results for DMBCs with perfect feedback. In particular, they improve over the Shayevitz-Wigger scheme for perfect feedback \cite{wigger} when this latter is restricted to send all the update information in the cloud center. This represents a {prominent} special case of the Shayevitz-Wigger scheme that subsumes the capacity-achieving scheme by Wang \cite{wang} or by Georgiadis \& Tassiulas \cite{greek} for the two-user BEBC where both receivers know all erasure events. 


\item Subject to a slight modification, our coding schemes apply also to a setup with noisy feedback links when the receivers can code over them. All our achievable regions remain valid also in this modified setup.

\item A special case of our type-II coding scheme applies also to state-dependent DMBCs without feedback but where  the receivers learn the state causally and the transmitter learns it strictly causally. Our new achievable region for this state-dependent setup recovers all previous achievable regions. In particular the Degrees of Freedom (DoF) result by Maddah-Ali \& Tse \cite{maddah}, and the achievable regions  and capacity regions presented by Kim, Chia, and El Gamal  in \cite{kimchia}.

\end{itemize}}

 {The idea of our type-I schemes  extends to  more general networks. In \cite{Youlong'Allerton}, such extended coding schemes  are proposed for the  discrete memoryless multicast network (DMMN), where one transmitter wishes to communicate a message to multiple receivers over a relay network. The result shows that {with feedback}, one can strictly improve over noisy network coding~\cite{Lim'11} and distributed decode-forward coding \cite{Lim'ITW} for some DMMNs.}

The rest of this paper is organized as follows: In Sections~\ref{system} and \ref{sec:prelim}, we describe the channel model and present some previous results on BCs with and without feedback.  In Section~\ref{Sec:Definition}, we { define the class of strictly essentially less-noisy BCs.  In Section \ref{sec:simple}, we propose a simple coding scheme that motivates our work.   Sections~\ref{Sec:NewRegions} and \ref{Sec:Usefulness} present our main results: new achievable regions with rate-limited feedback and conditions under which they improve over the nofeedback capacity region.   {Section~\ref{sec:allschm} describes our new coding schemes achieving the rate regions in Section~\ref{Sec:NewRegions}.
In Section \ref{Sec:Relations}, we compare our new achievable regions with previous results and discuss extensions of our results to related setups. Finally, in Section \ref{sec:Examples}, we numerically evaluate one of our achievable regions for several examples.

\subsection{Notation}
Let $\mathbb{R}$ denote the set of real numbers and $\mathbb{Z}^+$ the set of positive integers. We use capital letters to denote  random variables and small letters for their realizations, e.g., $X$ and $x$. For  $j\in\mathbb{Z}^+$, we use the short hand notations $X^j$ and $x^j$ for the $j$-tuples $X^j:=(X_1,\ldots, X_j)$ and $x^j:=(x_1,\ldots, x_j)$. Sets are usually denoted by caligraphic letters, e.g., $\set{S}$.  For a finite set $\set{S}$, we use $|\set{S}|$ for its cardinality and $\set{S}^j$ for its $j$-fold Cartesian product $\set{S}^j:=\set{S}\times \cdots \times \set{S}$, for $j\in\mathbb{Z}^+$.
For a subset $\set{S}\subset\Reals^2$ we use 
$\textnormal{bd}(\set{S})$ to denote its boundary and  $\textnormal{int}(\set{S})$ for its interior.
We also use caligraphic letters for events, mostly $\mathcal{E}$. Moreover, we denote the complement of event $\mathcal{E}$ by $\mathcal{E}^c$.

Given a distribution $P_A$ over some alphabet $\set{A}$, a positive real number $\varepsilon>0$, and a positive integer $n$, $\set{T}_{\varepsilon}^{(n)}(P_A)$ is the typical set in \cite{book:gamal}. 
Given a positive integer $n$, let $\mathbf{1}_{[n]}$ denote the all-one tuple of length $n$, e.g., $\mathbf{1}_{[3]}=(1,1,1)$.

We use definitions $\bar{a}:=(1-a)$ and $a* b:=\bar{a}b+a\bar{b}$, for  $a, b\in[0,1]$. 
Moreover, $Z\sim \textnormal{Bern}(p)$ denotes that $Z$ is  a binary random variable taking values 0 and 1 with probabilities $1-p$ and $p$. We use $H_\textnormal{b}(\cdot)$ for the binary entropy function; thus the entropy of random variable $Z\sim \textnormal{Bern}(p)$ is given by $H_\textnormal{b}(p)$.

}


}
\section{Channel model}\label{system}

\begin{figure}[!t]
\centering
\includegraphics[width=0.45\textwidth]{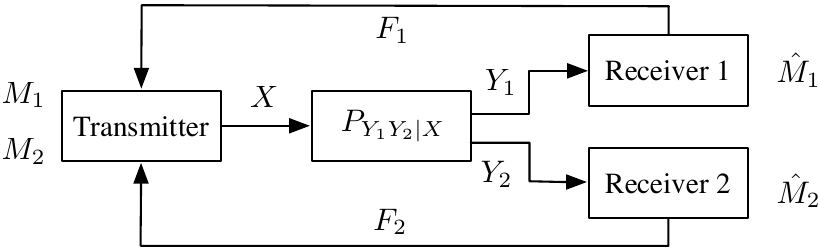}
\caption{Broadcast channel with rate-limited  feedback } \label{fig:BCFBmodel}
\vspace{-5mm}
\end{figure}

Communication takes place over a DMBC with rate-limited feedback, see Figure~\ref{fig:BCFBmodel}. The setup is characterized by the finite input alphabet $\set{X}$, the finite output alphabets $\set{Y}_1$ and $\set{Y}_2$, the channel law $P_{Y_1Y_2|X}$, and nonnegative feedback rates $R_{\textnormal{Fb},1}$ and $R_{\textnormal{Fb},2}$. 
If at  discrete-time $t$ the transmitter sends the channel input $x_t\in \set{X}$, then Receiver~$i\in\{1,2\}$ observes the output $Y_{i,t}\in \set{Y}_i$, where the pair $(Y_{1,t}, Y_{2,t})\sim P_{Y_1Y_2|X}(\cdot, \cdot|x_t)$. Also, after observing $Y_{i,t}$, Receiver~$i$ can send a feedback signal $F_{i,t}\in \set{F}_{i,t}$ to the transmitter, where $\set{F}_{i,t}$ denotes the finite alphabet of $F_{i,t}$ and is a design parameter of a scheme. The feedback link between  the transmitter and Receiver $i$ is assumed to be instantaneous and noiseless but \emph{rate-limited} to $R_{\textnormal{Fb},i}$ bits on average. Thus,
if the transmission takes place over a total blocklength $N$, then
\begin{subequations} \label{consFB0}
\begin{IEEEeqnarray}{rCl}
|\set{F}_{i,1}|&\times&\cdots\times|\set{F}_{i,N}|\leq 2^{NR_{\textnormal{Fb},i}}, \quad i\in\{1,2\}.
\end{IEEEeqnarray}
\end{subequations}

The goal of the communication is that the transmitter conveys two independent private messages $M_1\in \{1,\ldots,\lfloor 2^{NR_1} \rfloor\}$ and $M_2\in \{1,\ldots,\lfloor 2^{NR_2} \rfloor\}$, to Receiver~1 and 2, respectively. Each $M_i$, $i\in\{1,2\}$, is uniformly distributed over the set $\set{M}_i:=\{1,\ldots,\lfloor 2^{NR_i} \rfloor\}$, where $R_i$ denotes the private rate of transmission of Receiver $i$.

The transmitter is comprised of a sequence of encoding functions $\big\{f^{(N)}_t\big\}^N_{t=1}$ of the form
$ f_t^{(N)}: \set{M}_1\times\set{M}_2\times{\set{F}}_{1,1}\times\cdots \times{\set{F}}_{1,t-1}\times{\set{F}}_{2,1}\times\cdots \times{\set{F}}_{2,t-1}\to \set{X}$
that is used to produce the channel inputs as
\begin{equation}
X_t=f^{(N)}_t\big(M_1,M_2,{F}^{t-1}_1,{F}^{t-1}_2\big), \qquad t\in\{1,\ldots, N\}.
\end{equation}

Receiver~$i\in\{1,2\}$ is comprised of a sequence of feedback-encoding functions $\big\{\psi^{(N)}_{i,t}\big\}_{t=1}^N$ of the form
$
\psi^{(N)}_{i,t}: \set{Y}_i^{t}\to \set{F}_{i,t}
$
that is used to produce the symbols
\begin{equation}
F_{i,t}=\psi^{(N)}_{i,t}(Y_{i,1},\ldots,Y_{i,t}), \qquad t\in\{1,\ldots, N\},
\end{equation}
sent over the feedback link,
and of a decoding function
$
\Phi^{(N)}_i: \set{Y}_i^{N}\to\set{M}_i$
used to produce a guess of Message $M_i$:
\begin{equation}
\hat{M}_i=\Phi^{(N)}_{i}(Y_i^N).
\end{equation}

A rate region $(R_1, R_2)$ with averaged feedback rates $R_{\textnormal{Fb},1}\geq 0$, $R_{\textnormal{Fb},2}\geq 0$ is called achievable if for every blocklength $N$, there exists a set of encoding functions $\big\{f_t^{(N)}\big\}_{t=1}^N$ and ñfor $i\in\{1,2\}$ there exists a set of decoding functions $\Phi^{(N)}_i$, feedback alphabets $\{\set{F}_{i,t}\}^N_{t=1}$ satisfying~\eqref{consFB0}, and feedback-encoding functions $\big\{\psi_{i,t}^{(N)} \big\}_{t=1}^N$  such that the error probability
\begin{equation}\label{errorprob}P^{(N)}_e:=\Pr\big(M_1\neq \hat{M}_1~\textnormal{or}~M_2 \neq \hat{M}_2\big)\end{equation}
tends to zero as the blocklength $N$ tends to infinity. The closure of the set of achievable rate pairs $(R_1, R_2)$ is called the \textit{feedback capacity region} and is denoted by $\set{C}_\textnormal{Fb}$.

In the special case $R_{\textnormal{Fb},1}=R_{\textnormal{Fb},2}=0$ the feedback signals are constant and the setup is equivalent to a setup without feedback. We denote the capacity region for this setup  $\set{C}_\textnormal{NoFB}$. 

{When $R_{\FB,1}\geq \log_2|\mathcal{Y}_1|$ and $R_{\FB,2}\geq \log_2|\mathcal{Y}_2|$, our setup is equivalent to a perfect-feedback setup where after each channel use the receivers feed back their channel outputs. }

\section{Previous inner and outer bounds} \label{sec:prelim}
We recall some previous results on the capacity region of DMBCs without and with feedback. 

\subsection{DMBC without feedback}\label{sec:nofb}

\subsubsection{Marton's coding} The capacity region of DMBCs without feedback is in general unknown. 
The best known inner bound {without feedback} is Marton's region \cite{marton'79}, $\set{R}_{\textnormal{Marton}}$, which is the set of all nonnegative  rate pairs $(R_1,R_2)$ satisfying
\begin{subequations}\label{eq:martonRegion}
\begin{IEEEeqnarray}{rCl}
R_1&\leq& I(U_0,U_1;Y_1) \label{eq:Marton1}\\
R_2&\leq& I(U_0,U_2;Y_2)\label{eq:Marton2}\\
R_1\!+\!R_2&\leq&I(U_0,U_1;\!Y_1)\!+\!I(U_2;\!Y_2|U_0)\!-\!I(U_1;\!U_2|U_0)\label{eq:Marton3}\\
R_1\!+\!R_2&\leq&I(U_0,U_2;\!Y_2)\!+\!I(U_1;\!Y_1|U_0)\!-\!I(U_1;\!U_2|U_0)\label{eq:Marton4}\IEEEeqnarraynumspace
\end{IEEEeqnarray}
\end{subequations}
for some probability mass function (pmf)   $P_{U_0U_1U_2}$ and a function $f\colon \set{U}_0\times \set{U}_1 \times \set{U}_2 \to \set{X}$ such that  $X=f(U_0,U_1,U_2)$.

To evaluate Marton's region, it suffices to consider distributions $P_{U_0U_1U_2X}$ for which  {one of the following conditions holds {\cite{Pinsker'80,Kramer'08,Youlong'thesis}}:
\begin{itemize}
\item $I(U_0;Y_1)=I(U_0;Y_2)$;
\item $I(U_0;Y_1)<I(U_0;Y_2)$ and $U_1=\textnormal{const.}$;
\item $I(U_0;Y_1)>I(U_0;Y_2)$ and $U_2=\textnormal{const.}$.
\end{itemize}

\subsubsection{Superposition coding region} 
An important subset of Marton's region is the \emph{superposition coding region}, $\set{R}_{\textnormal{SuperPos}}^{(1)}$, which results when  Marton's constraints~\eqref{eq:martonRegion} are specialized to $U_1=\textnormal{const.}$ and $X=U_2$. That means, $\set{R}_{\textnormal{SuperPos}}^{(1)}$ is defined as the set of all nonnegative rate pairs $(R_1,R_2)$ satisfying 
\begin{subequations}\label{SPcoding}
\begin{IEEEeqnarray}{rCl}
R_1&\leq& I(U;Y_1) \\
R_2&\leq& I(X;Y_2|U)\\
R_1+R_2&\leq&I(X;Y_2)
\end{IEEEeqnarray}
\end{subequations}
for some pmf $P_{UX}$. The superposition coding region $\set{R}_{\textnormal{SuperPos}}^{(2)}$ is defined {in the same way as  $\set{R}_{\textnormal{SuperPos}}^{(1)}$ but with exchanged indices $1$ and $2$.

\subsubsection{Nair-El Gamal outer bound} 
In \cite{CNair'07}  Nair-El Gamal proposed an outer bound on the capacity region of DMBCs without feedback. It is the set of all nonnegative }rate pairs $(R_1,R_2)$ satisfying 
   \begin{subequations}\label{outerbound}
\begin{IEEEeqnarray}{rCl}
R_1&\leq& I(U;Y_1) \\
R_2&\leq& I(V;Y_2)\\
R_1+R_2&\leq&I(U;Y_1)+I(X;Y_2|U)\\
R_1+R_2&\leq&I(V;Y_2)+I(X;Y_1|V)
\end{IEEEeqnarray}
\end{subequations}
for some pmf $P_{UVX}$.

\begin{figure*}[ht]
\begin{subequations}\label{eq:inner}
\begin{IEEEeqnarray}{rCl}
 R_1 &\leq& I(U_0,U_1;Y_1,V_1|Q) - I(U_0,U_1,U_2,Y_2;V_0,V_1|Y_1,Q)\label{eq:inner1}
\\
 R_2 &\leq& I(U_0,U_2;Y_2,V_2|Q) - I(U_0,U_1,U_2,Y_1;V_0,V_2|Y_2,Q)\label{eq:inner2}
\\
 R_1+R_2 &\leq& I(U_1; Y_1,V_1|U_0,Q) + I(U_2; Y_2,V_2|U_0,Q) + \min_{i\in\{1,2\}}I(U_0;Y_i,V_i|Q)-\max_{i\in\{1,2\}} I(U_0,U_1,U_2,Y_1,Y_2;V_0|Y_i,Q)\nonumber \\
 && -I(U_0,U_1,U_2,Y_2;V_1|V_0,Y_1,Q) -I(U_0,U_1,U_2,Y_1,Y_2;V_2|V_0,Y_2,Q) -    I(U_1;U_2|U_0,Q) \label{eq:swSum1}
\\
\nonumber R_1+R_2 &\leq& I(U_0,U_1; Y_1,V_1|Q) + I(U_0,U_2; Y_2,V_2|Q) - I(U_1;U_2|U_0,Q) \\
&& - I(U_0,U_1,U_2,Y_2;V_0,V_1|Y_1,Q) -I(U_0,U_1,U_2,Y_1;V_0,V_2|Y_2,Q)\label{eq:inner4}
\end{IEEEeqnarray}
\end{subequations}
\rule{\textwidth}{0.6pt}
\end{figure*}

The Nair-El Gamal outer bound is known to coincide with Marton's region for the following classes of DMBCs,  which also play a role in the present paper: 
\begin{itemize}
\item \emph{stochastically or physically degraded DMBCs \cite{Cover'98}} 
\item \emph{less noisy DMBCs \cite{Korner'77}} 
\item \emph{essentially less noisy DMBCs \cite{CNair'10}} 
\item \emph{more capable DMBCs \cite{Korner'77}}.
\end{itemize}
In all these classes of DMBCs one of the two receivers is stronger than the other receiver in some sense. This makes that superposition coding is as good as the more general Marton coding and achieves capacity.

{The various classes of BCs satisfy the relationships} \cite{CNair'10}, \cite{konerMarton'75}:
\begin{itemize}
\item degraded $\subsetneq$ less-noisy $\subsetneq$ more capable,
\item less noisy $\subsetneq$ essentially less noisy,
\item  essentially less-noisy $\nsubseteq$ more capable,
\item  more capable $\nsubseteq$ essentially less-noisy. 
\end{itemize}

\subsection{DMBC with feedback}\label{sec:prevfb}

Previous results on the DMBC with feedback mostly focus on perfect feedback, {which in our setup corresponds to $R_{\FB,1}\geq \log_{2}|\set{Y}_1|$ and $R_{\FB,2}\geq \log_2|\set{Y}_2|$. The previous results that  are most closely related to our work are:} 

\subsubsection{Shayevitz-Wigger achievable region}

The achievable region with feedback that is most closely related to our paper is the Shayevitz-Wigger region  \cite{wigger}. It is the 
 set of all nonnegative rate pairs $(R_1,R_2)$  that satisfy~\eqref{eq:inner} for some pmf $P_QP_{U_0U_1U_2|Q}$$P_{V_0V_1V_2|U_0U_1U_2Y_1Y_2Q}$ and some function $f\colon \set{Q}\times\set{U}_0\times \set{U}_1 \times \set{U}_2 \to \set{X}$, where $X=f(U_0, U_1,U_2,Q)$.

\subsubsection{Ozarow-Leung outer bound}
A simple outer bound on the capacity region with output feedback is given in \cite{ozarow'84}. It equals the capacity region $\set{C}_\textnormal{Enh}^{(1)}$ of an enhanced DMBC where the outputs $Y_1^n$ are also revealed to Receiver 2. Notice that this enhanced  DMBC is {physically degraded} and thus, with and without feedback, its capacity region is given by the set of all nonnegative rate pairs $(R_1,R_2)$ that satisfy
\begin{subequations}\label{eq:Enh}
\begin{IEEEeqnarray}{rCl}
R_1 & \leq & I(U;Y_1)\\
R_2 & \leq & I(X;Y_1,Y_2|U)
\end{IEEEeqnarray}
\end{subequations}
for some pmf $P_{UX}$.

Exchanging everywhere in the previous paragraph  indices~$1$ and $2$, we can define a similar enhanced capacity region $\set{C}_{\textnormal{Enh}}^{(2)}$, which is also an outer bound to $\set{C}_{\FB}$. The intersection  $\set{C}_{\textnormal{Enh}}^{(1)}\cap \set{C}_{\textnormal{Enh}}^{(2)}$  yields an even tighter outerbound \cite{wang,greek}.

\section{Definitions}\label{Sec:Definition}

We recall the definition of {essentially less-noisy DMBCs} since they are  important for this paper. 

\begin{Definition}[From \cite{CNair'10}]\label{def:essentially}
A subset $\mathcal{P}_\set{X}$ of all pmfs  on the input alphabet $\mathcal{X}$ is said to be a  \emph{sufficient class} of pmfs  for a DMBC if the following holds: Given any joint pmf $P_{UVX}$ there exists a joint pmf $P_{UVX}'$ that satisfies 
\begin{IEEEeqnarray}{rCl}
P_X'(x)&\in& \mathcal{P}_\set{X} \nonumber\\
I_P(U;Y_1) &\leq& I_{P'}({U};Y_1) \nonumber\\
I_P(V;Y_2) &\leq& I_{P'}({V};Y_2) \nonumber\\ 
I_P(U;Y_1) + I_{P}(X;Y_2|U) &\leq& I_{P'}({U};Y_1)+I_{P'}(X;Y_2|{U}) \nonumber\\
I_P(V;Y_2)  +I_P(X;Y_1|V) &\leq& I_{P'}({V};Y_2)+I_{P'}(X;Y_1|{V}) \nonumber\IEEEeqnarraynumspace
\end{IEEEeqnarray}
where the notations $I_P$ and $I_{P'}$ indicate that the mutual informations are computed assuming that $(U,V,X)\sim P_{UVX}$ and $(U,V,X)\sim P_{UVX}'$ and $P'_X(x)$ is the marginal obtained from $P'_{UVX}$.
\end{Definition}
\begin{Definition}[From \cite{CNair'10}]
A DMBC is called \emph{essentially less-noisy} if there exists a \emph{sufficient class} of pmfs $\set{P}_\set{X}$ such that whenever $P_X\in \mathcal{P}_\set{X}$, then for all conditional pmfs $P_{U|X}$,
 \begin{equation}\label{essenLN}
 I(U;Y_1)\leq I(U;Y_2).
 \end{equation}
 \end{Definition}
The class of essentially less-noisy DMBCs contains as special cases the BSBC and the BEBC. Also the memoryless Gaussian BC is essentially less noisy. 

For essentially less-noisy DMBCs, Marton's coding (or superposition coding) is known to achieve capacity \cite{CNair'10}. To evaluate the superposition coding region $\set{R}_{\textnormal{SuperPos}}^{(1)}$ of an essentially less-noisy DMBC, it suffices to evaluate the region given by constraints~\eqref{SPcoding} for pmfs $P_{UX}$ that satisfy $I(U;Y_1)\leq I(U;Y_2)$.

In this paper, we introduce the new term \emph{strictly essentially less-noisy BC},  a subclass of  essentially less-noisy DMBCs. 
 
 \begin{Definition}[Strictly Essentially Less-Noisy]
 The definition of  a \emph{strictly essentially less-noisy} DMBC coincides with the definition of  an essentially less-noisy DMBC except that Inequality~\eqref{essenLN} needs to be strict whenever $I(U;Y_1)>0$. 
  \end{Definition}
The BSBC and the BEBC with different cross-over probabilities or different erasure probabilities at the two receivers are strictly essentially less-noisy.}

\section{Motivation: A Simple Scheme}\label{sec:simple}
{ We sketch a simple scheme that motivates  our work. Consider first the coding scheme in Subsection~\ref{sec:noff} without feedback, on which we build our coding scheme with feedback in Subsection~\ref{sec:dd}.
\subsection{A coding scheme without feedback}\label{sec:noff}
Assume each message  is split into $B$ submessages, $M_1=(M_{1,1},\ldots, M_{1,B})$ and $M_2=(M_{2,1}, \ldots, M_{2,B})$. We apply block-Markov coding with $B+1$ blocks of length~$n=\lfloor \frac{N}{B+1}\rfloor$, and in each block $b\in\{1,\ldots,B+1\}$ we use superposition coding  to send fresh messages $M_{1,b}$ and $M_{2,b}$. Message $M_{1,b}$ is sent in the cloud center $U_b^n$ and Message~$M_{2,b}$ in the satellite codeword $X^n_{b}$. Thus, the scheme is expected to perform well when the following gap is nonnegative:
\begin{equation}\label{eq:gap}
{\Gamma} := I(U;Y_2)-I(U;Y_1)\geq 0.
\end{equation}

After each block, both Receivers~1 and 2 decode the cloud center codeword $U_b^n(M_{1,b})$ by means of joint typicality decoding. By the Packing Lemma {\cite{book:gamal}}, Receiver~1 will be successful  with high probability whenever 
\begin{IEEEeqnarray}{rCl}
R_1 & < & I(U;Y_1)\label{eq:cons_mot1}
\end{IEEEeqnarray}
and Receiver~2 will be successful whenever
\begin{IEEEeqnarray}{rCl}\label{eq:cons_mot11}
R_1 & < & I(U;Y_2).
\end{IEEEeqnarray}
Receiver~2 also decodes the satellite codeword $X_b^n(M_{2,b}|M_{1,b})$, which is possible with very high probability whenever $R_2 < I(X;Y_2|U)$.

 We notice that when
\begin{equation}\label{eq:deltapos}
{\Gamma} >0,
\end{equation}
Constraint~\eqref{eq:cons_mot11} is not active in view of Constraint~\eqref{eq:cons_mot1}.  In this case,
Receiver~2 would be able to decode the cloud center even if it contained a second additional message of rate~${\Gamma}$.  The problem is that adding an arbitrary additional message to the cloud center, might make it impossible for Receiver~1  to decode since then \eqref{eq:cons_mot1} could be violated. 

We now show that when there is feedback from Receiver~1, the transmitter can identify a suitable additional message that it can add to the cloud center and that will improve the performance of the scheme.

\subsection{Our coding scheme with feedback}\label{sec:dd}
Assume there is feedback  from Receiver~$1$, i.e., $R_{\FB,1}>0$; we ignore feedback from Receiver~2.
\begin{figure}[ht!]
\includegraphics[width=0.5\textwidth]{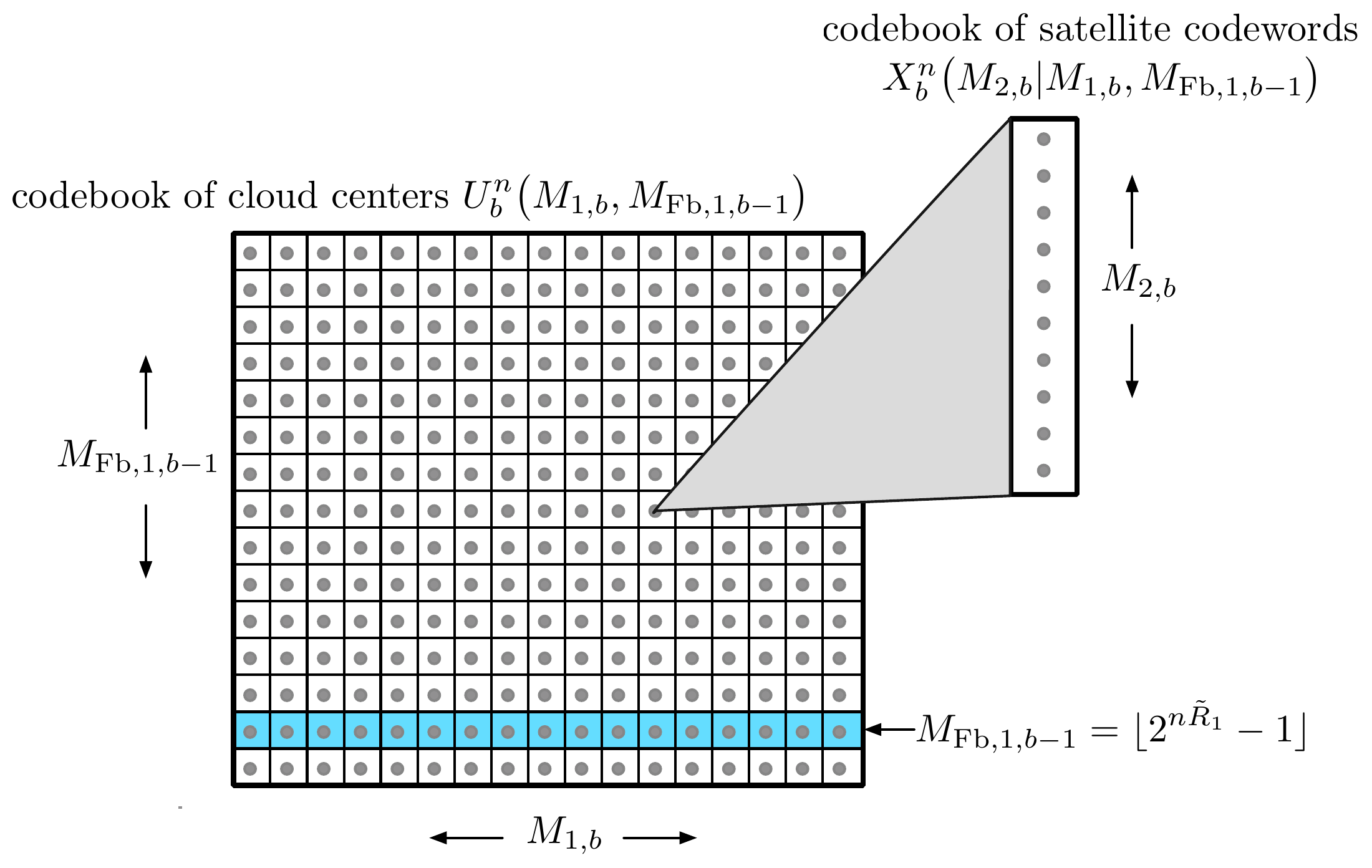}
\caption{Code construction in the simple motivating scheme for each block~$b$. Each dot represents a codeword.} 
\label{fig:superpos}
\end{figure}
As in the previous subsection, we employ a block-Markov strategy. In each block we use superposition coding as shown in Figure~\ref{fig:superpos}: the block-$b$ cloud center $U_b^n$ encodes the two messages $M_{1,b}$  and $M_{\FB,1,b-1}$ and the only satellite encodes message $M_{2,b}$. Here, $M_{\FB,1,b-1}$ is a feedback message that Receiver~1 sent back at the end of block $b-1$. 
Specifically, Receiver~1 generates $M_{\FB,1,b-1}$ as a Wyner-Ziv message that compresses Receiver~1's block-$(b-1)$ outputs  $Y_{1,b-1}^n$ so that a decoder that has side-information $Y_{2,b-1}^n$ can reconstruct the compressed outputs $\tilde{Y}_{1,b-1}^n$ {of $Y_{1,b-1}^n$}. Message~$M_{\FB,1,b-1}$'s rate 
 $\tilde{R}_1$ thus has to satisfy
 \begin{subequations}\label{eq:constraints_simple}
 \begin{equation}\label{eq:cons_mot2}
 \tilde{R}_1 > I(\tilde{Y}_1; Y_1|Y_2).
 \end{equation}
Also, since it is sent over the feedback link, it has to satisfy
 \begin{equation}\label{eq:cons_mot3}
 \tilde{R}_1 < R_{\FB,1},
 \end{equation}
 and for ease of exposition, we further restrict
  \begin{equation}\label{eq:cons_mot30}
 \tilde{R}_1 < \Gamma.
 \end{equation}
\end{subequations}

Decoding is performed as follows. After each block $b$, Receiver~1 decodes the cloud center. Since it is already aware of message $M_{\FB,1,b-1}$ (it created it itself), in the decoding it can restrict attention to all cloud-center codewords that correspond to the correct value of $M_{\FB,1,b-1}$. (In the code construction in Figure~\ref{fig:superpos}, it restricts to a specific row of the cloud center codebook. For example, when $M_{\FB,1,b-1}=\lfloor 2^{n\tilde{R}_1}\rfloor-1$ it restricts to the light blue row.) Thus, for Receiver~1 the situation is as if message $M_{\FB,1,b-1}$ was not present and had not been sent at all.

After each block~$b$, Receiver~2 performs the following three decoding steps:
\begin{itemize}
\item Based on $Y_{2,b}^n$ it decodes both messages $M_{1,b}$ and $M_{\FB,1,b-1}$ in the cloud center.
\item It uses the Wyner-Ziv message $M_{\FB,1,b-1}$ and its block-$(b-1)$ outputs $Y_{2,b-1}^n$ to reconstruct $\tilde{Y}_{1,b-1}^n$, the compressed version of $Y_{1,b-1}^n$.
\item  Based on the triple $(\tilde{Y}_{1,b-1}^n, Y_{2,b-1}^n, U^n_{b-1})$ it decodes its intended message $M_{2,b-1}$ sent in the satellite $X_{b-1}^n(M_{2,b-1}|M_{1,b-1}, M_{\FB,1,b-2})$ of the previous block. 
\end{itemize}

Receiver~1 errs with vanishingly small probability of error if
 \begin{subequations}\label{eq:simple}
\begin{IEEEeqnarray}{rCl}\label{eq:aa}
R_1 & < & I(U;Y_1).
\end{IEEEeqnarray}
Receiver~2 errs with vanishingly small probability of error if 
\begin{IEEEeqnarray}{rCl}
R_2 & < & I(X;\tilde{Y}_1,Y_2|U)= I(X;Y_2|U)+ I(X;\tilde{Y}_1|U,Y_2), \IEEEeqnarraynumspace
\end{IEEEeqnarray}
and if
\begin{IEEEeqnarray*}{rCl}
R_1+\tilde{R}_1 < I(U;Y_2)=I(U;Y_1)+\Gamma,
\end{IEEEeqnarray*}
{which is already implied by \eqref{eq:cons_mot30} and \eqref{eq:aa}.}
\end{subequations}

We conclude by~\eqref{eq:constraints_simple} and \eqref{eq:simple} that the {error} probability of our scheme tends to 0 as the blocklength $n$ and the number of blocks $B$ tend to infinity, whenever the rate pair $(R_1,R_2)$ satisfies \eqref{eq:simple} 
for some pmfs $P_{UX}$ and $P_{\tilde{Y}_1|Y_1}$ that satisfy 
\begin{equation}\label{eq:con}
I(\tilde{Y}_1; Y_1|Y_2) < \min\{{\Gamma}, R_{\FB,1}\}.
\end{equation}

Our new constraints~\eqref{eq:simple}  differ from the original superposition coding constraints \eqref{SPcoding} mainly in that the output $Y_2$ can be replaced by the pair $(Y_2, \tilde{Y}_1)$. This is because in our scheme the compressed output $\tilde{Y}_1$ is conveyed to Receiver~2.
Remarkably, there is no cost in conveying this compressed output $\tilde{Y}_1$ to Receiver~2: the compression information for Receiver~2 can be freely \emph{piggybacked} on the data sent to Receiver~1. Our new scheme thus improves over the standard superposition coding scheme whenever $\tilde{Y}_1$ is useful at Receiver~2, i.e.,  whenever  $I(X;\tilde{Y}_1|U,Y_2)>0$.}

%

 In the presented scheme, the transmitter simply relays the information it received over the feedback link to the other receiver. In this sense, the feedback link and part of the cloud center allow to establish a communication link from  Receiver~1 to Receiver~2, where the link is of rate
\begin{equation}
\min \{ {\Gamma}, R_{\FB,1}\}.
\end{equation}
We use this link to describe the compressed output $\tilde{Y}_1$  to Receiver~2.  
 
 \subsection{Extensions}
{For ease of exposition, we kept above coding scheme as simple as possible. It is easily extended in the following directions:
\begin{itemize}
\item The  block-$b$ Wyner-Ziv code that compresses $Y_{1,b}^n$  can be {superposed} on the cloud center $U_b^n$, {since Receiver 1  has already decoded this cloud center before generating (or using) the Wyner-Ziv message $M_{\textnormal{FB,1,b}}$}.
\item When $R_{\FB,2}>0$, also Receiver~2 can send a  Wyner-Ziv compression messages over the feedback link; now \emph{two} additional messages $M_{\textnormal{Fb},1,b-1}$ and  $M_{\textnormal{Fb},2,b-1}$ have to be included  in the block-$b$ cloud center.
\item Superposition coding can be replaced by full Marton coding.
\item The receivers can decode the cloud center and their satellites \emph{jointly} 
based on their own outputs and the compressed version of the other receiver's outputs. 
\item Sliding-window decoding at the receivers can be replaced by backward decoding.
\end{itemize}
Each of these modifications can only improve our scheme. However, there is a tension between the first modification and  {the last two modifications}. If a receiver uses  backward decoding instead of sliding-window decoding, it can't superpose its Wyner-Ziv code on the cloud center. This is simply because at the time it has to generate its Wyner-Ziv mesage, it hasn't yet decoded the cloud center. The same applies under sliding-window decoding if the receiver decodes the cloud center and the satellite \emph{jointly}.

For each receiver, one has thus to decide on the following two options:
\begin{itemize}
\item[1)] use successive sliding-window-decoding of the cloud center and the satellite of the Marton code, and  superpose the Wyner-Ziv code on the Marton cloud center; or 
\item[2)] use joint backward-decoding of the cloud center and the satellite of the Marton code, and \emph{do not} {superpose the Wyner-Ziv code on the cloud center}.
\end{itemize}
 It is unclear which of the two options performs better. In Section~\ref{sec:allschm} we present all three possible combinations: both receivers apply Option~1 (Scheme~IA); both receivers apply Option~2 (Scheme~IB); one receiver applies Option~1 and the other Option~2 (Scheme~IC). The corresponding three achievable regions are given in  Theorems~\ref{theosw}--\ref{Thm:oneside} in the next-following Section~\ref{Sec:NewRegions}.

}

\section{New Achievable regions}\label{Sec:NewRegions}

{The following achievable regions are based on the coding schemes in Section~\ref{sec:allschm}. 

Our first region is achieved by our scheme~IA  described in Section~\ref{sec:schemeb}, where both receivers apply successive sliding-window-decoding of the cloud center and the satellite of the Marton code, and  they {superpose} the Wyner-Ziv code on the Marton cloud center. 

\begin{Theorem}[Sliding-Window Decoding]\label{theosw}
The capacity region $\set{C}_{\textnormal{Fb}}$ includes the set $\set{R}_\textnormal{relay,sw}${\footnote{{The subscript ``relay" indicates that the transmitter simply relays the feedback information and the subscript ``sw" indicates that  sliding-window decoding is applied.}}} of all nonnegative rate pairs $(R_1,R_2)$ that satisfy
\begin{subequations}\label{eq:region_relayb}
\begin{IEEEeqnarray}{rCl}
R_1 \!&\leq& \!I(U_0,U_1;Y_1,\tilde{Y}_2|Q)-I(\tilde{Y}_2;U_0,Y_2|Y_1,Q)\quad~~\\
R_2 \!&\leq& \!I(U_0,U_2;Y_2,\tilde{Y}_1|Q)-I(\tilde{Y}_1;U_0,Y_1|Y_2,Q)\quad ~~\\
R_1 \!&\leq&\! I(U_0;Y_2|Q)+\Delta_2 -I(\tilde{Y}_1;Y_1|U_0,U_2,Y_2,Q)\qquad\\
R_2\! &\leq&\! I(U_0;Y_1|Q)+\Delta_1 -I(\tilde{Y}_2;Y_2|U_0,U_1,Y_1,Q)\\
R_1\!+\!R_2 \! &\leq& I(U_0,U_1;Y_1,\tilde{Y}_2|Q) +\Delta_1\nonumber\\
&&-I(\tilde{Y}_2;U_0,Y_2| Y_1,Q) -I(U_1;U_2|U_0,{Q})\label{In2bR2}\\
R_1\!+\!R_2 \!&\leq& \!I(U_0,U_2;Y_2,\tilde{Y}_1|Q)+\Delta_2 \nonumber \\
& &-I(\tilde{Y}_1;U_0,Y_1| Y_2,Q)-I(U_1;U_2|U_0,{Q}) \label{In2bSum1}\\
R_1\!+\!R_2 \!&\leq& \! I(U_0,U_1;Y_1,\tilde{Y}_2|Q)+ I(U_0,U_2;Y_2,\tilde{Y}_1|Q)\nonumber \\ 
& &  -I(\tilde{Y}_2;U_0,Y_2|Y_1,Q)-I(\tilde{Y}_1;U_0,Y_1|Y_2,Q)\quad\nonumber\\
& &  - I(U_1;U_2|U_0,{Q})  \label{In2bSum3}
\end{IEEEeqnarray}
\end{subequations}
where
\begin{IEEEeqnarray*}{rCl}
\Delta_1&:=&\min\{I(U_2;Y_2,\tilde{Y}_1|U_0,Q), \nonumber\\&&\quad I(U_2;Y_2,\tilde{Y}_1|U_0,Q)-I(\tilde{Y}_1;Y_1|U_0,Y_2,Q)+R_{\FB,1}\}\nonumber\\
\Delta_2&:=&\min\{I(U_1;Y_1,\tilde{Y}_2|U_0,Q),\nonumber\\&&\quad I(U_1;Y_1,\tilde{Y}_2|U_0,Q) -I(\tilde{Y}_2;Y_2|U_0,Y_1,Q)+R_{\FB,2}\}\nonumber
\end{IEEEeqnarray*}
for some pmf $P_{Q}P_{U_0U_1U_2|Q}P_{\tilde{Y}_1|Y_1U_0Q}P_{\tilde{Y}_2|Y_2U_0Q}$ and some function $f\colon  \mathcal{U}_0\times\mathcal{U}_1\times \mathcal{U}_2\times \mathcal{Q} \to \mathcal{X}$  such that $\Delta_1,\Delta_2\geq 0$ and \begin{subequations}\label{eq:Thm1condition}
\begin{IEEEeqnarray}{rCl}
I(\tilde{Y}_1;Y_1|U_0,U_2,Y_2,Q)\!&\leq& \!  \min\{I(U_0;Y_2|Q), R_{\textnormal{Fb},1}\}\quad\quad~\\
I(\tilde{Y}_2;Y_2|U_0,U_1,Y_1,Q)\! &\leq&  \! \min\{I(U_0;Y_1|Q), R_{\textnormal{Fb},2}\}\quad\quad~
  \end{IEEEeqnarray}
  \end{subequations}
where $X= f(U_0,U_1,U_2,Q)$.
\end{Theorem}
\begin{IEEEproof}
See Section \ref{sec:schemeb}.
\end{IEEEproof}
For  {sufficiently large feedback rates  $R_{\textnormal{Fb},1}$ and $R_{\textnormal{Fb},2}$}  (in particular for $R_{\textnormal{Fb},1}\geq {\log_2} |\mathcal{Y}_1|$ and $R_{\textnormal{Fb},2}\geq {\log_2} |\mathcal{Y}_2|$), we have \begin{subequations} \label{eq:DeltaPfb}\begin{IEEEeqnarray}{rCl}
\Delta_1&=&I(U_2;Y_2,\tilde{Y}_1|U_0,Q)\\
\Delta_2&=&I(U_1;Y_1,\tilde{Y}_2|U_0,Q)
\end{IEEEeqnarray}
\end{subequations}

The second region is based on our Scheme~IB  described in Section~\ref{sec:backwarddecoding}, where the receivers {apply joint backward-decoding of the cloud center and the satellite of the Marton code, but they \emph{do not} {superpose the Wyner-Ziv code on the cloud center.}}

 \begin{Theorem}[Backward Decoding]\label{theobw}
The capacity region $\set{C}_{\textnormal{Fb}}$ includes the set $\set{R}_\textnormal{relay,bw}${\footnote{{The subscript ``bw" stands for backward decoding.}}}  of all nonnegative rate pairs $(R_1,R_2)$ that satisfy
\begin{subequations}\label{eq:region_relay}
\begin{IEEEeqnarray}{rCl}
R_1 &\leq& I(U_0,U_1;Y_1,\tilde{Y}_2|Q)-I(\tilde{Y}_2;Y_2|Y_1,Q)\qquad\label{In2R1}\\
R_2 &\leq& I(U_0,U_2;Y_2,\tilde{Y}_1|Q)-I(\tilde{Y}_1;Y_1|Y_2,Q)\qquad  \label{In2R2}\\
R_1+R_2 &\leq& I(U_0,U_1;Y_1,\tilde{Y}_2|Q) + \Delta_1\nonumber \\
& &  -I(\tilde{Y}_2;Y_2| Y_1,Q)\!-\!I(U_1;U_2|U_0,\!{Q})\label{In2Sum1}\\
R_1+R_2 &\leq& I(U_0,U_2;Y_2,\tilde{Y}_1|Q)+ \Delta_2\nonumber \\ 
& &-I(\tilde{Y}_1;Y_1|Y_2,Q)\!-\! I(U_1;U_2|U_0,\!{Q})\label{In2Sum2}\IEEEeqnarraynumspace\\
R_1+R_2 &\leq&  I(U_0,U_1;Y_1,\tilde{Y}_2|Q) -I(\tilde{Y}_2;Y_2|Y_1,Q) \nonumber \\ 
& & +I(U_0,U_2;Y_2,\tilde{Y}_1|Q)-I(\tilde{Y}_1;Y_1|Y_2,Q)\nonumber \\ 
& & - I(U_1;U_2|U_0,{Q}) \label{In2Sum3}
\end{IEEEeqnarray}
\end{subequations}
for some pmf $P_{Q}P_{U_0U_1U_2|Q}P_{\tilde{Y}_1|Y_1Q}P_{\tilde{Y}_2|Y_2Q}$ and some function $f\colon  \mathcal{U}_0\times\mathcal{U}_1\times \mathcal{U}_2\times \mathcal{Q} \to \mathcal{X}$  such that 
\begin{subequations}\label{eq:Thm2condition}
\begin{IEEEeqnarray}{rCl}
I(\tilde{Y}_1;Y_1|U_0,U_2,Y_2,Q)&\leq&    R_{\textnormal{Fb},1}\\
I(\tilde{Y}_2;Y_2|U_0,U_1,Y_1,Q)&\leq&    R_{\textnormal{Fb},2}
  \end{IEEEeqnarray}
  \end{subequations}
where $X= f(U_0,U_1,U_2,Q)$.
\end{Theorem}
\begin{IEEEproof}
See Section~\ref{sec:backwarddecoding}.
\end{IEEEproof}

Setting  $\tilde{Y}_1=\tilde{Y}_2=\textnormal{const.}$, i.e.,  both receivers do not send any feedback, the region $\set{R}_\textnormal{relay,bw}$ specializes to $\set{R}_{\textnormal{Marton}}$.

\begin{Remark}
{Constraints~\eqref{eq:region_relay} and \eqref{eq:Thm2condition} are looser than Constraints~\eqref{eq:region_relayb} and \eqref{eq:Thm1condition}, respectively}. But in Theorem~\ref{theobw} we have the conditional pmfs $P_{\tilde{Y}_1|Y_1}$ and $P_{\tilde{Y}_2|Y_2}$ whereas in Theorem~\ref{theosw} we allow for more general pmfs $P_{\tilde{Y}_1|Y_1U_0}$ and $P_{\tilde{Y}_2|Y_2U_0}$.  It is thus not clear in general which of the achievable regions in Theorems ~\ref{theosw} or~\ref{theobw} is larger.
\end{Remark}

{The third region is based on our Scheme~IC (Section~\ref{schm:1c}), where Receiver~1 applies successive sliding-window-decoding of the cloud center and the satellite of the Marton code, and  superposes the Wyner-Ziv code on the Marton cloud center, and Receiver~2 applies 
joint backward-decoding of the cloud center and the satellite of the Marton code, but \emph{does not} superpose the Wyner-Ziv code.

The scheme is particularly interesting when there is no feedback from Receiver~2, $R_{\FB,2}=0$, and when Marton's scheme specializes to superposition coding with no satellite codeword for Receiver~1. The rate region corresponding to this special case is presented in Corollary~\ref{Cor:thm3} ahead.} 

\begin{Theorem}[Hybrid Sliding-Window Decoding and Backward Decoding] \label{Thm:oneside}
{For} $R_{\FB,2}=0$, the capacity region $\set{C}_{\textnormal{Fb}}$ includes the set $\set{R}_\textnormal{relay,hb}^{(1)}${\footnote{{The subscript ``hb" stands for hybrid decoding.}}} of all nonnegative rate pairs $(R_1,R_2)$ that satisfy
\begin{subequations}\label{eq:region_relays}
\begin{IEEEeqnarray}{rCl}
R_1 &\leq& I(U_0,U_1;Y_1|Q)\label{eq:thm31a}\\
R_2 &\leq &{ I(U_0,U_2;\tilde{Y}_1,Y_2|Q)}\nonumber\\
&& {- I(\tilde{Y}_1;U_0,U_1,U_2,Y_1| {Y}_2, Q)}\label{eq:thm31b}\\
R_1+R_2 &\leq& I(U_0,U_1;Y_1|Q)+ \Delta_1\nonumber\\
&&  -I(U_1;U_2|U_0,Q)\label{eq:thm31c}\\
R_1+R_2 &\leq& {I(U_1;Y_1|U_0,Q)+  I(U_0,U_2;\tilde{Y}_1,Y_2|Q)}\nonumber\\
&& - I(\tilde{Y}_1;U_0,U_1,U_2,Y_1| {Y}_2, Q)\nonumber\\ &&
 -I(U_1;U_2|U_0,Q)\label{eq:thm31d}
\IEEEeqnarraynumspace
\end{IEEEeqnarray}
\end{subequations}
for some pmf $P_{Q}P_{U_0U_1U_2|Q}P_{\tilde{Y}_1|Y_1U_0Q}$ and some function $f\colon  \mathcal{U}_0\times\mathcal{U}_1\times \mathcal{U}_2\times \mathcal{Q} \to \mathcal{X}$  such that 
\begin{IEEEeqnarray}{rCl}\label{eq:fbRateThm3}
I(\tilde{Y}_1;U_1,\!Y_1|U_0,U_2, Y_2,Q)&\leq&    R_{\textnormal{Fb},1}.
  \end{IEEEeqnarray}

The capacity region  $\set{C}_{\textnormal{Fb}}$  also includes the region $\set{R}_\textnormal{relay,hb}^{(2)}$ which is obtained by exchanging indices $1$ and $2$  in the above definition of $\set{R}_\textnormal{relay,hb}^{(1)}$.
\end{Theorem}
\begin{IEEEproof}
See Section~\ref{schm:1c}.
\end{IEEEproof}

{If superposition coding is used instead of Marton coding, Theorem \ref{Thm:oneside} reduces to the following corollary.}

\begin{Corollary}  \label{Cor:thm3}
The capacity region $\set{C}_{\textnormal{Fb}}$ includes the set $\set{R}_\textnormal{relay,sp}^{(1)}${\footnote{{The subscript ``sp" stands for superposition coding.}}}  of all nonnegative rate pairs $(R_1,R_2)$ that satisfy
\begin{subequations}\label{Coreq:region_relays}
\begin{IEEEeqnarray}{rCl}
R_1 &\leq& I(U;Y_1|Q)\label{eq:coro1a}\\
R_1\!+\!R_2 &\leq& I(U;Y_1|Q)+ I(X;Y_2,\tilde{Y}_1|U, \!Q)\label{eq:coro1b}\\
R_1\!+\!R_2 &\leq& I(X;Y_2|Q)- I(\tilde{Y}_1;Y_1|U,Y_2, Q)\label{eq:coro1c}
\IEEEeqnarraynumspace
\end{IEEEeqnarray}
\end{subequations}
for some pmf $P_{Q}P_{UX|Q}P_{\tilde{Y}_1|Y_1UQ}$ such that 
\begin{IEEEeqnarray}{rCl}\label{eq:fbRatecoro1}
I(\tilde{Y}_1;Y_1|U,Y_2,Q)&\leq&    R_{\textnormal{Fb},1}.
  \end{IEEEeqnarray}

The capacity region  $\set{C}_{\textnormal{Fb}}$  also includes the region $\set{R}_\textnormal{relay,sp}^{(2)}$ which is obtained by exchanging indices $1$ and $2$  in the above definition of $\set{R}_\textnormal{relay,sp}^{(1)}$.
\end{Corollary}
\begin{IEEEproof}
Let $\tilde{Y}_2=U_1=\textnormal{const.}$, $U=U_0$ and $X=U_2$. Constraint \eqref{eq:thm31a} then specializes to \eqref{eq:coro1a} and Constraint \eqref{eq:thm31b} is redundant compared to Constraint \eqref{eq:thm31d}. Observe that Constraints   \eqref{eq:thm31d} and \eqref{eq:fbRateThm3} are looser than Constraints  \eqref{eq:coro1c} and \eqref{eq:fbRatecoro1}, respectively. Also, by \eqref{eq:fbRatecoro1}, Constraint \eqref{eq:thm31c} reduces to \eqref{eq:coro1b}. Thus the capacity region  $\set{C}_{\textnormal{Fb}}$  includes the region $\set{R}_\textnormal{relay,sp}^{(1)}$. Similar arguments hold for  $\set{R}_\textnormal{relay,sp}^{(2)}$.
\end{IEEEproof}

\begin{Remark}
The region $\set{R}_{\textnormal{relay,hb}}^{(1)}$ contains the regions in Theorems~ \ref{theosw} and \ref{theobw} when these latter are specialized to $U_1=\textnormal{const.}$, $U_2=X$, and $R_{\FB,2}=0$. 
\end{Remark}

{In our Schemes~IA--IC the transmitter simply relays the compression information it received over each of the feedback links to the other receiver, as is the case also for our motivating scheme in the previous section~\ref{sec:simple}. 

Alternatively, the transmitter can  use this feedback information to first reconstruct the compressed versions of the channel outputs  and then  compress them jointly  with  the Marton codewords. The  indices resulting from this latter compression are then sent to the two receivers. {The following Theorem \ref{theo5} presents the rate region achieved by this Scheme~II.}}
 \begin{Theorem}\label{theo5}
The capacity region $\set{C}_{\textnormal{Fb}}$ includes the set $\set{R}_\textnormal{proc.}${\footnote{{The subscript ``proc." indicates that the transmitter processes the feedback information it receives.}}} of all nonnegative rate pairs $(R_1,R_2)$ that satisfy
 \begin{subequations}\label{eq:innerf}
\begin{IEEEeqnarray}{rCl}
R_{1}
&\leq& I(U_0,U_1;Y_1,\tilde{Y}_1,V|Q)\nonumber\\&&
-I(V;U_0,U_1,U_2, \tilde{Y}_2|\tilde{Y}_1,Y_1,Q) \\
R_2 &\leq& I(U_0,U_2;Y_2,\tilde{Y}_2,V|Q)\nonumber\\
&&-I(V;U_0,U_1,U_2,\tilde{Y}_1|\tilde{Y}_2,Y_2,Q) \\
R_1 +R_2 & \leq &I(U_0,U_1;Y_1,\tilde{Y}_1,V|Q)-I(U_1;U_2|U_0,{Q}) \nonumber \\ &&-I(V;U_0,U_1,U_2, \tilde{Y}_2|\tilde{Y}_1,Y_1,Q) \nonumber\\ &&+ I(U_2;Y_2,\tilde{Y}_2,V|U_0,Q)\\
R_1+R_2&\leq& I(U_0,U_2;Y_2,\tilde{Y}_2,V|Q)-I(U_1;U_2|U_0,{Q})\nonumber \\ &&-I(V;U_0,U_1,U_2,\tilde{Y}_1|\tilde{Y}_2,Y_2,Q) \nonumber \\
&& +I(U_1;Y_1,\tilde{Y}_1,V|U_0,Q) \\
R_1 +R_2 & \leq &I(U_0,U_1;Y_1,\tilde{Y}_1,V|Q)-I(U_1;U_2|U_0,{Q})\nonumber \\ &&-I(V;U_0,U_1,U_2,\tilde{Y}_2|\tilde{Y}_1,Y_1,Q) \nonumber \\&&-I(V;U_0,U_1,U_2, \tilde{Y}_1|\tilde{Y}_2,Y_2,Q)\nonumber \\
&&+   I(U_0,U_2;Y_2,\tilde{Y}_2,V|Q)
\end{IEEEeqnarray}
\end{subequations}
for some pmf $P_{Q}P_{U_0U_1U_2|Q}P_{\tilde{Y}_1|Y_1Q}P_{\tilde{Y}_2|Y_2Q}P_{V|U_0U_1U_2\tilde{Y}_1\tilde{Y}_2Q}$ and some function $f\colon \mathcal{X} \to \mathcal{U}_0\times\mathcal{U}_1\times \mathcal{U}_2\times \mathcal{Q}$  
where the feedback-rates have to satisfy
\begin{subequations}\label{eq:feedback_rates}
\begin{IEEEeqnarray}{rCl}
I(Y_1;\tilde{Y}_1|U_0,U_1,U_2,\tilde{Y}_2,Q) & \leq & R_{\FB,1}\\
I(Y_2;\tilde{Y}_2|U_0,U_1,U_2,\tilde{Y}_1,Q)  & \leq & R_{\FB,2}\\
I(Y_1,Y_2;\tilde{Y}_1,\tilde{Y}_2|U_0,U_1,U_2,Q)  & \leq & R_{\FB,1}\!+\!R_{\FB,2}\quad 
\end{IEEEeqnarray}
\end{subequations}
and where $X= f(U_0,U_1,U_2,Q)$.
\end{Theorem}
\begin{IEEEproof}
See Section~\ref{sec:proc}.
\end{IEEEproof}

{When $R_{\FB,1}, R_{\FB,2}$ are sufficiently large, so that we can choose $\tilde{Y}_1=Y_1$ and $\tilde{Y}_2=Y_2$, we have the following corollary to Theorem~\ref{theo5}.

\begin{Corollary} \label{Re:KimChia}
Let  $\set{R}_\textnormal{proc.}^{\infty}$ be the set of all nonnegative rate pairs $(R_1,R_2)$ that satisfy
 \begin{subequations}\label{eq:pri}
\begin{IEEEeqnarray}{rCl}
R_{1} &\leq& I(U_0,U_1;Y_1,V|Q)-I(V;U_0,U_1,U_2, {Y}_2|Y_1,Q) \nonumber\\\label{eq:inner1f}\\
R_2&\leq& I(U_0,U_2;Y_2,V|Q)-I(V;U_0,U_1,U_2,{Y}_1|Y_2,Q) \nonumber \\\label{eq:inner2f} \\
R_1 +R_2 & \leq &I(U_1;Y_1,V|U_0,Q)+I(U_2;Y_2,V|U_0,Q)\nonumber  \\ &&-I(U_1;U_2|U_0,{Q})+\min_{i\in\{1,2\}} \{ I(U_0;Y_i,V|Q)\nonumber\\
&&-I(V;U_0,U_1,U_2, Y_1,{Y}_2|Y_i,Q) \}\label{eq:corsum}\\
R_1 +R_2 & \leq &I(U_0,U_1;Y_1,V|Q)+   I(U_0,U_2;Y_2,V|Q) \nonumber \\ &&-I(V;U_0,U_1,U_2, {Y}_1|Y_2,Q)\nonumber \\&&-I(V;U_0,U_1,U_2,{Y}_2|Y_1,Q)\nonumber \\&& -I(U_1;U_2|U_0,{Q}) \label{eq:inner4f}
\end{IEEEeqnarray}
\end{subequations}
for some pmf $P_{Q}P_{U_0U_1U_2|Q} P_{V|U_0U_1U_2{Y}_1{Y}_2}$ and some function $f\colon \mathcal{X} \to \mathcal{U}_0\times\mathcal{U}_1\times \mathcal{U}_2\times \mathcal{Q}$, where $X= f(U_0,U_1,U_2,Q)$.

When $R_{\FB,1} \geq \log_2|\set{Y}_1|$ and $R_{\FB,{2}} \geq \log_2|\set{Y}_2|$,\footnote{Smaller feedback rates suffice in general; for simplicity we use these conditions on the feedback rates.}
  \begin{equation}
\set{R}_\textnormal{proc.}^{\infty}\in \set{C}_{\FB}.
\end{equation}
\end{Corollary}
{As we show in Subsection~\ref{sec:related1}, our new region $\set{R}_{\textnormal{proc.}}^{\infty}$ includes an important special case of the Shayevitz-Wigger region~\cite{wigger}.}
}

\section{Usefulness of Feedback}\label{Sec:Usefulness}
{Our Scheme~IC  (which leads to Theorem~\ref{Thm:oneside}) can be used to prove the following result on the usefulness of rate-limited feedback for DMBCs. (Similar results can be shown based on our other proposed schemes.)}

\begin{Theorem}\label{thm:useful}
Fix a DMBC. Consider random variables $(U_0^{\M},U_1^{\M}, U_2^{\M}, X^{\M})$ such that
\begin{equation}\label{eq:delta}
{\Gamma^{\M}} := I(U_0^{\M};Y_2^{\M})- I(U_0^{\M};Y_1^{\M} )>0.
\end{equation}
Let the rate pair $(R_1^{\M},R_2^\M)$  satisfy Marton's constraints~\eqref{eq:martonRegion} when evaluated for $(U_0^{\M},U_1^{\M}, U_2^{\M}, X^{\M})$ where Constraint~\eqref{eq:Marton2} has to hold with strict inequality. 

Also, let $(R_1^{\En}, R_2^{\En})$ be a rate pair in the capacity region  $\set{C}^{(1)}_{\textnormal{Enh}}$ of the enhanced DMBC.

If the feedback-rate from Receiver~1 is positive, $R_{\FB,1}>0$, then for all sufficiently small $\gamma\in(0,1)$, {the rate pair $(R_1,R_2)$,
 \begin{subequations}\label{eq:rate-pair}
\begin{IEEEeqnarray}{rCl}
R_1&=& (1-\gamma)R_1^{\M} + \gamma R_{1}^{\En}\\
R_2&=& (1-\gamma)R_2^{\M} + \gamma R_{2}^{\En}
\end{IEEEeqnarray}
 \end{subequations}
lies in $\set{R}_{\textnormal{relay,hb}}^{(1)}$,
\begin{equation}\label{eq:new}
(R_1,R_2)  \in \set{R}_{\textnormal{relay,hb}}^{(1)},
\end{equation} and is thus achievable.}

An analogous statement holds when indices $1$ and $2$ are exchanged.
\end{Theorem}
\begin{IEEEproof}
See Appendix~\ref{sec:prooftheorem}.
\end{IEEEproof}
{
The following remark elaborates on the condition of the theorem that a rate pair satisfies Constraint~\eqref{eq:Marton2} with strict inequality. {It will be used in the proof of Corollary~\ref{theo3}.}

\begin{Corollary}\label{cor:gen}
Assume $R_{\FB,1}>0$. If there exists a rate pair $(R_1^{\M}, R_2^{\M})$ that satisfies the conditions in Theorem~\ref{thm:useful} and that lies on the boundary of $\set{R}_{\textnormal{Marton}}$ but strictly in the interior of $\set{C}_{\textnormal{Enh}}^{(1)}$, then 
\begin{equation}\label{eq:first}
 \set{R}_{\textnormal{Marton}} \subsetneq \set{C}_{\textnormal{Fb}}.
\end{equation}

If for the considered DMBC moreover $\set{R}_{\textnormal{Marton}}=\set{C}_{\textnormal{NoFB}}$,  
\begin{equation}\label{eq:second}
 \set{C}_{\textnormal{NoFB}} \subsetneq \set{C}_{\textnormal{Fb}}.
\end{equation}
\end{Corollary}
\begin{proof}
Inclusion~\eqref{eq:second} follows from \eqref{eq:first}. We show \eqref{eq:first}. Since $(R_1^{\M}, R_2^{\M})$ is  in the interior of $\set{C}_{\textnormal{Enh}}^{(1)}$, there exists a rate pair $(R_1^{\En}, R_2^{\En})\in\set{C}_{\textnormal{Enh}}^{(1)}$ with $R_{1}^{\En} > R_1^{\M}$ and $R_2^{\En}> R_2^{\M}$. Now, since $(R_1^{\M}, R_2^{\M})$ lies on the boundary of $\set{R}_{\textnormal{Marton}}$, the rate pair in~\eqref{eq:rate-pair} must lie  outside $\set{R}_{\textnormal{Marton}}$ for any $\gamma\in(0,1)$.  By Theorem~\ref{thm:useful}, Equation~\eqref{eq:new}, this rate pair {is  achievable with rate-limited feedback }
for all $\gamma\in(0,1)$ that are sufficiently close to 0.
\end{proof}

For many DMBCs such as the BSBC or  the BEBC with unequal cross-over probabilities or unequal erasure probabilities to the two receivers, or the BSC/BEC-BC where the two channels have different capacities, it is easily verified that the conditions of Corollary~\ref{cor:gen} hold whenever the DMBCs are not physically degraded. Thus, our corollary immediately shows that for these DMBCs rate-limited feedback strictly increases capacity. (See also Examples~\ref{ex:bsbc} and \ref{ex:2} in the next Section.)

For the BSBC and the BEBC,  Theorem~\ref{thm:useful} can even be used to show that all boundary points $(R_1>0, R_2>0)$ of  $\set{C}_{\textnormal{NoFB}}$ can be improved with rate-limited feedback, see the following Corollary~\ref{theo3} and the paragraph thereafter.

\begin{Remark}\label{rem:conditions} 
{For given  random variables  $U_0^{\M}, U_1^{\M}, U_2^{\M}, X^{\M}$ Marton's region, i.e., the rate region defined by Constraints~\eqref{eq:martonRegion}, is either a pentagon (both single-rate constraints as well as at least one of the sum-rates are active), a quadrilateral (only the two single-rate constraints are active), or a triangle (only one single-rate constraint and at least one of the sum-rate constraints are active).  

In the case of superposition coding with $U_1^{\M}=\textnormal{const.}$ and $U_2^{\M}=X^{\M}$ and when Condition~\eqref{eq:delta} holds, then the region is a quadrilateral and the only active constraints are~\eqref{eq:Marton1} and \eqref{eq:Marton3}. Thus, in this case, constraint~\eqref{eq:Marton2} holds with strict inequality for all rate pairs in this region.

Whenever the region defined by Marton's constraints~\eqref{eq:martonRegion} is a pentagon, then the only rate pair in this pentagon that satisfies Constraint~\eqref{eq:Marton2} with equality is the dominant corner point of maximum $R_2$-rate.}
\end{Remark} 

  \begin{Corollary}\label{theo3} Consider a DMBC where $Y_2$ is strictly essentially less-noisy than $Y_1$. Assume $R_{\textnormal{Fb},1}>0$. We have: 
\begin{enumerate}
\item If a rate pair $(R_1,R_2)$ lies on the boundary of $\set{C}_\textnormal{NoFB}$ but in the interior of $\set{C}_\textnormal{Enh}^{(1)}$, then  $(R_1,R_2)$ lies in the interior of $\set{C}_{\textnormal{Fb}}$, i.e., {with rate-limited feedback one can improve over this rate pair}.
\item If  $\set{C}_\textnormal{NoFB}$ does not coincide with $ \set{C}_{\textnormal{Enh}}^{(1)}$, then $\set{C}_\textnormal{NoFB}$ is also a strict subset of $\set{C}_{\textnormal{Fb}}$, i.e., feedback strictly improves capacity.
\end{enumerate}
Analogous statements hold if indices $1$ and $2$ are exchanged. 
\end{Corollary}

\begin{IEEEproof}[Proof of Corollary~\ref{theo3}]
2.) follows from 1.) We prove 1.)
For strictly essentially less-noisy DMBCs,  $\set{C}_{\textnormal{NoFB}}$ is achieved by superposition coding. Thus,  $\set{R}_{\textnormal{Marton}}= \set{C}_{\textnormal{NoFB}}$ and in the evaluation of Marton's region one can restrict attention to auxiliaries of the form $U_1=\textnormal{const.}$ and $U_2=X$. By the definition of strictly essentially-less noisy, when evaluating Marton's region we can further restrict attention to auxiliary random variables that satisfy~\eqref{eq:delta}. Thus, by  Remark~\ref{rem:conditions}, any boundary point of $\set{R}_{\textnormal{Marton}}$ satisfies the conditions of Theorem~\ref{thm:useful}.  Repeating the proof steps for Corollary~\ref{cor:gen}, we can prove that these boundary points cannot be boundary points of $\set{C}_{\FB}$ whenever they lie in the interior of $\set{C}_{\textnormal{Enh}}^{(1)}$. 
\end{IEEEproof}

As mentioned, all BSBCs and BEBCs with  unequal cross-over probabilities or unequal erasure probabilities to the two receivers are strictly essentially less-noisy.  Also,  for these BCs $\set{C}_{\textnormal{NoFB}}$ has no common boundary points $(R_1>0, R_2>0)$ with the sets  $\set{C}_{\textnormal{Enh}}^{(1)}$ or $\set{C}_{\textnormal{Enh}}^{(2)}$  unless the BC is physically degraded. Thus, for these BCs the corollary implies that, unless the BC is physically degraded, rate-limited feedback improves all boundary points $(R_1>0,R_2>0)$ of $\set{C}_{\textnormal{NoFB}}$ whenever $R_{\FB,1},R_{\FB,2}>0$.

Notice that when a DMBC is physically degraded in the sense that output $Y_1$ is a degraded version of $Y_2$, then $\set{C}_{\textnormal{NoFB}}=\set{C}_{\textnormal{Enh}}^{(1)}$. In this case, (even infinite-rate) feedback does not increase the capacity of physically degraded DMBCs \cite{gamal'78}. 

{
Theorem~\ref{thm:useful}  exhibits sufficient conditions that our coding schemes improve over the nofeedback capacity. These conditions however are not necessary. In fact, by  our discussions in Section~\ref{Sec:Relations}  and the examples in \cite{greek,wang,kimchia}, our schemes can improve over the nofeedback capacity even for DMBCs where for any choice of the auxiliaries $I(U_0;Y_2)= I(U_0;Y_1)$ holds.

On a related note, Bracher and Wigger~\cite{annina} showed that our type-I schemes can improve over the nofeedback capacity even for semideterministic\footnote{Semideterministic means that the outputs at one of the two receivers are deterministic functions of the inputs.} DMBCs. This might be surprising since without feedback the capacity region of semideterministic DMBCs is achieved by degenerate Marton coding \emph{without cloud center}~\cite{semideterministic}, i.e., $U_0=\textnormal{const.}$.

\section{Examples}\label{sec:Examples}
\begin{Example}\label{ex:bsbc}
Consider the BSBC with input $X$ and outputs $Y_1$ and $Y_2$ described by:
\begin{subequations}
   \begin{IEEEeqnarray}{rCl}
   Y_1=X\oplus Z_1,\quad Y_2=X\oplus Z_2,
 \end{IEEEeqnarray}
 \end{subequations}
 for $Z_1\sim \text{Bern}(p_1)$ and $Z_2\sim \text{Bern}(p_2)$ independent noises. 
 Let 
 $Q=\textnormal{const.}$, $U\sim  \text{Bern}(1/2)$, $W_1\sim \text{Bern}(\beta_1)$ and $W_2\sim \text{Bern}(\beta_2)$, for $\beta_1,\beta_2\in[0,1/2]$, where $U,W_1,W_2$ are independent. Also set  $X=U\oplus W_1$, and $\tilde{Y}_1=  Y_1\oplus W_2$.  Then 
 \begin{IEEEeqnarray*}{rCl}
 I(U;Y_1)=1-H_b(\beta_1 *p_1),\quad
 I(X;Y_2)=1-H_b(p_2),
\end{IEEEeqnarray*}
 and  
 \begin{IEEEeqnarray*}{rCl}
  I(X;\tilde{Y}_1,Y_2|U)&=& H(\alpha_1,\!\alpha_2,\!\alpha_3,\!\alpha_4) \!-\!H_b(p_2)\!-\!H_b(\beta_2*p_1)\nonumber\\
 I(\tilde{Y}_1;Y_1|Y_2,U)&=&H(\alpha_1,\!\alpha_2,\!\alpha_3,\!\alpha_4)\!-\!H_b(\beta_1*{p_2})
\!-\!H_b(\beta_2)
\end{IEEEeqnarray*}
 where   \begin{IEEEeqnarray*}{rCl}
 \alpha_1&=&(p_1*\beta_2)p_2\beta_1 +(1-p_1*\beta_2)\bar{p_2}\bar{\beta_1}\nonumber\\
 \alpha_2&=&(p_1*\beta_2)\bar{p_2}\beta_1+(1-p_1*\beta_2)p_2\bar{\beta_1}\nonumber\\
 \alpha_3&=&(p_1*\beta_2)\bar{p_2}\bar{\beta_1}+(1-p_1*\beta_2)p_2\beta_1\nonumber\\
 \alpha_4&=&(p_1*\beta_2)p_2\bar{\beta_1}+(1-p_1*\beta_2)\bar{p_2}\beta_1. 
  \end{IEEEeqnarray*}
For this choice, the region defined by the constraints in~Corollary~\ref{Cor:thm3} evaluates to:
 \begin{subequations}\label{eq:regBS}
   \begin{IEEEeqnarray}{rCl}
   R_1&\leq&  1-H_b(\beta_1 *p_1)\\
  R_1+ R_2&\leq& 1-H_b(\beta_1 *p_1)+H(\alpha_1,\alpha_2,\alpha_3,\alpha_4) \nonumber\\
  &&\quad-H_b(p_2)-H_b(\beta_2*p_1) \\
   R_1+ R_2&\leq& 1-H_b(p_2)-H(\alpha_1,\alpha_2,\alpha_3,\alpha_4)\nonumber\\
  &&\quad+H_b(\beta_1*{p_2})+H_b(\beta_2)
 \end{IEEEeqnarray}
 \end{subequations}
 for some $\beta_1, \beta_2\in[0,1/2]$ satisfying
 \begin{IEEEeqnarray}{rCl}
 H(\alpha_1,\alpha_2,\alpha_3,\alpha_4)\!-\!H_b(\beta_1*{p_2})\!-\!H_b(\beta_2)\leq R_{\textnormal{Fb},1}\quad 
 \end{IEEEeqnarray}
  and where $H(\alpha_1,\alpha_2,\alpha_3,\alpha_4)$  denotes the entropy of a quaternary random variable with probability masses $(\alpha_1,\alpha_2,\alpha_3,\alpha_4)$. 
  
 The region is plotted in Figure~\ref{fig:rateBSBC} against the no-feedback capacity region $\set{C}_\textnormal{NoFB}$. 
  \end{Example}
  \begin{figure}[!t]
\centering
\includegraphics[width=0.5\textwidth]{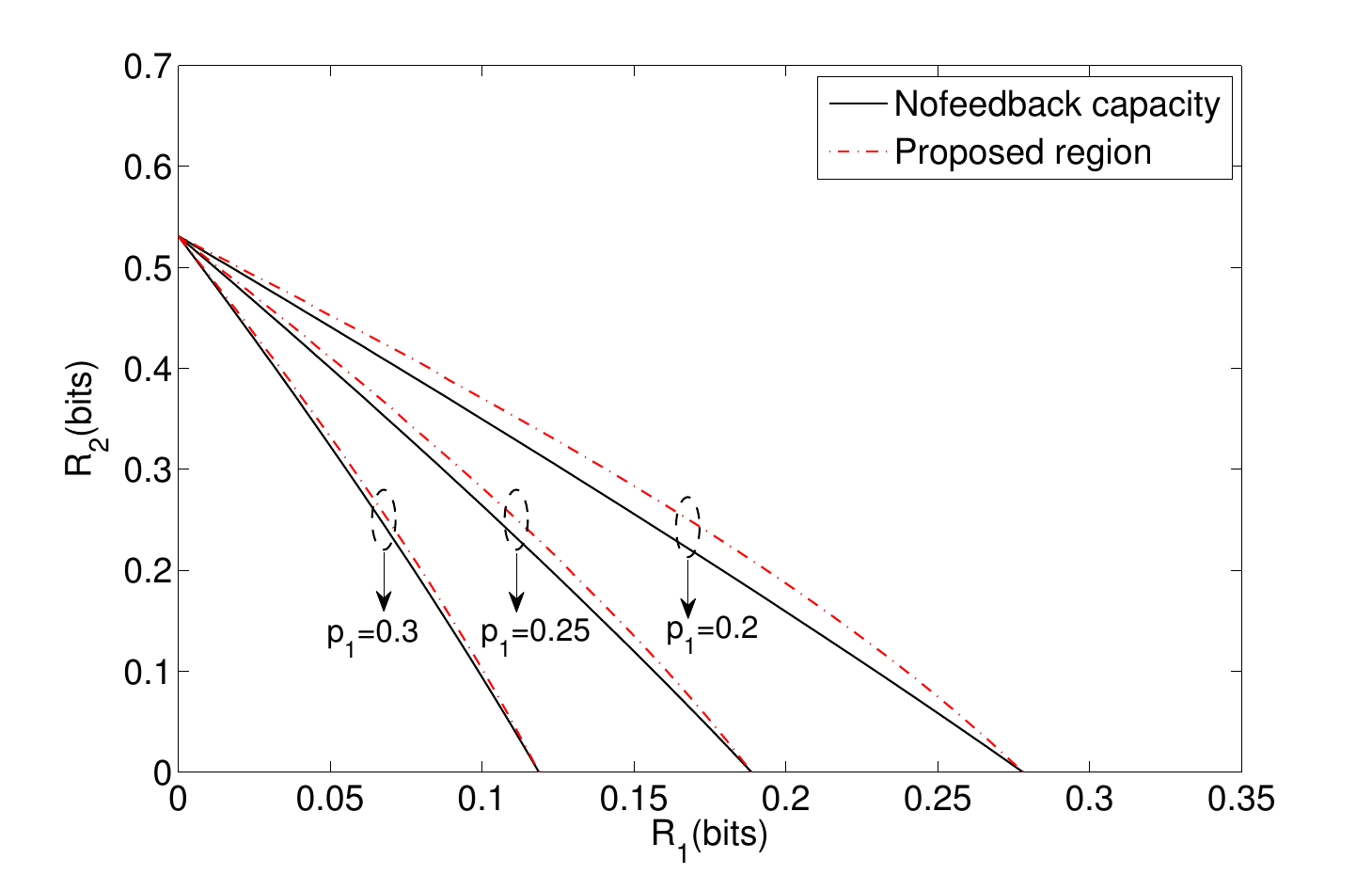}
\caption{$\set{C}_{\textnormal{NoFB}}$  and the achievable region in~\eqref{eq:regBS} are plotted for BSBCs with parameters $p_2=0.1$ and $p_1\in\{0.2,0.25,0.3\}$ and for  feedback rate $R_{\textnormal{Fb},1}=0.8$.} \label{fig:rateBSBC}
\vspace{-4mm}
\end{figure}

\begin{Example}\label{ex:2}
Consider a DMBC where the channel from $X$ to $Y_1$ is a BSC with cross-over probability $p\in(0,1/2)$, and the channel from $X$ to $Y_2$ is an independent BEC with erasure probability $e\in(0,1)$.
We show that our feedback regions $\set{R}_{\textnormal{relay,sp}}^{(1)}$ and $\set{R}_{\textnormal{relay,sp}}^{(2)}$ improve over a large part of the boundary points of $\set{C}_{\textnormal{NoFB}}$ for all values of $e, p$ unless $H_b(p)= e$, no matter how small $R_{\FB,1},R_{\FB,2}>0$.

We distinguish different parameter ranges of our channel. 
\begin{itemize}
\item  $0<e < H_b(p)$: In this case, the nofeedback capacity region $\set{C}_{\textnormal{NoFB}}$ \cite{CNair'10} is formed by the set of rate pairs $(R_1,R_2)$ that  for some $s\in[0,1/2]$ satisfy
\begin{subequations}
\begin{IEEEeqnarray}{rCl}
R_1&\leq &1- H_b(s * p),\label{eq:s1ep}\\
 R_2&\leq & (1-e)H_b(s), \label{eq:s2ep}\\
 R_1+R_2& \leq  & 1-e.\label{eq:sumep}
\end{IEEEeqnarray}
\end{subequations}
We specialize the region $\set{R}_\textnormal{relay,sp}^{(1)}$ to the following choices. Let $Q=\textnormal{const.}$, $U\sim \textnormal{Bern}(1/2)$,   $X=U\oplus  V$, where $V\sim \textnormal{Bern}(s)$ independent of $U$,  and  $\tilde{Y}_1= Y_1$ with probability $\gamma \in(0,1)$ and  $\tilde{Y}_1=\Delta$ with probability $1-\gamma$,
where 
\begin{equation}\label{eq:condgamma}
\gamma \leq   \frac{R_{\FB,1}}{(1-e)H_b(p)+eH_b(s*p)}.
\end{equation} 
Condition \eqref{eq:condgamma} assures that the described choice satisfies~\eqref{eq:fbRatecoro1}. Then,
\begin{IEEEeqnarray*}{rCl}
 I(U;Y_1)=1- H_b(s*p),\quad I(X;Y_2)=1-e,
\end{IEEEeqnarray*}
 and  
 \begin{IEEEeqnarray*}{rCl}
  I(X;\tilde{Y}_1,Y_2|U)&=& \gamma e\big(H_b(s*p)\!-\!H_b(p)\big)\nonumber\\
  &&+(1-e)H_b(s)\nonumber\\
 I(\tilde{Y}_1;Y_1|Y_2,U)&=&\gamma(H_b(p)(1-e)+eH_b(s*p)).
\end{IEEEeqnarray*}

When $R_{\FB,1}>0$, by Corollary~\ref{Cor:thm3}, all rate pairs $(R_1,R_2)$ satisfying
\begin{subequations}\label{eq:fbregion}
\begin{IEEEeqnarray}{rCl}
R_1 &\leq& 1- H_b(s*p)\\
R_1\!+\!R_2 &\leq&  1- H_b(s*p)+(1-e) H_b(s) \nonumber \\ 
&&+\gamma e ( H_b(s*p)-H_b(p))\\
R_1\!+\!R_2 &\leq&1\!-\!e\!-\!\gamma(H_b(p)(1\!-\!e)\!+\!eH_b(s\!*\!p)) \IEEEeqnarraynumspace
\end{IEEEeqnarray}
\end{subequations} 
are achievable for any $\gamma\in(0,1)$ satisfying~\eqref{eq:condgamma}.

As shown in \cite{CNair'10}, the points $(R_1,R_2)$ of the form
\begin{equation}\label{eq:boundary}
(1-H_b(s*p), (1-e)H_b(s)),\qquad s\in(0,s_0),
\end{equation} 
are all on the dominant boundary of $\set{C}_{\textnormal{NoFb}}$, 
where $s_0\in(0,1/2)$ is the unique solution to 
\begin{equation}
1-H_b(s_0*p)+(1-e)H_b(s_0)=1-e.
\end{equation} 
For these boundary points, only the single-rate constraints~\eqref{eq:s1ep} and \eqref{eq:s2ep} are active, but not~\eqref{eq:sumep}. Thus, comparing \eqref{eq:boundary} to our feedback region~\eqref{eq:fbregion}, we can conclude that by choosing $\gamma$ sufficiently small, all boundary points~\eqref{eq:boundary} lie strictly in the interior of our feedback region $\set{R}_{\textnormal{relay,sp}}^{(1)}$ when $R_{\FB,1}>0$.

\item $0<H_b(p)<e <1$: The nofeedback capacity region $\set{C}_{\textnormal{NoFb}}$ equals the time-sharing region given by the union of all rate pairs $(R_1,R_2)$ that for some $\alpha\in[0,1]$ satisfy
\begin{subequations}\label{eq:nofbregion2}
\begin{IEEEeqnarray}{rCl}
R_1 & \leq & \alpha (1-H_b(p))\\
R_2 & \leq & (1-\alpha) (1-e).
\end{IEEEeqnarray}
\end{subequations}

We specialize the region $\set{R}_{\textnormal{relay,sp}}^{(2)}$ to the following choices: $Q\sim\textnormal{Bern}(\alpha)$; if $Q=0$ then $U \sim\textnormal{Bern}(1/2)$,   $X=U$, and $\tilde{Y}_2=\textnormal{const.}$; if $Q=1$ then $U=\textnormal{const.}$, $X\sim\textnormal{Bern}(1/2)$, and $\tilde{Y}_2= Y_2$ with probability $\gamma\in(0,1)$ and $\tilde{Y}_2=\Delta$ with probability $1-\gamma$, where  in  order to satisfy the average feedback rate constraint,
\begin{equation}\label{eq:condgamma2}
\gamma \leq   \frac{R_{\FB,2}}{\alpha((1-e)H_b(p)+H_b(e))}.
\end{equation} 
When $R_{\FB,2}>0$, by Theorem~\ref{Thm:oneside}, all rate pairs $(R_1,R_2)$ satisfying\begin{subequations}\label{eq:yesfbregion2}
\begin{IEEEeqnarray}{rCl}
R_1 &\leq& \alpha (1- H_b(p)) +\alpha (1-e)\gamma H_b(p)\IEEEeqnarraynumspace\\
R_1+R_2 &\leq& (1-\alpha) (1-e) +\alpha (1- H_b(p)) \nonumber\\&&~+\alpha (1-e)\gamma H_b(p)\\
R_1+R_2 &\leq&   (1- H_b(p)) - (1-\alpha) \gamma H_b(e).
\end{IEEEeqnarray}
\end{subequations} 
are achievable for any $\gamma\in(0,1)$ satisfying~\eqref{eq:condgamma2}.

Since here $1-H_b(p)>1-e$,  for small $\gamma>0$ the feedback region in~\eqref{eq:yesfbregion2} improves over $\set{C}_{\textnormal{NoFB}}$ given in~\eqref{eq:nofbregion2}. In fact,~\eqref{eq:yesfbregion2} improves over all  boundary points $(R_1>0, R_2>0)$ of   $\set{C}_{\textnormal{NoFB}}$. \end{itemize}
\end{Example}
  \begin{figure}[!t]
\centering
\includegraphics[width=0.5\textwidth]{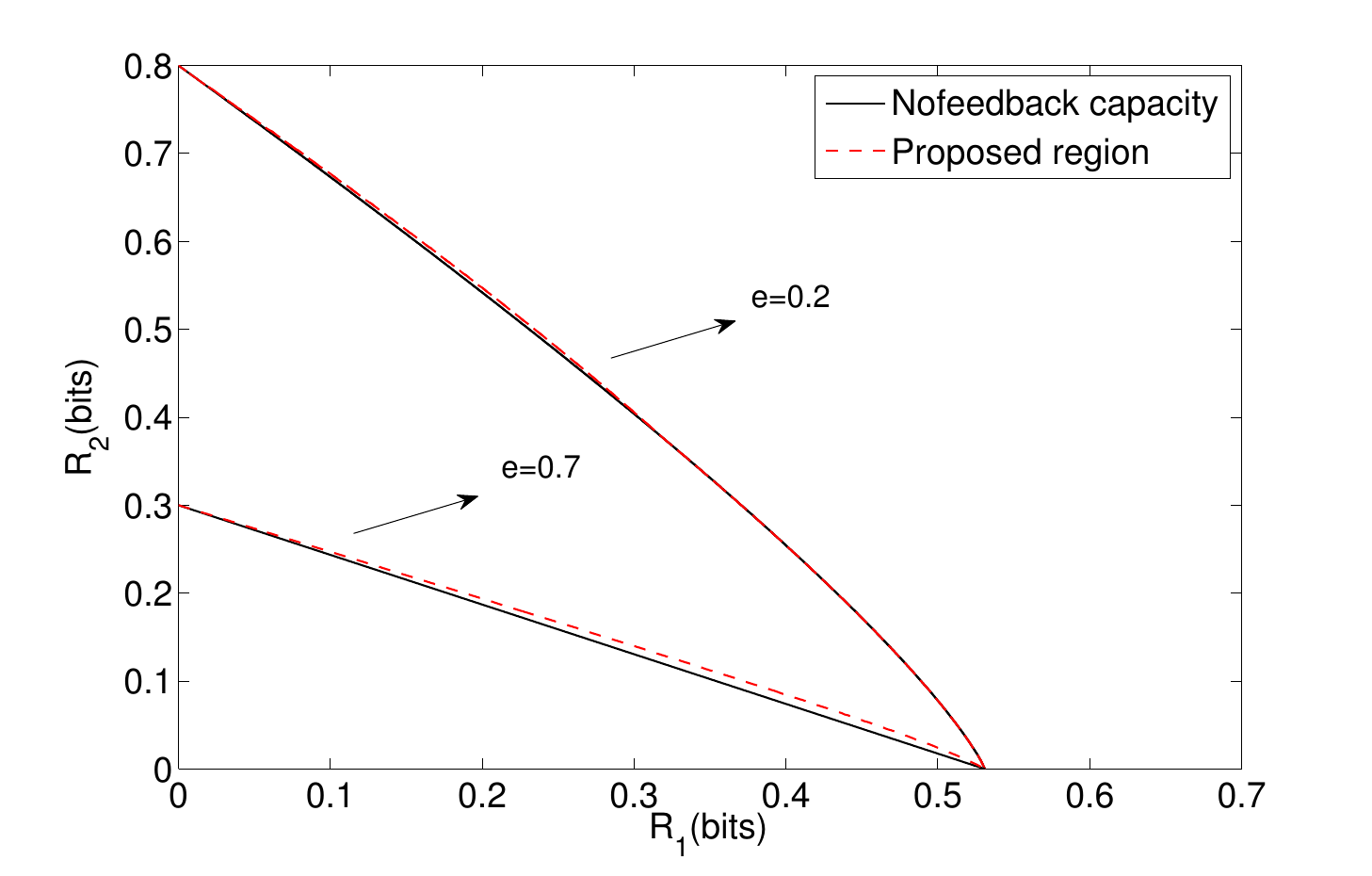}
\caption{$\set{C}_{\textnormal{NoFB}}$  and the achievable regions in~\eqref{eq:fbregion} and \eqref{eq:yesfbregion2} are plotted for a BSC/BEC-BC when the BSC has parameter $p=0.1$ and the BEC has parameter $e\in\{0.2,0.7\}$. Notice that $0.2 < H_b(p)<0.7$. The feedback rates $R_{\textnormal{Fb},1}=R_{\textnormal{Fb},2}=0.8$.} \label{fig:BSBEC}
\vspace{-4mm}
\end{figure}

\begin{Remark}The BSC/BEC-BC in Example~\ref{ex:2}, is particularly interesting, because depending on the values of the parameters $e$ and $p$, the BC is either degraded, less noisy, more capable, or  essentially less-noisy \cite{CNair'10}. We conclude that when Receiver~1 is ``stronger" than  Receiver~2,  $\mathcal{R}^{(2)}_{\textnormal {Relay,sp}}$ improves over the no-feedback capacity, for all these classes of BCs even with only one feedback link that is of arbitrary small, but positive rate. Similar arguments hold for  $\mathcal{R}_\textnormal{relay,sp}^{(1)}$ when Receiver 2 is ``stronger" than  Receiver 1. 

We plotted our regions~\eqref{eq:fbregion} and \eqref{eq:yesfbregion2} versus the nofeedback capacity region in Figure~\ref{fig:BSBEC} for $p=0.1$ and $e=0.2$ or $e=0.7$. In the first case the DMBC is more capable and in the second case it is essentially less-noisy.
\end{Remark}

{In the next example we consider the Gaussian BC with independent noises. We evaluate the region defined by the constraints of Corollary~\ref{Cor:thm3} for a set of jointly Gaussian distributions on the input and the auxiliary random variables. A rigorous proof that our achievability result in Corollary~\ref{Cor:thm3} holds also for the Gaussian BC and Gaussian random variables is omitted  for brevity.} 
\begin{Example}
Consider the  \emph{Gaussian broadcast channel} 
\begin{subequations}
   \begin{IEEEeqnarray}{rCl}
   Y_1&=&X+ Z_1\\
      Y_2&=&X+ Z_2
 \end{IEEEeqnarray}
 \end{subequations}
 where $Z_1\sim \set{N}(0,N_1)$ and $Z_2\sim \set{N}(0,N_2)$ are independent noises. Assume an average transmission power $P$, and $0<N_2<N_1<P$.

 Let $Q=\textnormal{const.}$,  $U\sim  \set{N}(0,\bar{\alpha} P)$, $W_1\sim \set{N}(0,{\alpha} P)$ and $W_2\sim \set{N}(0,\beta)$, for $\alpha\in[0,1], \beta> 0$,  where $U,W_1,W_2$ are independent.  {We use $C(x):=\frac{1}{2}\log(1+x)$, for $x\geq 0$}.
Setting
$X=U+ W_1$, $\tilde{Y}_1=Y_1+ W_2$, we have 
 \begin{IEEEeqnarray*}{rCl}
 I(U;Y_1)= C\Big(\frac{\bar{\alpha}P}{\alpha P+N_1}\Big),  \quad I(X;Y_2)=C\Big(\frac{P}{N_2}\Big),
\end{IEEEeqnarray*}
 and  
 \begin{IEEEeqnarray*}{rCl}
  I(X;Y_2,\tilde{Y}_1|U)&=&C\Big(\frac{\alpha P}{N_2}\Big) +C\Big(\frac{\alpha PN_2}{(\alpha P+N_2)(N_1+\beta)}\Big)\nonumber\\
 I(\tilde{Y}_1;Y_1|Y_2,U)&=&C\Big(\frac{\alpha P(N_1+N_2)+N_1N_2}{\beta(N_2+\alpha P)}\Big).
\end{IEEEeqnarray*}
For these choices, the region defined by the constraints in~Corollary~\ref{Cor:thm3} evaluates to:
\begin{subequations}\label{eq:GausBS}
   \begin{IEEEeqnarray}{rCl}
  R_1&\leq& C\Big(\frac{\bar{\alpha}P}{\alpha P+N_1}\Big)\\
 R_1+R_2&\leq&  C\Big(\frac{\bar{\alpha}P}{\alpha P+N_1}\Big)+C\Big(\frac{\alpha P}{N_2}\Big) \nonumber\\&&\quad+C\Big(\frac{\alpha PN_2}{(\alpha P+N_2)(N_1+\beta)}\Big) \\
 R_1+ R_2&\leq& C\Big(\frac{P}{N_2}\Big) \!-\!C\Big(\frac{\alpha P(N_1\!+\!N_2)\!+\!N_1N_2}{\beta(N_2+\alpha P)}\Big)\qquad 
   \end{IEEEeqnarray}
 \end{subequations}
  for some $\alpha \in[0,1]$ and $\beta> 0$ satisfying
  \begin{IEEEeqnarray}{rCl}\label{gauFb}
 C\Big(&&\frac{\alpha P(N_1+N_2)+N_1N_2}{\beta(N_2+\alpha P)}\Big) \leq  R_{\textnormal{Fb},1}.
 \end{IEEEeqnarray}
 
 The region is plotted in Figure~\ref{fig:rateGaussianBC} against the nofeedback capacity region $\set{C}_\textnormal{NoFB}$, linear-feedback capacity \cite{belhadj'14} and {the Ozarow-Leung's  achievable region  \cite{ozarow'84}  with perfect feedback.}

{From the plots and from expressions~\eqref{eq:GausBS} and \eqref{gauFb} above, it is evident that the capacity region is  increased for any positive feedback-rate $R_{\FB,1}, R_{\FB,2}>0$. By Proposition~\ref{prop:noisy_fb}, the same holds also when the feedback links are additive Gaussian noise channels of capacities $R_{\FB,1}$ and $R_{\FB,2}$. This result very much hinges upon the fact that we allow the receivers to code over the feedback link. In fact,  Pillai and Prabhakaran showed in a recent work \cite{pillai} that when the receivers cannot code over the feedback links but are obliged to send back their received channel outputs, then the capacities without feedback and with noisy feedback coincide whenever the noise variances on the two feedback links exceed a certain threshold.} 

\begin{figure}[!t]
\centering
\includegraphics[width=0.5\textwidth]{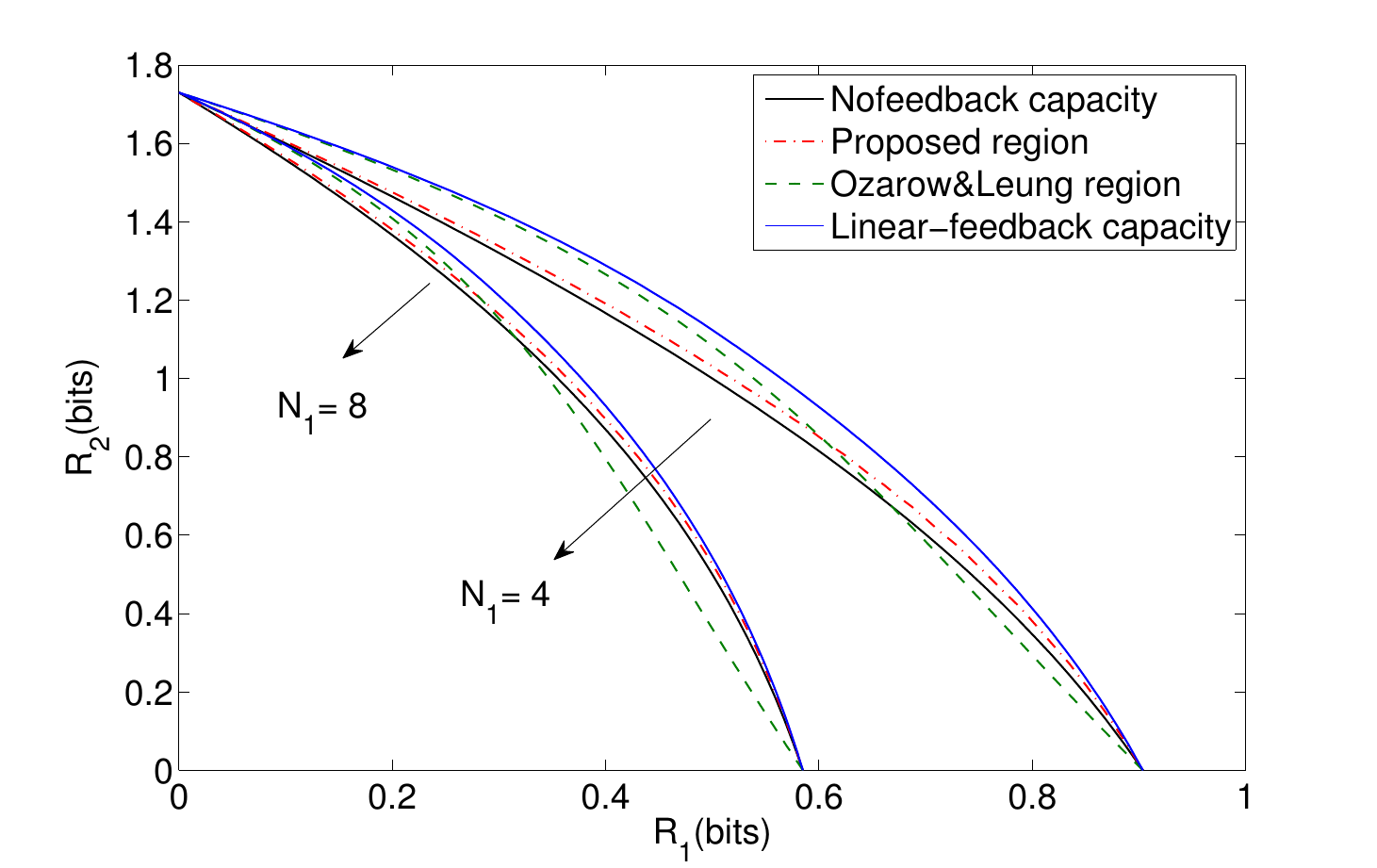}
\caption{$\set{C}_{\textnormal{NoFB}}$, linear-feedback capacity \cite{belhadj'14}, Ozarow and Leung's achievable region \cite{ozarow'84},  and the achievable region in~\eqref{eq:GausBS} are plotted for Gaussian BCs with parameters $P=10$, $N_2=1$, $N_1\in\{4,8\}$ and  feedback rate $R_{\textnormal{Fb},1}=0.8$.} \label{fig:rateGaussianBC}
\end{figure}

\end{Example}

\section{Coding schemes}\label{sec:allschm}
\subsection{Coding Scheme IA: Sliding-window decoding (Theorem~\ref{theosw})}\label{sec:schemeb}

For simplicity, we  only describe the scheme for $Q=\textnormal{const.}$ A general $Q$ can be introduced by coded time-sharing  \cite[Section 4.5.3]{book:gamal}. That means all the codebooks need to be {superposed} on a $P_Q$-i.i.d. random vector $Q^n$ that is revealed to transmitter and receivers, and this $Q^n$ sequence needs to be included in all the joint-typicality checks.

Choose nonnegative rates ${R}_1', {R}_2', \tilde{R}_1, \tilde{R}_2, \hat{R}_1,\hat{R}_2$, auxiliary finite alphabets $\set{U}_0,\set{U}_1,\set{U}_2, \set{\tilde{Y}}_1,\set{\tilde{Y}}_2$, a  function $f$ of the form $f$: $\set{U}_0\times\set{U}_1\times\set{U}_2\to \set{X}$, and pmfs $P_{U_0U_1U_2}$, $P_{\tilde{Y}_1|U_0Y_1}$, $P_{\tilde{Y}_2|U_0Y_2}$. 
Transmission takes place over $B+1$ consecutive blocks, with length $n$ for each block.  We denote the $n$-length blocks of inputs and outputs in block $b$ by $x^n_b$, $y_{1,b}^n$ and $y_{2,b}^n$. 

{Define $\set{J}_i:=\{1,\ldots,\lfloor 2^{n\hat{R}_i} \rfloor\}$,  $\set{K}_i:=\{1,\ldots,\lfloor 2^{n{R}'_i} \rfloor\}$, {and}
$\set{L}_i:=\{1,\ldots,\lfloor 2^{n\tilde{R}_i} \rfloor\}$, for $i\in\{1,2\}$.
The messages are in product form: 
{$M_i=(M_{i,1},\ldots, M_{i,B})$,  $i\in\{1,2\}$, with $M_{i,b}=(M_{c,i,b},M_{p,i,b})$ for $b\in\{1,\ldots,B\}$}. The submessages 
$M_{c,i,b}$, and $M_{p,i,b}$ are uniformly distributed over the sets 
$\set{M}_{c,i}:=\{1,\ldots,\lfloor 2^{nR_{c,i}} \rfloor\}$ and $\set{M}_{p,i}:=\{1,\ldots,\lfloor 2^{nR_{p,i}} \rfloor\}$, respectively, where $R_{p,i}, R_{c,i}>0$ and so that $R_i=R_{p,i}+R_{c,i}$\footnote{Due to the floor operations and {since transmission takes place over $B+1$ blocks whereas the messages $M_1$ and $M_2$ are split  into only $B$ submessages, $R_1$ and $R_2$ here do not exactly represent the transmission rates of messages $M_1$ and $M_2$. In the limit $n
\to\infty$ and $B\to\infty$, which is our case of interest, $R_1$ and $R_2$ however approach these transmission rates. Therefore, we neglect this technicality in the following.}}.
 The feedback message $M_{\FB,i,b-1}$  is uniformly distributed over the set $\set{L}_i$.  The common message ${M}_{c,b}:=(M_{c,1,b},M_{c,2,b})$ is   uniformly distributed over the set   $\set{M}_c:=\set{M}_{c,1}\times  \set{M}_{c,2}$.  
  \vspace{1mm}
 }
 
{The coding scheme is 
 {illustrated in} Table \ref{tab:Scheme1A}. In each block~$b$, the transmitter uses Marton's coding to send {both feedback messages $(M_{\FB,1,b-1},M_{\FB,2,b-1})$ and the common message~${M}_{c,b}$ in the cloud center $U^n_{0,b}({M}_{c,b},M_{\FB,1,b-1},M_{\FB,2,b-1})$, and} the private message $M_{p,i,b}$ in the satellite $U_{i,b}^n(M_{p,i,b},K_{i,b}|{M}_{c,b},M_{\FB,1,b-1},M_{\FB,2,b-1})$, for $K_{i,b}\in\set{K}_i$, $i\in\{1,2\}$. Receiver $1$ first decodes the messages sent in the cloud center $U^n_{0,b}$,  and then {simultaneously reconstructs Receiver 2's compression outputs  $\tilde{Y}^n_{2,b-1}$  and decodes its intended messages  sent in the satellite}  $U_{1,b-1}^n$. Finally, {it compresses its  channel outputs $Y^n_{1,b}$ by means of  Wyner-Ziv coding  and sends the compression} index $M_{\FB,1,b}$ over the feedback link.   Receiver~2 behaves in an analogous way.

 \begin{table*}[ht!]
\begin{center}
\caption{\textnormal{Coding scheme IA: The first four rows in the table depict the encoding process at the transmitter; the following two rows the compression mechanism at the receivers; and the last two rows the messages decoded at the two receivers. {The left-to-right arrows in the last two lines indicate that the receivers apply forward decoding.}}}
\begin{tabular}{>{\bfseries}lcccccc}
\toprule
Block &1& 2& $\ldots$ & $b$& $\ldots$  \\
\midrule
& $u^n_{0,1}({m}_{c,1},1,1)$ & $u^n_{0,1}({m}_{c,2},m_{\FB,1,1},m_{\FB,2,1})$& $\ldots$  & $u^n_{0,b}({m}_{c,b},m_{\FB,1,b-1},m_{\FB,2,b-1})$& $\ldots$  \\
$X$& $u_{1,1}^n(m_{p,1,1},k_{1,1}|{m}_{c,1},1,1)$ &  $u_{1,2}^n(m_{p,1,2},k_{1,2}|{m}_{c,2},m_{\FB,1,1},m_{\FB,2,1})$ & $\ldots$  & $u_{1,b}^n(m_{p,1,b},k_{1,b}|{m}_{c,b},m_{\FB,1,b-1},m_{\FB,2,b-1})$ &$\ldots$  \\
& $u_{2,1}^n(m_{p,2,1},k_{2,1}|{m}_{c,1},1,1)$ &  $u_{2,2}^n(m_{p,2,2},k_{2,2}|{m}_{c,2},m_{\FB,1,1},m_{\FB,2,1})$ & $\ldots$  & $u_{2,b}^n(m_{p,2,b},k_{2,b}|{m}_{c,b},m_{\FB,1,b-1},m_{\FB,2,b-1})$ &$\ldots$  \\
&$x^n_{1}(u^n_{0,1},u^n_{1,1}, u^n_{2,1})$ & $x^n_{2}(u^n_{0,2},u^n_{1,2}, u^n_{2,2})$ & $\ldots$ & $x^n_{b}(u^n_{0,b},u^n_{1,b}, u^n_{2,b})$ & $\ldots$\\ \midrule
$\tilde{Y}_1$& $\tilde{y}_{1,1}^n (m_{\FB,1,1}, j_{1,1}|{m}_{c,1},1,1\big)$ &  $\tilde{y}_{1,2}^n (m_{\FB,1,2}, j_{1,2}|{m}_{c,2},m_{\FB,1,1},m_{\FB,2,1}\big)$ & $\ldots$  & $\tilde{y}_{1,b}^n (m_{\FB,1,b}, j_{1,b}|{m}_{c,b},m_{\FB,1,b-1},m_{\FB,2,b-1}\big)$ &$\ldots$  \\
$\tilde{Y}_2$& $\tilde{y}_{2,1}^n (m_{\FB,2,1}, j_{2,1}|{m}_{c,1},1,1\big)$ &  $\tilde{y}_{2,2}^n (m_{\FB,2,2}, j_{2,2}|{m}_{c,2},m_{\FB,1,1},m_{\FB,2,1}\big)$ & $\ldots$  & $\tilde{y}_{2,b}^n (m_{\FB,2,b}, j_{2,b}|{m}_{c,b},m_{\FB,1,b-1},m_{\FB,2,b-1}\big)$ &$\ldots$  \\\midrule
${Y}_1$& $  {\hat{m}}_{c,1}^{{(1)}}$ &  $\rightarrow  ({\hat{m}}_{c,2}^{{(1)}}, \hat{m}_{\FB,2,1}), (\hat{j}_{2,1},\hat{m}_{p,1,1}, \hat{k}_{1,1})$ & $\ldots$  & $\rightarrow  ({\hat{m}}_{c,b}^{{(1)}}, \hat{m}_{\FB,2,b\!-\!1}), (\hat{j}_{2,b\!-\!1},\hat{m}_{p,1,b\!-\!1}, \hat{k}_{1,b\!-\!1})$ &$\ldots$  \\
${Y}_2$& $  {\hat{m}}_{c,1}^{{(2)}}$ &  $\rightarrow  ({\hat{m}}_{c,2}^{{(2)}}, \hat{m}_{\FB,1,1}), (\hat{j}_{1,1},\hat{m}_{p,2,1}, \hat{k}_{2,1})$ & $\ldots$  & $\rightarrow  ({\hat{m}}_{c,b}^{{(2)}}, \hat{m}_{\FB,1,b\!-\!1}), (\hat{j}_{1,b\!-\!1},\hat{m}_{p,2,b\!-\!1}, \hat{k}_{2,b\!-\!1})$ &$\ldots$  \\
\bottomrule
 \label{tab:Scheme1A}
\end{tabular}
\end{center}
\end{table*}
 
 }
 
\textit{1) Codebook generation}:
For each block $b \in \{1,\ldots,B+1\}$,  randomly and independently generate $\lfloor 2^{n(R_{c,1}+R_{c,2}+\tilde{R}_1+\tilde{R}_2)}\rfloor $ sequences $u^n_{0,b}({m}_{c,b},m_{\FB,1,b-1}, m_{\FB,2,b-1})$, for ${{{m}_{c,b}}} \in \set{M}_c$ and $m_{\FB,i,b-1}\in\set{L}_i$, for $i\in\{1,2\}$.  Each sequence $u^n_{0,b}({m}_{c,b},m_{\FB,1,b-1},m_{\FB,2,b-1})$ is drawn according to the product distribution $\prod_{t=1}^nP_{U_0}(u_{0,b,t})$, where $u_{0,b,t}$ denotes the $t$-th entry of $u_{0,b}^n({m}_{c,b},m_{\FB,1,b-1},m_{\FB,2,b-1})$.

For $i\in\{1,2\}$ and each $u^n_{0,b}({m}_{c,b},m_{\FB,1,b-1},m_{\FB,2,b-1})$ randomly and conditionally independently generate 
$\lfloor 2^{n(R_{p,i}+R'_i)}\rfloor$ sequences $u_{i,b}^n(m_{p,i,b},k_{i,b}|{m}_{c,b}, m_{\FB,1,b-1}$, $m_{\FB,2,b-1})$, for $m_{p,i,b} \in \set{M}_{p,i}$ and $k_{i,b}\in \set{K}_i$, where each  $u_{i,b}^n(m_{p,i,b},k_{i,b}|{m}_{c,b},m_{\FB,1,b-1},m_{\FB,2,b-1})$ is drawn according to the product distribution $\prod_{t=1}^nP_{U_i|U_0}(u_{i,b,t}|u_{0,b,t})$, where $u_{i,b,t}$ denotes the $t$-th entry of $u_{i,b}^n\big(m_{p,i,b},k_{i,b}|{m}_{c,b},m_{\FB,1,b-1},m_{\FB,2,b-1}\big)$.

{Similarly, for $i\in\{1,2\}$  and each tuple $({m}_{c,b},m_{\FB,1,b-1},m_{\FB,2,b-1})\in \set{M}_{c}\times \set{L}_{1}\times \set{L}_2$ randomly generate $\lfloor 2^{n(\tilde{R}_i+ \hat{R}_i)}\rfloor $ sequences $\tilde{y}_{i,b}^n (m_{\FB,i,b}, j_{i,b}|{m}_{c,b}$, $m_{\FB,1,b-1},m_{\FB,2,b-1}\big)$, for  $m_{\FB,i,b}\in\set{L}_i$ and $j_{i,b}\in\set{J}_i$, by drawing  each  $\tilde{y}_{i,b}^n (m_{\FB,i,b}, j_{i,b}|{m}_{c,b},m_{\FB,1,b-1},m_{\FB,2,b-1}\big)$ according to the product distribution $\prod_{t=1}^nP_{\tilde{Y}_i|U_0,Y_i}(\tilde{y}_{i,b,t}|u_{0,b,t})$ }where $\tilde{y}_{i,b,t}$ denotes the $t$-th entry of $\tilde{y}_{i,b}^n (m_{\FB,i,b}, j_{i,b}|{m}_{c,b},m_{\FB,1,b-1},m_{\FB,2,b-1}\big)$.

All codebooks are revealed to  transmitter and  receivers. 
\vspace{1mm}

\textit{2) Encoding}:
We describe the encoding for a fixed block $b\in\{1,\ldots, B+1\}$. 
Assume  that  $M_{c,i,b}=m_{c,i,b}$, $M_{p,i,b}=m_{p,i,b}$, for $i\in\{1,2\}$ and  that the feedback messages sent after block $b-1$ are $M_{\FB,1,b-1}=m_{\FB,1,b-1}$ and $M_{\FB,2,b-1}=m_{\FB,2,b-1}$. Define ${m}_{c,b}:=(m_{c,1,b}, m_{c,2,b})$. To simplify notation, let $m_{\FB,i,0}=m_{c,i,B+1}=m_{p,i,B+1}=1$, for $i\in\{1,2\}$ and ${m}_{c,B+1}=(1,1)$.

 The transmitter looks for a pair   $(k_{1,b},k_{2,b})\in \set{K}_1\times \set{K}_2$  that satisfies
 \begin{IEEEeqnarray*}{rCl}
 &&\big( u_{1,b}^n(m_{p,1,b},k_{1,b}|{m}_{c,b},m_{\FB,1,b-1},m_{\FB,2,b-1}) ,\nonumber \\
&& \quad u_{2,b}^n(m_{p,2,b},k_{2,b}|{m}_{c,b},m_{\FB,1,b-1},m_{\FB,2,b-1}),\nonumber\\
&& \qquad u^n_{0,b}({m}_{c,b},m_{\FB,1,b-1},m_{\FB,2,b-1}) \big)   \in \set{T}_{\varepsilon/16}^{(n)}(P_{U_0U_1U_2}). \quad
 \end{IEEEeqnarray*}
If there is exactly one pair $(k_{1,b},k_{2,b})$ that satisfies the above condition, the transmitter chooses this pair.  If there are multiple such pairs, it chooses one of them uniformly at random. Otherwise it chooses a pair $(k_{1,b},k_{2,b})$ uniformly at random over the entire set  $\set{K}_1\times \set{K}_2$. In block $b$ the transmitter then sends the inputs $x_b^n=(x_{b,1}, \ldots, x_{b,n})$,  where
 \begin{equation}
 x_{b,t}=f(u_{0,b,t},u_{1,b,t},u_{2,b,t}), \qquad t\in\{1,\ldots,n\},
 \end{equation}
and $u_{0,b,t}$, $u_{1,b,t}$, $u_{2,b,t}$ denote the $t$-th symbols of the chosen Marton codewords  $u^n_{0,b}({m}_{c,b},m_{\FB,1,b-1},m_{\FB,2,b-1})$, $u_{1,b}^n(m_{p,1,b},k_{1,b}|{m}_{c,b},m_{\FB,1,b-1},m_{\FB,2,b-1})$, and\\ $u_{2,b}^n(m_{p,2,b},k_{2,b}|{m}_{c,b},m_{\FB,1,b-1},m_{\FB,2,b-1})$. 
 \vspace{1mm}

\textit{3) Decoding  and generation of feedback messages at receivers}: We describe the operations performed at Receiver~1. Receiver~2 behaves in an analogous way.

After each block $b\in\{1,\ldots, B+1\}$, and after observing the outputs $y_{1,b}^n$, 
Receiver~1 looks  for a pair of indices $({\hat{m}}_{c,b}^{{(1)}}, \hat{m}_{\FB,2,b-1})\in \set{M}_{c}\times \set{L}_2$ that satisfies
 \begin{IEEEeqnarray*}{rCl}
\big(u^n_{0,b}({\hat{m}}_{c,b}^{{(1)}}, m_{\FB,1,b-1},\hat{m}_{\FB,2,b-1}),  y_{1,b}^n \big)\in \set{T}_{\varepsilon/8}^{(n)}(P_{U_0Y_1}).
\end{IEEEeqnarray*}
Notice that Receiver 1 already knows $m_{\FB,1,b-1}$ because it has created it itself after the previous block $b-1$. 

If there are multiple such pairs, the receiver chooses one of them at random. If there is no such pair, then it chooses $({\hat{m}}_{c,b}^{{(1)}}, \hat{m}_{\FB,2,b-1})$ randomly over the set $\set{M}_{c}\times \set{L}_2$.

After decoding the cloud center in block $b$, Receiver 1 looks for a tuple $(\hat{j}_{2,b-1},\hat{m}_{p,1,b-1}, \hat{k}_{1,b-1})\in \set{J}_2\times \set{M}_{p,1} \times \set{K}_1$ that satisfies
 \begin{IEEEeqnarray*}{rCl}
&&\big(\tilde{y}_{2,b\!-\!1}^n(\hat{m}_{\FB,2,b\!-\!1}, \hat j_{2,b\!-\!1}|\hat{{m}}_{c,b\!-\!1}^{{(1)}}, m_{\FB,1,b\!-\!2},\hat{m}_{\FB,2,b\!-\!2}), y_{1,b-1}^n,\nonumber\\
&&~\quad u^n_{1,b-1}(\hat{m}_{p,1,b-1},\hat{k}_{1,b-1}|\hat{{m}}_{c,b-1}^{{(1)}},m_{\FB,1,b\!-\!2},\hat{m}_{\FB,2,b\!-\!2}),\nonumber\\
&&\qquad u^n_{0,b-1}(\hat{{m}}_{c,b-1}^{{(1)}},\! m_{\FB,1,b\!-\!2},\!\hat{m}_{\FB,2,b\!-\!2}) \big)\in \set{T}_\varepsilon^{(n)}(P_{U_{0}U_{1}Y_1\tilde{Y}_2}).
\end{IEEEeqnarray*}
It further looks 
for a pair $(m_{\FB,1,b}, j_{1,b})\in \set{L}_1\times \set{J}_1$  that satisfies
 \begin{IEEEeqnarray*}{rCl}
 \lefteqn{
(\tilde{y}_{1,b}^n(m_{\FB,1,b}, j_{1,b}|{\hat{m}}_{c,b}^{{(1)}}, m_{\FB,1,b-1},\hat{m}_{\FB,2,b-1}),}  \qquad \\ 
& &  u^n_{0,b}({\hat{m}}_{c,b}^{{(1)}}, m_{\FB,1,b-1},\hat{m}_{\FB,2,b-1}), y_{1,b}^n)\in\set{T}_{\varepsilon/4}^{(n)}(P_{Y_1U_0\tilde{Y}_1})
\end{IEEEeqnarray*}
and sends the index $m_{\FB,1,b}$ over the feedback link. If there is more than one such pair $(m_{\FB,1,b}, j_{1,b})$ the encoder chooses one of them at random. If there is none, it chooses the index  $m_{\FB,1,b}$ that it sends over the feedback link uniformly at random over $\set{L}_1$. The receivers  thus only send a feedback message at the end of each block $1,\ldots,B$. 

After  decoding block~$B+1$, Receiver 1 produces the product message 
$\hat{m}_1=(\hat{m}_{1,1},\ldots,\hat{m}_{1,B})$ as its guess, where $\hat{m}_{1,b}=(\hat{m}_{c,1,b}^{{(1)}},\hat{m}_{p,1,b})$, for $b\in\{1,\ldots,B\}$, {and $\hat{m}_{c,1,b}^{{(1)}}$ denotes the first component of ${\hat{m}}_{c,b}^{{(1)}}$.}
 \vspace{1mm}

\textit{4) Analysis}: See Appendix~\ref{sec:analysis_relayb}.

\subsection{Coding Scheme IB: Backward decoding (Theorem~\ref{theobw})}\label{sec:backwarddecoding}

The coding scheme is similar to Scheme IA, except that here {the two receivers jointly decode the cloud center and their intended satellite using backward coding, and the Wyner-Ziv codes are not {superposed}  on the Marton cloud center. 
 The scheme is {illustrated in} Table \ref{tab:Scheme1B}.  

 \begin{table*}[ht!]
\begin{center}
\caption{ \textnormal{Coding scheme IB: The first four rows in the table depict the encoding process at the transmitter; the following two rows the compression mechanism at the receivers; and the last two rows the messages decoded at the two receivers. {The right-to-left arrows in the last two rows indicate that the receivers use backward decoding.}} }
\begin{tabular}{>{\bfseries}lcccccc}
\toprule
Block &1& 2& $\ldots$ & $b$& $\ldots$  \\
\midrule
& $u^n_{0,1}({m}_{c,1},1,1)$ & $u^n_{0,1}({m}_{c,2},m_{\FB,1,1},m_{\FB,2,1})$& $\ldots$  & $u^n_{0,b}({m}_{c,b},m_{\FB,1,b-1},m_{\FB,2,b-1})$& $\ldots$  \\
$X$& $u_{1,1}^n(m_{p,1,1},k_{1,1}|{m}_{c,1},1,1)$ &  $u_{1,2}^n(m_{p,1,2},k_{1,2}|{m}_{c,2},m_{\FB,1,1},m_{\FB,2,1})$ & $\ldots$  & $u_{1,b}^n(m_{p,1,b},k_{1,b}|{m}_{c,b},m_{\FB,1,b-1},m_{\FB,2,b-1})$ &$\ldots$  \\
& $u_{2,1}^n(m_{p,2,1},k_{2,1}|{m}_{c,1},1,1)$ &  $u_{2,2}^n(m_{p,2,2},k_{2,2}|{m}_{c,2},m_{\FB,1,1},m_{\FB,2,1})$ & $\ldots$  & $u_{2,b}^n(m_{p,2,b},k_{2,b}|{m}_{c,b},m_{\FB,1,b-1},m_{\FB,2,b-1})$ &$\ldots$  \\
&$x^n_{1}(u^n_{0,1},u^n_{1,1}, u^n_{2,1})$ & $x^n_{2}(u^n_{0,2},u^n_{1,2}, u^n_{2,2})$ & $\ldots$ & $x^n_{b}(u^n_{0,b},u^n_{1,b}, u^n_{2,b})$ & $\ldots$\\ 
\midrule
$\tilde{Y}_1$& $\tilde{y}_{1,1}^n (m_{\FB,1,1}, j_{1,1}\big)$ &  $\tilde{y}_{1,2}^n (m_{\FB,1,2}, j_{1,2}\big)$ & $\ldots$  & $\tilde{y}_{1,b}^n (m_{\FB,1,b}, j_{1,b}\big)$ &$\ldots$  \\
$\tilde{Y}_2$& $\tilde{y}_{2,1}^n (m_{\FB,2,1}, j_{2,1}\big)$ &  $\tilde{y}_{2,2}^n (m_{\FB,2,2}, j_{2,2}\big)$ & $\ldots$  & $\tilde{y}_{2,b}^n (m_{\FB,2,b}, j_{2,b}\big)$ &$\ldots$  \\\midrule
${Y}_1$& $  (\hat{j}_{2,1} ,\hat {{m}}_{c,1}^{{(1)}}, \hat{m}_{p,1,1}, \hat{k}_{1,1})$ &  $\leftarrow  (\hat{j}_{2,2} ,\hat {{m}}_{c,2}^{{(1)}}, \hat m_{\FB,2,1}, \hat{m}_{p,1,2}, \hat{k}_{1,2})$ &$\ldots$ & $\leftarrow  (\hat{j}_{2,b} ,\hat {{m}}_{c,b}^{{(1)}}, \hat m_{\FB,2,b-1}, \hat{m}_{p,1,b}, \hat{k}_{1,b})$ &$\ldots$  \\
${Y}_2$& $  (\hat{j}_{1,1} ,\hat {{m}}_{c,1}^{{(2)}}, \hat{m}_{p,2,1}, \hat{k}_{2,1})$ &  $\leftarrow  (\hat{j}_{1,2} ,\hat {{m}}_{c,2}^{{(2)}}, \hat m_{\FB,1,1}, \hat{m}_{p,2,2}, \hat{k}_{2,2})$ &$\ldots$ & $\leftarrow  (\hat{j}_{1,b} ,\hat {{m}}_{c,b}^{{(2)}}, \hat m_{\FB,1,b-1}, \hat{m}_{p,2,b}, \hat{k}_{2,b})$ &$\ldots$  \\
\bottomrule
 \label{tab:Scheme1B}
\end{tabular}
\end{center}
\end{table*}
 }
 
For simplicity, we describe the scheme without the coded time-sharing random variable $Q$, i.e., for $Q=\textnormal{const.}$  

Choose nonnegative rates ${R}_1', {R}_2', \tilde{R}_1, \tilde{R}_2, \hat{R}_1,\hat{R}_2$, auxiliary finite alphabets $\set{U}_0,\set{U}_1,\set{U}_2, \set{\tilde{Y}}_1,\set{\tilde{Y}}_2$, a  function $f$ of the form $f$: $\set{U}_0\times\set{U}_1\times\set{U}_2\to \set{X}$, and pmfs $P_{U_0U_1U_2}$, $P_{\tilde{Y}_1|Y_1}$, $P_{\tilde{Y}_2|Y_2}$. 
Transmission takes place over $B+1$ consecutive blocks, with length $n$ for each block.  We denote the $n$-length blocks of inputs and outputs in block $b$ by $x^n_b$, $y_{1,b}^n$ and $y_{2,b}^n$. 

Define $\set{J}_i:=\{1,\ldots,\lfloor 2^{n\hat{R}_i} \rfloor\}$,  $\set{K}_i:=\{1,\ldots,\lfloor 2^{n{R}'_i} \rfloor\}$, {and}
$\set{L}_i:=\{1,\ldots,\lfloor 2^{n\tilde{R}_i} \rfloor\}$ , for $i\in\{1,2\}$.
The messages are in product form: 
$M_i=(M_{i,1},\ldots, M_{i,B})$,  $i\in\{1,2\}$,  with $M_{i,b}=(M_{c,i,b},M_{p,i,b})$ for $b\in\{1,\ldots,B\}$. The submessages 
$M_{c,i,b}$, and $M_{p,i,b}$ are uniformly distributed over the sets 
$\set{M}_{c,i}:=\{1,\ldots,\lfloor 2^{nR_{c,i}} \rfloor\}$ and $\set{M}_{p,i}:=\{1,\ldots,\lfloor 2^{nR_{p,i}} \rfloor\}$, respectively, where $R_{p,i}, R_{c,i}>0$ and so that $R_i=R_{p,i}+R_{c,i}$.

\textit{1) Codebook generation}:
For each block $b \in \{1,\ldots,B+1\}$,  randomly and independently generate $\lfloor 2^{n(R_{c,1}+R_{c,2}+\tilde{R}_1+\tilde{R}_2)}\rfloor$ sequences $u^n_{0,b}({m}_{c,b},m_{\FB,1,b-1}, m_{\FB,2,b-1})$, for ${m}_{c,b} \in \set{M}_c:=\set{M}_{c,1}\times  \set{M}_{c,2}$ and $m_{\FB,i,b-1}\in\set{L}_i$, for $i\in\{1,2\}$.   Each sequence $u^n_{0,b}({m}_{c,b},m_{\FB,1,b-1},m_{\FB,2,b-1})$ is drawn according to the product distribution $\prod_{t=1}^nP_{U_0}(u_{0,b,t})$, where $u_{0,b,t}$ denotes the $t$-th entry of $u_{0,b}^n({m}_{c,b},m_{\FB,1,b-1},m_{\FB,2,b-1})$.

For $i\in\{1,2\}$ and each {tuple $({m}_{c,b},m_{\FB,1,b-1},m_{\FB,2,b-1})$ randomly  generate }
$\lfloor 2^{n(R_{p,i}+R'_i)} \rfloor$ sequences $u_{i,b}^n(m_{p,i,b},k_{i,b}|{m}_{c,b},m_{\FB,1,b-1},m_{\FB,2,b-1})$, for $m_{p,i,b} \in \set{M}_{p,i}$ and $k_{i,b}\in \set{K}_i$ {by randomly drawing each codeword}  $u_{i,b}^n(m_{p,i,b},k_{i,b}|{m}_{c,b},m_{\FB,1,b-1},m_{\FB,2,b-1})$ according to the product distribution $\prod_{t=1}^nP_{U_i|U_0}(u_{i,b,t}|u_{0,b,t})$, where $u_{i,b,t}$ denotes the $t$-th entry of $u_{i,b}^n\big(m_{p,i,b},k_{i,b}|{m}_{c,b},m_{\FB,1,b-1},m_{\FB,2,b-1}\big)$. 

Also, for $i\in\{1,2\}$, 
generate $\lfloor 2^{n(\tilde{R}_i+ \hat{R}_i)} \rfloor$ sequences $\tilde{y}_{i,b}^n (m_{\FB,i,b}, j_{i,b}\big)$, for  $m_{\FB,i,b}\in\set{L}_i$ and $j_{i,b}\in\set{J}_i$, by drawing all the entries independently according to the same distribution $P_{\tilde{Y}_i}$.

All codebooks are revealed to  transmitter and  receivers. %
\vspace{1mm}

\textit{2) Encoding}:
We describe the encoding for a fixed block $b\in\{1,\ldots, B+1\}$. 
Assume  that  $M_{c,i,b}=m_{c,i,b}$, $M_{p,i,b}=m_{p,i,b}$, for $i\in\{1,2\}$, and  that the feedback messages sent after block $b-1$ are $M_{\FB,1,b-1}=m_{\FB,1,b-1}$ and $M_{\FB,2,b-1}=m_{\FB,2,b-1}$. Define ${m}_{c,b}:=(m_{c,1,b}, m_{c,2,b})$. To simplify notation, let $m_{\FB,i,0}=m_{c,i,B+1}=m_{p,i,B+1}=1$, for $i\in\{1,2\}$ and $m_{c,B+1}=(1,1)$.

 The transmitter looks for a pair   $(k_{1,b},k_{2,b})\in \set{K}_1\times \set{K}_2$  that satisfies
 \begin{IEEEeqnarray*}{rCl}
&& \big( u_{1,b}^n(m_{p,1,b},k_{1,b}|{m}_{c,b},m_{\FB,1,b-1},m_{\FB,2,b-1}) ,\nonumber \\
&& \quad u_{2,b}^n(m_{p,2,b},k_{2,b}|{m}_{c,b},m_{\FB,1,b-1},m_{\FB,2,b-1}),
\nonumber\\
&&\qquad u^n_{0,b}({m}_{c,b},m_{\FB,1,b-1},m_{\FB,2,b-1}) \big )\in \set{T}_{\varepsilon/16}^{(n)}(P_{U_0U_1U_2}).\quad
 \end{IEEEeqnarray*}
If there is exactly one pair $(k_{1,b},k_{2,b})$ that satisfies the above condition, the transmitter chooses this pair.  If there are multiple such pairs, it chooses one of them uniformly at random. Otherwise it chooses a pair $(k_{1,b},k_{2,b})$ uniformly at random over the entire set  $\set{K}_1\times \set{K}_2$. In block $b$ the transmitter then sends the inputs $x_b^n=(x_{b,1}, \ldots, x_{b,n})$,  where
 \begin{equation}
 x_{b,t}=f(u_{0,b,t},u_{1,b,t},u_{2,b,t}), \qquad t\in\{1,\ldots,n\},
 \end{equation}
and $u_{0,b,t}$, $u_{1,b,t}$, $u_{2,b,t}$ denote the $t$-th symbols of the chosen Marton codewords  $u^n_{0,b}({m}_{c,b},m_{\FB,1,b-1},m_{\FB,2,b-1})$, $u_{1,b}^n(m_{p,1,b},k_{1,b}|{m}_{c,b},m_{\FB,1,b-1},m_{\FB,2,b-1})$, and\\ $u_{2,b}^n(m_{p,2,b},k_{2,b}|{m}_{c,b},m_{\FB,1,b-1},m_{\FB,2,b-1})$. 
 \vspace{1mm}

\textit{3) Generation of feedback messages at receivers}: We describe the operations performed at Receiver~1. Receiver~2 behaves in an analogous way.

After each block $b\in\{1,\ldots, B\}$, and after observing the outputs $y_{1,b}^n$, 
Receiver~1 looks for a pair $(m_{\FB,1,b}, j_{1,b})\in \set{L}_1\times \set{J}_1$  that satisfies
\begin{equation}
(\tilde{y}_{1,b}^n(m_{\FB,1,b}, j_{1,b}), y_{1,b}^n)\in\set{T}_{\varepsilon/4}^{(n)}(P_{\tilde{Y}_1{Y}_1})
\end{equation}
and sends the index $m_{\FB,1,b}$ over the feedback link. If there is more than one such pair $(m_{\FB,1,b}, j_{1,b})$ the encoder chooses one of them at random. If there is none, it chooses the index  $m_{\FB,1,b}$ that it sends over the feedback link uniformly at random over~$\set{L}_1$. 

In our scheme the receivers  thus only send a feedback message at the end of each block $1,\ldots,B$. 
\vspace{1mm}

\textit{4) Decoding at receivers}: 
We describe the operations performed at Receiver~1. Receiver~2 behaves in an analogous way.

The receivers apply backward decoding and thus start decoding only after the transmission terminates. Then, for each block $b\in\{1,\ldots,B+1\}$, starting with the last block $B+1$, Receiver~1 performs the following operations. From the previous decoding step in block $b+1$, it already knows the feedback message $m_{\FB,2,b}$. Moreover, it also knows its own feedback messages $m_{\FB,1,b-1}$ and $m_{\FB,1,b}$ because it has created them itself, see point \textit{3)}. Now, when observing ${y}_{1,b}^n$, Receiver~1  looks for a tuple $(\hat{j}_{2,b} ,\hat {{m}}_{c,b}^{{(1)}}, \hat m_{\FB,2,b-1}, \hat{m}_{p,1,b}, \hat{k}_{1,b})\in  \set{J}_2\times \set{M}_c\times \set{L}_2 \times \set{M}_{p,1} \times \set{K}_1$ that satisfies 
 \begin{IEEEeqnarray*}{rCl}\lefteqn{
\Big(u^n_{1,b}(\hat{m}_{p,1,b},\!\hat{k}_{1,b}|{\hat{m}}_{c,b}^{{(1)}},\!m_{\FB,1,b\!-\!1},\!\hat{m}_{\FB,2,b\!-\!1}),\tilde{y}_{2,b}^n({m}_{\FB,2,b}, \hat j_{2,b}), }\nonumber\\
&&\qquad u^n_{0,b}({\hat{m}}_{c,b}^{{(1)}},\! m_{\FB,1,b\!-\!1},\!\hat{m}_{\FB,2,b\!-\!1}),y_{1,b}^n \Big)\in \set{T}_ \varepsilon ^{(n)}(P_{U_0U_1Y_1\tilde{Y}_2}).
\end{IEEEeqnarray*}
After decoding block~$1$, Receiver 1 produces the product message 
$\hat{m}_1=(\hat{m}_{1,1},\ldots,\hat{m}_{1,B})$ as its guess, where $\hat{m}_{1,b}=(\hat{m}_{c,1,b}^{{(1)}},\hat{m}_{p,1,b})$, for $b\in\{1,\ldots,B\}$, {and $\hat{m}_{c,1,b}^{(1)}$ denotes the first component of ${\hat{m}}_{c,b}^{(1)}$.}
 \vspace{1mm}

\textit{5) Analysis}: See Appendix~\ref{sec:analysis_relay}. 

\subsection{Coding Scheme IC: Hybrid sliding-window decoding and backward decoding (Theorem~\ref{Thm:oneside})}\label{schm:1c}

Our third scheme~IC  is a mixture of the Scheme IA and IB:   Receiver~1 behaves as in  Scheme IA and   Receiver~2 as in Scheme IB. {It 
achieves region $\set{R}_\textnormal{relay,hb}^{(1)}$. A similar scheme achieves region $\set{R}_\textnormal{relay,hb}^{(2)}$.} 

{For simplicity we consider only $Q=\textnormal{const.}$ }

\textit{1) Codebook generation}:
The codebooks are generated  as  in Scheme {IA}, described in point 1) in Section \ref{sec:schemeb}, but where  $\tilde{R}_2=\hat{R}_2=0$. 

\textit{2) Encoding}:
The transmitter performs {the} same encoding  procedure as in Section \ref{sec:schemeb}, but where $m_{\FB,2,b-1}=1$ is constant for all blocks $b\in\{1,\ldots,B+1\}$. 

\textit{3) Receiver 1}:
In each block $b\in\{1,\ldots,B+1\}$,  Receiver 1 first simultaneously decodes the cloud center and its satellite. Specifically,   
Receiver 1 looks for a tuple $(\hat{{m}}_{c,b-1}, \hat{m}_{p,1,b-1}, \hat{k}_{1,b-1})\in \set{M}_c\times \set{M}_{p,1} \times \set{K}_1$ that satisfies
 \begin{IEEEeqnarray*}{rCl}
&&\big(u^n_{0,b-1}(\hat{{m}}_{c,b-1},\! m_{\FB,1,b\!-\!2},1), y_{1,b-1}^n,\nonumber\\
&&~ ~u^n_{1,b\!-\!1}(\hat{m}_{p,1,b\!-\!1},\!\hat{k}_{1,b\!-\!1}|\hat{{m}}_{c,b\!-\!1},\!m_{\FB,1,b\!-\!2},\!1) \big)\in \set{T}_ \varepsilon ^{(n)}(P_{U_{0}U_{1}Y_1}).
\end{IEEEeqnarray*}
It further compresses the outputs $y^n_{1,b}$ and sends the feedback message $m_{\FB,1,b}$ over the feedback link as in Scheme IA, see point 3) in Section \ref{sec:schemeb}.

\textit{4) Receiver 2}:
Receiver 2 performs backward decoding  as in Scheme IB, see point 4) in Section \ref{sec:backwarddecoding}.

\textit{5) Analysis}: {Similar to the analysis of the schemes IA and IB presented in appendices~\ref{sec:analysis_relayb} and \ref{sec:analysis_relay}. Details are omitted.} 

\subsection{Coding Scheme 2: Encoder processes feedback-info}\label{sec:processing}

The scheme described in this subsection differs from the previous scheme
in that in each block $b$, after receiving the feedback messages $M_{\FB,1,b}, M_{\FB,2,b}$, the encoder first reconstructs the compressed versions of the channel outputs, $\tilde{Y}_{1,b}^n$ and $\tilde{Y}_{2,b}^{n}$. It then newly compresses the quintuple consisting of  $\tilde{Y}_{1,b}^n$ and $\tilde{Y}_{2,b}^{n}$ and the Marton codewords $U_{0,b}^n$, $U_{1,b}^n$, $U_{2,b}^n$ that it had sent during block $b$. This new compression information is sent to the two receivers in the next-following block $b+1$ as part of the  cloud center of Marton's code.

Decoding at the receivers is based on backward decoding. For each block $b$,  each receiver~$i\in\{1,2\}$ uses its observed outputs $Y_{i,b}^n$ to  simultaneously reconstruct the encoder's compressed signal and  decode its intended messages sent in block~$b$.  {The scheme is illustrated in Table \ref{tab:Scheme2}.}


 \begin{table*}[ht!]
\begin{center}
\caption{ \textnormal{Coding scheme II: The first row depicts the compression mechanism at the transmitter; the following four rows the encoding mechanism at the transmitter; the subsequent two rows the compression mechanism at the receivers; and the last two rows the messages decoded at the two receivers. {The right-to-left arrows in the last two rows indicates that backward-decoding is used.} }}
\begin{tabular}{>{\bfseries}lcccccc}
\toprule
Block &1& 2& $\ldots$ & $b$& $\ldots$  \\
\midrule
& $v^n_{1}(n_1|1)$ & $v^n_{2}(n_2|n_{1})$& $\ldots$  & $v^n_{b}(n_b|n_{b-1})$& $\ldots$  \\\midrule
& $u^n_{0,b}({m}_{c,1},1)$ & $u^n_{0,2}({m}_{c,2},n_{1})$& $\ldots$  & $u^n_{0,b}({m}_{c,b},n_{b-1})$ & $\ldots$  \\ 
$X$& $u_{1,1}^n( m_{p,1,1}, k_{1,1}|{m}_{c,1}, 1)$ &  $u_{1,2}^n( m_{p,1,2}, k_{1,2}|{m}_{c,2}, n_{1})$& $\ldots$  & $u_{1,b}^n( m_{p,1,b}, k_{1,b}|{m}_{c,b}, n_{b-1})$ &$\ldots$  \\ 
& $u_{2,1}^n( m_{p,2,1}, k_{2,1}|{m}_{c,1}, 1)$ &  $u_{2,2}^n( m_{p,2,2}, k_{2,2}|{m}_{c,2}, n_{1})$& $\ldots$  & $u_{2,b}^n( m_{p,2,b}, k_{2,b}|{m}_{c,b}, n_{b-1})$ &$\ldots$  \\
&$x^n_{1}(u^n_{0,1},u^n_{1,1}, u^n_{2,1})$ & $x^n_{2}(u^n_{0,2},u^n_{1,2}, u^n_{2,2})$ & $\ldots$ & $x^n_{b}(u^n_{0,b},u^n_{1,b}, u^n_{2,b})$ & $\ldots$\\ 
\midrule
$\tilde{Y}_1$&  $\tilde{y}_{1,1}^n (m_{\FB,1,1}, j_{1,1})$&  $\tilde{y}_{1,2}^n (m_{\FB,1,2}, j_{1,2})$ & $\ldots$  &  $\tilde{y}_{i,b}^n (m_{\FB,1,b}, j_{1,b})$ &$\ldots$  \\
$\tilde{Y}_2$&  $\tilde{y}_{2,1}^n (m_{\FB,2,1}, j_{2,1})$&   $\tilde{y}_{2,2}^n (m_{\FB,2,2}, j_{2,2})$ & $\ldots$  &  $\tilde{y}_{2,b}^n (m_{\FB,2,b}, j_{2,b})$ &$\ldots$  \\ \midrule
${Y}_1$&   $  ({\hat{m}}_{c,1}^{{(1)}}, \hat{m}_{p,1,1}, \hat{k}_{1,1})$ &  $\leftarrow  ({\hat{m}}_{c,2}^{{(1)}}, \hat{m}_{p,1,2}, \hat{k}_{1,2},\hat{n}^{{(1)}}_{1})$  &$\ldots$ & $\leftarrow  ({\hat{m}}_{c,b}^{{(1)}}, \hat{m}_{p,1,b}, \hat{k}_{1,b},\hat{n}^{{(1)}}_{b-1})$ &$\ldots$  \\
${Y}_2$&   $  ({\hat{m}}_{c,1}^{{(2)}}, \hat{m}_{p,2,1}, \hat{k}_{2,1})$ &  $\leftarrow  ({\hat{m}}_{c,2}^{{(2)}}, \hat{m}_{p,2,2}, \hat{k}_{2,2},\hat{n}^{{(2)}}_{1})$  &$\ldots$ & $\leftarrow  ({\hat{m}}_{c,b}^{{(2)}}, \hat{m}_{p,2,b}, \hat{k}_{2,b},\hat{n}^{{(2)}}_{b-1})$ &$\ldots$   \\
\bottomrule
 \label{tab:Scheme2}
\end{tabular}
\end{center}
\end{table*}


\label{sec:proc}
For simplicity, we  only describe the scheme for $Q=\textnormal{const.}$ 

Choose nonnegative rates ${R}_1', {R}_2', \tilde{R}_1, \tilde{R}_2, \hat{R}_1,\hat{R}_2, \tilde{R}_v$, auxiliary finite alphabets $\set{U}_0,\set{U}_1,\set{U}_2, \set{\tilde{Y}}_1,\set{\tilde{Y}}_2$, $\set{V}$, a  function $f$ of the form $f$: $\set{U}_0\times\set{U}_1\times\set{U}_2\to \set{X}$, and pmfs $P_{U_0U_1U_2}$, $P_{\tilde{Y}_1|Y_1}$, $P_{\tilde{Y}_2|Y_2}$, and $P_{V|U_0U_1U_2\tilde{Y}_1\tilde{Y}_2}$. 
Transmission takes place over $B+1$ consecutive blocks, with length $n$ for each block.  We denote the $n$-length blocks of channel inputs and outputs in block $b$ by $x^n_b$, $y_{1,b}^n$ and $y_{2,b}^n$. 

Define $\set{J}_i:=\{1,\ldots,\lfloor 2^{n\hat{R}_i} \rfloor\}$,  $\set{K}_i:=\{1,\ldots,\lfloor 2^{n{R}'_i} \rfloor\}$, {and}
$\set{L}_i:=\{1,\ldots,\lfloor 2^{n\tilde{R}_i} \rfloor\}$, for $i\in\{1,2\}$, and $\set{N}:=\{1,\ldots,\lfloor 2^{n\tilde{R}_v} \rfloor\}$ 
The messages are in product form: 
$M_i=(M_{i,1},\ldots, M_{i,B})$, $i\in\{1,2\}$, with $M_{i,b}=(M_{c,i,b},M_{p,i,b})$ for  $b\in\{1,\ldots,B\}$. The submessages 
$M_{c,i,b}$, and $M_{p,i,b}$ are uniformly distributed over the sets 
$\set{M}_{c,i}:=\{1,\ldots,\lfloor 2^{nR_{c,i}} \rfloor\}$ and $\set{M}_{p,i}:=\{1,\ldots,\lfloor 2^{nR_{p,i}} \rfloor\}$, respectively, where $R_{p,i}, R_{c,i}>0$ and so that $R_i=R_{p,i}+R_{c,i}$. 

\textit{1) Codebook generation}:
For each block $b \in \{1,\ldots,B+1\}$,  randomly and independently generate $\lfloor 2^{n(R_{c,1}+R_{c,2}+\tilde{R}_1+\tilde{R}_2)}\rfloor$ sequences $u^n_{0,b}({m}_{c,b},n_{b-1})$, for ${m}_{c,b} \in \set{M}_c:=\set{M}_{c,1}\times  \set{M}_{c,2}$ and $n_{b-1}\in\set{N}$.   Each sequence $u^n_{0,b}({m}_{c,b},n_{b-1})$ is drawn according to the product distribution $\prod_{t=1}^nP_{U_0}(u_{0,b,t})$, where $u_{0,b,t}$ denotes the $t$-th entry of $u_{0,b}^n({m}_{c,b},n_{b-1})$.

For $i\in\{1,2\}$ and each pair $({m}_{c,b},n_{b-1})$ randomly generate 
$\lfloor 2^{n(R_{p,i}+R'_i)}\rfloor$ sequences $u_{i,b}^n(m_{p,i,b},k_{i,b}|{m}_{c,b},n_{b-1})$, for $m_{p,i,b} \in \set{M}_{p,i}$ and $k_{i,b}\in \set{K}_i$, by drawing each codeword  $u_{i,b}^n(m_{p,i,b},k_{i,b}|{m}_{c,b},n_{b-1})$ according to the product distribution $\prod_{t=1}^nP_{U_i|U_0}(u_{i,b,t}|u_{0,b,t})$, where $u_{i,b,t}$ denotes the $t$-th entry of $u_{i,b}^n\big(m_{p,i,b},k_{i,b}|{m}_{c,b},n_{b-1}\big)$. 

Also, for $i\in\{1,2\}$, 
generate $\lfloor 2^{n(\tilde{R}_i+ \hat{R}_i)}\rfloor$ sequences $\tilde{y}_{i,b}^n (m_{\FB,i,b}, j_{i,b})$, for  $m_{\FB,i,b}\in\set{L}_i$ and $j_{i,b}\in\set{J}_i$ by drawing all the entries independently according to the same distribution $P_{\tilde{Y}_i}$;

Finally, for each $n_{b-1}\in\set{N}$, generate $ \lfloor 2^{n{R}_v}\rfloor $ sequences $v^n_{b}(n_b|n_{b-1})$, for $n_b\in\set{N}$ by drawing all entries independently according to the same distribution $P_V$.

All codebooks are revealed to  transmitter and  receivers. %
\vspace{1mm}

\textit{2) Encoding}:
 We describe the encoding for a fixed block $b\in\{1,\ldots, B+1\}$. 
Assume  that in this block we wish to send messages  $M_{c,i,b}=m_{c,i,b}$, $M_{p,i,b}=m_{p,i,b}$, for $i\in\{1,2\}$, and define ${m}_{c,b}:=(m_{c,1,b}, m_{c,2,b})$. To simplify notation, let $m_{\FB,i,0}=m_{c,i,B+1}=m_{p,i,B+1}=1$, for $i\in\{1,2\}$,   and also $n_{-1}=n_0=1$. 

The first step in the encoding is to reconstruct the compressed outputs pertaining to the previous block $\tilde{Y}_{1,b-1}^n$ and $\tilde{Y}_{2,b-1}^n$. Assume that  after block $b-1$ the transmitter received the feedback messages $M_{\FB,1,b-1}=m_{\FB,1,b-1}$ and $M_{\FB,2,b-1}=m_{\FB,2,b-1}$, and  that in this previous block it had produced the Marton codewords
$u_{0,b-1}^n:=u_{0,b-1}^n({m}_{c,b-1}, n_{b-2})$, $u_{1,b-1}^n:=u_{1,b-1}^n( m_{p,1,b-1}, k_{1,b-1}|{m}_{c,b-1}, n_{b-2})$, and $u_{2,b-1}^n:=u_{2,b-1}^n( m_{p,2,b-1}, k_{2,b-1}|{m}_{c,b-1}, n_{b-2})$. The transmitter then looks for a pair $(\hat{j}_{1,b-1}, \hat{j}_{2,b-1}) \in  \set{J}_1\times \set{J}_2$ that satisfies
\begin{IEEEeqnarray*}{rCl}
\lefteqn{
\big(u_{0,b-1}^n, u_{1,b-1}^n, 
  u_{2,b-1}^n,  
 \tilde{y}_{1,b-1}^n( m_{\FB,1,b-1},\hat{j}_{1,b-1}),} \nonumber \\
 & &\hspace{1cm}  \tilde{y}_{2,b-1}^n({m}_{\FB,2,b-1}, \hat{j}_{2,b-1})\big) \in \set{T}^{(n)}_{\varepsilon/4}(P_{U_0U_1U_2\tilde{Y}_1\tilde{Y}_2}).
\end{IEEEeqnarray*}
In a second step the encoder produces the new compression information pertaining to block $b-1$, which  it  then sends to the receivers during block $b$. To this end, it looks for an index $ n_{b-1}\in\set{N}$ that satisfies 
\begin{IEEEeqnarray*}{rCl}
\lefteqn{
\big(u_{0,b-1}^n, u_{1,b-1}^n, u_{2,b-1}^n,  \tilde{y}_{1,b-1}^n( m_{\FB,1,b-1}, \hat{j}_{1,b-1}),} \nonumber \\
&& \hspace{1.6cm} \tilde{y}_{2,b-1}^n({m}_{\FB,2,b-1}, \hat{j}_{2,b-1}), v_{b-1}^n( n_{b-1}|n_{b-2})\big)\nonumber \\ & & \qquad \hspace{4cm}  \in \set{T}_{\varepsilon/2}^{(n)}(P_{U_0U_1U_2\tilde{Y}_1\tilde{Y}_2V}).
\end{IEEEeqnarray*}
The transmitter now sends the fresh data and the compression message ${n}_{b-1}$ over the channel:  It thus looks for a pair   $(k_{1,b},k_{2,b})\in \set{K}_1\times \set{K}_2$  that satisfies
 \begin{IEEEeqnarray}{rCl}\lefteqn{
 \big(u^n_{0,b}({m}_{c,b},{n}_{b-1}),} \quad  \nonumber \\ 
 &&  u_{1,b}^n(m_{p,1,b},k_{1,b}|{m}_{c,b},{n}_{b-1}) ,\nonumber \\
&& \quad u_{2,b}^n(m_{p,2,b},k_{2,b}|{m}_{c,b},{n}_{b-1})\big)\in \set{T}_{\varepsilon/64}^{(n)}(P_{U_0U_1U_2}).\nonumber 
 \end{IEEEeqnarray}
If there is exactly one pair $(k_{1,b},k_{2,b})$ that satisfies the above condition, the transmitter chooses this pair.  If there are multiple such pairs, it chooses one of them uniformly at random. Otherwise it chooses a pair $(k_{1,b},k_{2,b})$ uniformly at random over the entire set  $\set{K}_1\times \set{K}_2$. In block $b$ the transmitter then sends the inputs $x_b^n=(x_{b,1}, \ldots, x_{b,n})$,  where
 \begin{equation}
 x_{b,t}=f(u_{0,b,t},u_{1,b,t},u_{2,b,t}), \qquad t\in\{1,\ldots,n\}.
 \end{equation}
and $u_{0,b,t}$, $u_{1,b,t}$, $u_{2,b,t}$ denote the $t$-th symbols of the chosen Marton codewords  $u^n_{0,b}({m}_{c,b},{n}_{b-1})$, $u_{1,b}^n(m_{p,1,b},k_{1,b}|{m}_{c,b},{n}_{b-1})$, and $u_{2,b}^n(m_{p,2,b},k_{2,b}|{m}_{c,b},{n}_{b-1})$. 

\textit{3) Generation of feedback messages at receivers}: We describe the operations performed at Receiver~1. Receiver~2 behaves in an analogous way.

After each block $b\in\{1,\ldots, B\}$, and after observing the outputs $y_{1,b}^n$, 
Receiver~1 looks for a pair of indices $(m_{\FB,1,b}, j_{1,b})\in \set{L}_1\times \set{J}_1$  that satisfies
\begin{equation}
(\tilde{y}_{1,b}^n(m_{\FB,1,b}, j_{1,b}), y_{1,b}^n)\in\set{T}_{\varepsilon/16}^{(n)}(P_{\tilde{Y}_1{Y}_1})
\end{equation}
and sends the index $m_{\FB,1,b}$ over the feedback link. If there is more than one such pair $(m_{\FB,1,b}, j_{1,b})$ the encoder chooses one of them at random. If there is none, it chooses the index  $m_{\FB,1,b}$ sent over the feedback link uniformly at random over $\set{L}_1$. 

In our scheme the receivers  thus only send a feedback message at the end of each block. 
\vspace{1mm}

\textit{4) Decoding at receivers}: 
We describe the operations performed at Receiver~1. Receiver~2 behaves in an analogous way.

The receivers apply backward decoding, so they wait until the end of the transmission. Then, for each block $b\in\{1,\ldots,B+1\}$, starting with the last block $B+1$, Receiver~1 performs the following operations. From the previous decoding step in block $b+1$, it already knows the compression index $n_{b}$. 
Now, when observing ${y}_{1,b}^n$, Receiver~1  
looks for a 
 tuple $({\hat{m}}_{c,b}^{{(1)}}, \hat{m}_{p,1,b}, \hat{k}_{1,b},\hat{n}_{b-1})\in\set{M}_c \times \set{M}_{p,1} \times \set{K}_1\times \set{N}$ that satisfies
 \begin{IEEEeqnarray*}{rCl}\lefteqn{
\big(u^n_{0,b}({\hat{m}}_{c,b}^{{(1)}}, \hat{n}_{b-1}),u^n_{1,b}(\hat{m}_{p,1,b},\hat{k}_{1,b}|{\hat{m}}_{c,b}^{{(1)}},\hat{n}_{b-1}),}\nonumber\\
&&\hspace{2.5cm}\quad{v}_{b}^n(n_{b}| \hat{n}_{b-1}), y_{1,b}^n, \tilde{y}_{1,b}^n(m_{\FB,1,b},j_{1,b}) \big)\nonumber \\ & & \hspace{5.4cm}\in \set{T}_ \varepsilon ^{(n)}(P_{U_0U_1VY_1\tilde{Y}_1}),
\end{IEEEeqnarray*}
where recall that Receiver~1 knows the indices $m_{\FB,1,b}$ and $j_{1,b}$ because it has constructed them itself under \textit{3)}.

After decoding block~$1$, Receiver 1 produces the product message 
$\hat{m}_1=(\hat{m}_{1,1},\ldots,\hat{m}_{1,B})$ as its guess, where $\hat{m}_{1,b}=(\hat{m}_{c,1,b}^{{(1)}},\hat{m}_{p,1,b})$, for $b\in\{1,\ldots,B\}$, and $\hat{m}_{c,1,b}^{{(1)}}$ denotes the first component of $\hat{m}_{c,1,b}^{{(1)}}$.\vspace{1mm}
 
\textit{5) Analysis}: See Appendix~\ref{sec:analysis_proc}. 


{
\section{Related setups and previous works}\label{Sec:Relations}
In this section we compare our results to previous works, and we discuss the related setups of DMBCs with noisy feedback channels and state-dependent DMBCs with strictly-causal state-information at the transmitter and causal state-information at the receivers.

\subsection{The Shayevitz-Wigger region and our region $\set{R}_{\textnormal{proc.}}^\infty$}\label{sec:related1}
{
Consider a restricted form of the  Shayevitz-Wigger region for output feedback where in the evaluation of Constraints~\eqref{eq:inner}  we limit the choices of  the auxiliaries to 
\begin{equation}\label{eq:choice1}
V_1=V_2=V_0=V.
\end{equation}
This restricted Shayevitz-Wigger region for output feedback is included in our new achievable region $\set{R}_{\textnormal{proc.}}^\infty$ when the feedback rates  satisfy
\begin{equation}
R_{\textnormal{Fb},1}\geq \log_2|\set{Y}_1| \qquad \textnormal{and} \qquad R_{\textnormal{Fb},2}\geq \log_2|\set{Y}_2|,
\end{equation}
 which is seen as follows. Recall that our achievable region $\set{R}_{\textnormal{proc.}}^\infty$ is characterized by Constraints~\eqref{eq:pri}. When the constraints in~\eqref{eq:inner} are specialized to \eqref{eq:choice1}, then \eqref{eq:inner1}, \eqref{eq:inner2}, and \eqref{eq:inner4} coincide with Constraints \eqref{eq:inner1f}, \eqref{eq:inner2f}, and \eqref{eq:inner4f}, and  the sum-rate constraint~\eqref{eq:swSum1} is tighter than the sum-rate constraint in \eqref{eq:corsum}. The latter holds because $\min_{i=\{1,2\}}\{a_i-b_i\}\geq \min_{i\in\{1,2\}}{a_i}-\max_{i\in\{1,2\}}{b_i}$ for any nonnegative $\{a_i,b_i\}^2_{i=1}$, 


As explained at the end of the next subsection, choices as in \eqref{eq:choice1} are particularly interesting, see also \eqref{eq:choice2} ahead.}

\subsection{The Shayevitz-Wigger region and our region $\set{R}_{\textnormal{relay,bw}}$} \label{sec:rel_relay}

Consider  a restricted Shayevitz-Wigger region where in the evaluation of Constraints~\eqref{eq:inner}  we limit the choices of  the auxiliaries to 
\begin{IEEEeqnarray}{rCl}\label{eq:choice2}
V_1=V_2=V_0= V=(\psi_1(Y_1, Q), \psi_2(Y_2,Q))
\end{IEEEeqnarray}
for some functions $\psi_1$ and $\psi_2$ on appropriate domains.

For choices satisfying~\eqref{eq:choice2}, the Shayevitz-Wigger rate-constraints in \eqref{eq:inner} reduce to
\begin{subequations}\label{eq:innerw}
\begin{IEEEeqnarray}{rCl}
 R_1 &\leq& I(U_0,U_1;Y_1,V|Q) - I(Y_2;V|Q,Y_1)\label{eq:inner1w}
\\
 R_2 &\leq& I(U_0,U_2;Y_2,V|Q) - I(Y_1;V|Q,Y_2)\label{eq:inner2w}
\\
 R_1+R_2 &\leq& I(U_1; Y_1,V|Q,U_0) + I(U_2; Y_2,V|Q,U_0) \nonumber \\
 & &+ \min_{i\in\{1,2\}}I(U_0;Y_i,V|Q) \nonumber \\
 & &-\max_{i\in\{1,2\}} I(Y_1,Y_2;V|Q,Y_i) - I(U_1;U_2|Q,U_0) \nonumber\\ \label{eq:swSum1w}
\\
\nonumber R_1+R_2 &\leq& I(U_0,U_1; Y_1,V|Q) + I(U_0,U_2; Y_2,V|Q) \\
&& - I(U_1;U_2|Q,U_0) - I(Y_2;V|Q,Y_1)\nonumber \\
&& -I(Y_1;V|Q,Y_2)\label{eq:inner4w}
\end{IEEEeqnarray}
\end{subequations}

Our new achievable region $\set{R}_\textnormal{relay,bw}$, which is characterized by Constraints~\eqref{eq:region_relay}, improves over this restricted Shayevitz-Wigger region whenever the feedback rates $R_{\FB,1}$, $R_{\FB,2}$ are sufficiently large  so that in our new region we can choose 
 \begin{equation}
 \label{eq:choice12}
  \tilde{Y}_1=\psi_1(Y_1,Q) \quad \textnormal{and} \quad\tilde{Y}_2=\psi_2(Y_2,Q)
 \end{equation} and so that  $\Delta_1,\Delta_2$ satisfy \eqref{eq:DeltaPfb}. 

In fact, for choices as in 
\eqref{eq:choice2}, the rate constraints in \eqref{eq:inner1w}, \eqref{eq:inner2w}, and \eqref{eq:inner4w} coincide with constraints \eqref{In2R1}, \eqref{In2R2}, and \eqref{In2Sum3}  when in the latter we choose $\tilde{Y}_1$ and $\tilde{Y}_2$ as described by \eqref{eq:choice12}. 
Moreover, the combination of the two sum-rate constraints~\eqref{In2Sum1} and \eqref{In2Sum2} is looser than the sum-rate constraint~\eqref{eq:swSum1w}, because the former involves a ``$\min_{i=\{1,2\}}\{a_i-b_i\}$"-term whereas the latter involves the smaller ``$\min_{i\in\{1,2\}}{a_i}-\max_{i\in\{1,2\}}{b_i}$"-term.

Choices as in \eqref{eq:choice1} or \eqref{eq:choice2}  describe an important special case of the Shayevitz-Wigger scheme, as all evaluations of the Shayevitz-Wigger region known to date are based on them. In particular, the capacity-achieving scheme for the two-user BEC-BC when both receivers are informed of all erasure events \cite{wang, greek} is a special case of the Shayevitz-Wigger scheme with a choice of auxiliaries as in \eqref{eq:choice2}. See also the following Subsections~\ref{Sec:Relations_Kim} and \ref{Sec:Relations_David} for further usages of this choice.

\subsection{Noisy feedback channels}\label{sec:ExNoisyFB}
Consider the related scenario where the feedback links are replaced by noisy channels of capacities $R_{\textnormal{Fb},1}$ and $R_{\textnormal{Fb},2}$ and where \emph{the decoders can code over these feedback links}. The following three modifications to our coding schemes suffice to ensure that our achievable regions remain valid: 
\begin{itemize} 
\item We {interleave} two instances of our coding schemes: one scheme operates during the odd blocks of the  BC and occupies the even blocks on the feedback links; the other scheme operates during the even blocks of the  BC and occupies the odd blocks on the feedback links. 
\item Instead of sending  after each block an uncoded feedback message  over the feedback links, the receivers  encode them using capacity-achieving codes for their feedback links and  send these codewords during the next block.  
\item After each block, the transmitter first decodes the messages sent over the feedback links during this block, and then uses the decoded feedback-messages in the same way as it used them in  the original scheme. 
\end{itemize}
\begin{Proposition} \label{prop:noisy_fb}Consider a DMBC with noisy feedback channels of capacities $R_{\FB,1}$ and $R_{\FB,2}$ where the receivers can code over the feedback channels. When modified as described above, our coding schemes in Section~\ref{sec:allschm} achieve the rate regions in 
 Theorems~\ref{theosw}--\ref{theo5}. 
\end{Proposition}
\begin{IEEEproof}
See Appendix~\ref{sec:proof_noisy}.
\end{IEEEproof}

\subsection{State-dependent DMBCs with state-information at the transmitter and the receivers}\label{sec:state}

{Consider a  state-dependent DMBC of channel law $P_{Y_1Y_2|XS}$, where the state $S$ is observed causally at both receivers and strictly-causally at the transmitter. That means, the transmitter learns the time-$t$ state $S_t$ right after producing the time-$t$ input $X_t$, and the receivers  learn $S_t$ at the same time as their outputs $Y_{1,t}$ or $Y_{2,t}$. 
There is no feedback. 

Our coding schemes in the previous section also apply to this related scenario when the following modifications are applied: {Let the ``feedback messages" $\{M_{1,\FB,b}\}$ and $\{M_{2,\FB,b}\}$, for $b\in\{1,\ldots,B\}$, be computed in function of the state symbols $S^n_b$ only. 
In this case, the receivers do not have to feed back anything, since the transmitter observes $S^n_{b}$ and can generate  $\{M_{1,\FB,b}\}$ and $\{M_{2,\FB,b}\}$ locally. }

{This means that }
our new achievable regions 
in Theorems~\ref{theosw}--\ref{theo5} are also valid for this related  scenario, if
we replace 
\begin{itemize}
\item the channel law $P_{Y_1Y_2|X}$ by a state-dependent law $P_{Y_1Y_2|XS}$; 
\item the compression variables $\tilde{Y}_1$ and $\tilde{Y}_2$ by the state $S$; 
\item the outputs $Y_1$ and $Y_2$ by the state-augmented outputs $(Y_1,S)$ and $(Y_2,S)$; 
\end{itemize} and we set $R_{\FB,1}, R_{\FB,2} >H(S)$.  (This last condition implies that Constraints~\eqref{eq:feedback_rates} are always satisfied {and} thus can be dropped.)

In particular, we have the following corollary to Theorem~\ref{theo5}:
\begin{Corollary}\label{cor:state}
Consider  a state-dependent DMBC $P_{Y_1Y_2|XS}$ {where the transmitter observes the state-sequence strictly-causally, and the receivers observe it causally.} 

For this scenario the region  $\set{R}_\textnormal{proc.}^{\textnormal{state}}$ is achievable, where  $\set{R}_\textnormal{proc.}^{\textnormal{state}}$ is the set of all nonnegative rate pairs $(R_1,R_2)$ that satisfy
 \begin{subequations}\label{eq:prig}
\begin{IEEEeqnarray}{rCl}
R_{1} &\leq& I(U_0,U_1;Y_1,V|S,Q)\nonumber \\
&& \; -I(V;U_0,U_1,U_2, {Y}_2|Y_1,S,Q)\label{eq:inner1g}\\
R_2&\leq& I(U_0,U_2;Y_2,V|S,Q)\nonumber \\
&& \;-I(V;U_0,U_1,U_2,{Y}_1|Y_2,S,Q) \label{eq:inner2g} \\
R_1 +R_2 & \leq &I(U_1;Y_1,V|U_0,S,Q)+I(U_2;Y_2,V|U_0,S,Q)\nonumber  \\ &&\;-I(U_1;U_2|U_0,S,{Q})\nonumber  \\ &&\;+\min_{i\in\{1,2\}} \{ I(U_0;Y_i,V|S,Q)\nonumber\\
&& \qquad -I(V;U_0,U_1,U_2, Y_1,{Y}_2|Y_i,S,Q) \}\label{eq:corsumg}\\
R_1 +R_2 & \leq &I(U_0,U_1;Y_1,V|S,Q)+   I(U_0,U_2;Y_2,V|S,Q) \nonumber \\ &&-I(V;U_0,U_1,U_2, {Y}_1|Y_2,S,Q)\nonumber \\&&-I(V;U_0,U_1,U_2,{Y}_2|Y_1,S,Q)\nonumber \\&& -I(U_1;U_2|U_0,{Q}) \label{eq:inner4g}
\end{IEEEeqnarray}
\end{subequations}
for some pmf $P_{Q}P_{U_0U_1U_2|Q} P_{V|U_0U_1U_2S{Q}}$ and some function $f\colon \mathcal{X} \to \mathcal{U}_0\times\mathcal{U}_1\times \mathcal{U}_2\times \mathcal{Q}$, where $X= f(U_0,U_1,U_2,Q)$.
\end{Corollary}

The state-dependent setup considered in this subsection can also be modelled as a DMBC with generalized feedback~\cite{wigger}.  Consequently, the Shayevitz-Wigger achievable region for generalized feedback \cite{wigger}\footnote{The Shayevitz-Wigger region for generalized feedback \cite{wigger} is characterized by constraints as shown in \eqref{eq:inner} but where in some places the outputs $Y_1$ and $Y_2$ have to be replaced by the generalized feedback output~$\tilde{Y}$.} (see in particular also \cite[Corollary~1]{kimchia}) applies. 

By similar arguments as used in the previous two subsections, it can be shown that  for state-dependent DMBCs with causal and strictly-causal state-information at the receivers and the transmitter our new achievable region $\set{R}_{\textnormal{proc.}}^{\textnormal{state}}$ improves over a restricted version of the Shayevitz-Wigger region for generalized feedback when in the latter the choice of auxiliaries is limited to~\eqref{eq:choice1}.

\subsection{The Kim-Chia-El Gamal region for state-dependent DMBCs and our regions  $\set{R}_{\textnormal{proc.}}^{\textnormal{state}}$ and $\set{R}_{\textnormal{relay,bw}}$} \label{Sec:Relations_Kim}
 
{Consider the class of \emph{deterministic} state-dependent DMBCs where  the two outputs $Y_1$ and $Y_2$ are deterministic functions of the input $X$ and the random state $S$.  As in the previous Subsection we assume that the transmitter has strictly-causal state-information and both receivers have causal state-information.

 Kim, Chia, and El Gamal \cite{kimchia} evaluated the maximum sum-rate achieved by the Shayevitz-Wigger scheme with generalized feedback \cite{wigger} for deterministic state-dependent DMBCs {when the auxiliary random variables of the Shayevitz-Wigger scheme are restricted to one of the following two choices: }
\begin{itemize} 
\item Choice~1 (coded time-sharing):
\begin{subequations}\label{eq:firstchoiceT}
\begin{equation}
Q=\begin{cases} 
0 & \textnormal{w. p. }~~1-2p \\
1 &  \textnormal{w. p. }~~p\\
2 &  \textnormal{w. p. }~~ p
\end{cases},
\end{equation}
\begin{equation}
V_0=V_1=V_2= \begin{cases} \emptyset &   \textnormal{if } Q=0\\ 
Y_2& \textnormal{if } Q=1\\
Y_1 & \textnormal{if } Q=2
\end{cases},
\end{equation}
and 
\begin{equation}
X= \begin{cases} 
U_0 &   \textnormal{if } Q=0\\ 
U_1& \textnormal{if } Q=1\\
U_2 & \textnormal{if } Q=2
\end{cases}
\end{equation}
\end{subequations}
for $0\leq p\leq 0.5$ and some pmf $P_{U_0U_1U_2}=P_{U_0}P_{U_1}P_{U_2}$.  \\
\item Choice~2 (randomized  superposition coding): 
\begin{subequations}\label{eq:secondchoiceT}
\begin{equation}
Q=\begin{cases} 1 &  \textnormal{w. p.}~~1/2\\
2 &  \textnormal{w. p.} ~~1/2
\end{cases},
\end{equation}
\begin{equation}
V_0=V_1=V_2= \begin{cases}
Y_2& \textnormal{if } Q=1\\
Y_1 & \textnormal{if } Q=2
\end{cases},
\end{equation}
and 
\begin{equation}
X= \begin{cases} 
U_1& \textnormal{if } Q=1\\
U_2 & \textnormal{if } Q=2
\end{cases}
\end{equation}
\end{subequations}
for some  pmf  $P_{U_0U_1U_2}=P_{U_0}P_{U_1|U_0} P_{U_2|U_0}$.  
\end{itemize}
Choice~1 essentially results in a coded time-sharing scheme, and was shown \cite{kimchia} to be optimal for  some deterministic state-dependent DMBCs, e.g., the state-dependent broadcast erasure channels \cite{greek, maddah} and the finite field deterministic channel \cite{maddah}. Choice~2 results in a randomized superposition coding scheme and can achieve larger rate regions than Choice~1. For no state-dependent DMBC a better choice for the auxiliaries for the Shayevitz-Wigger region  is known.

Since Kim, Chia, and El Gamal's choices of auxiliaries in~\eqref{eq:firstchoiceT} and \eqref{eq:secondchoiceT} satisfy Condition~\eqref{eq:choice1}, by the discussion in the last paragraph of  the previous subsection~\ref{sec:state},  our new region $\set{R}_{\textnormal{proc.}}^{\textnormal{state}}$ recovers  the sum-rates obtained by Kim-Chia-El Gamal.
}
 
\subsection{The Maddah-Ali\&Tse DoF region and our regions $\set{R}_{\textnormal{proc.}}^{\textnormal{state}}$ and $\set{R}_{\textnormal{relay,bw}}$} \label{Sec:Relations_David}

As pointed out in \cite{kimchia}, the celebrated Maddah-Ali \& Tse scheme \cite{maddah} for the two-user i.i.d. fading BC with strictly causal state-information is a special case of the Shayevitz-Wigger scheme for generalized feedback, when the scheme is extended to real-valued alphabets and {specialized to the choice of auxiliaries in \eqref{eq:firstchoiceT}}. 
Since their choice also satisfies \eqref{eq:choice12}, by the discussion at the end of Section~\ref{sec:state}, their DoF result for the two-user i.i.d. fading BC is also included in our new achievable region $\set{R}_{\textnormal{proc.}}^{\textnormal{state}}$. The same holds for the capacity of Maddah-Ali \& Tse's deterministic approximation of the i.i.d. fading BC, which is a deterministic state-dependent DMBC.

Various works have improved and extended Maddah-Ali \& Tse's DoF result. For example, Yang, Kobayashi, Gesbert, and Yi \cite{YangGesbert} modified and improved Maddah-Ali \& Tse's scheme to apply to the more general setup where the transmitters also obtain imperfect (rate-limited) \emph{causal} state-information. They showed that, under some mild assumptions, their improved scheme achieves the optimal DoF region for arbitrary stationary and ergodic fading processes. In this sense, they could bridge the gap between Maddah-Ali \& Tse's setup  and a setup where the transmitter learns the state causally (and not only strictly causally).

Chen and Elia \cite{ChenEliaAxiv} proposed a coding scheme for an even more general setup where the transmitter accumulates state-information about each fading sample over time. They were able to determine the optimal DoF region under some mild assumptions. 

Both the Yang et al. scheme \cite{YangGesbert} and the Chen and Elia scheme  \cite{ChenEliaAxiv} are related to the Shayevitz-Wigger scheme with a choice of auxiliaries as in \eqref{eq:choice2}, and thus also to our new region $\set{R}_{\textnormal{proc.}}^{\textnormal{state}}$. See \cite[Chapter~4]{Youlong'thesis} for a more detailed discussion.

}

\subsection{Comparison with noisy network coding}
Our  coding schemes are reminiscent of the compress-and-forward relay strategy \cite{elgamalcover79} or the noisy network coding for general networks \cite{Lim'11, Kramer'13} in the sense that the two receivers compress their channel outputs and send the compression indices over the feedback links.  The operations at the transmitter are however very different from noisy network coding. On one hand, we use Marton coding  to send independent private and common messages to the two receivers. On the other hand,  the transmitter either reconstructs the quantized version of the receivers' outputs (our type-II scheme) or it simply relays the compression messages (our type-I schemes). In particular, as explained in Section~\ref{sec:ExNoisyFB}, when the feedback links are noisy, the transmitter first decodes the compression messages sent over the feedback links and then operates on these compression messages. In noisy network coding the transmitter would  compress and forward the observed noisy feedback signals, without attempting to recover the transmitted compression messages.

}

\section*{Acknowledgment}
The authors would like to thank R. Timo for helpful discussions. 

\appendices

\section{Analysis of Scheme~IA (Theorem~\ref{theosw})}\label{sec:analysis_relayb}
By the symmetry of our code construction, the probability of error does not depend on the realizations of $M_{c,i,b}$, $M_{p,i,b}$, $K_{i,b}$, $J_{i,b}$, $M_{\FB,i,b}$,  for $i\in\{1,2\}$ and $b\in\{1,\ldots, B\}$. To simplify exposition we  therefore assume that  $M_{c,i,b}=M_{p,i,b}=K_{i,b}=J_{i,b}=M_{\FB,i,b}=1$ for all $i\in\{1,2\}$ and $b\in\{1,\ldots, B\}$. 
Under this assumption, an error  occurs if, and only if,  for some $b\in \{1,\ldots,B\}$, 
\[( \hat{M}_{p,1, b}, \hat{M}_{p,2, b}, \hat{M}_{c,1,b}^{(1)}, \hat{M}_{c,2,b}^{(2)}) \neq (1,1,1,1).\]

Epsilon
For each $b\in\{1,\ldots, B\}$, let $\Epsilon_b$ denote the event that in our coding scheme  at least one of the following holds for $i\in\{1,2\}$: 
\begin{itemize}
\item $\hat{J}_{i,b-1} \neq 1$;
\item $\hat{K}_{i,b-1} \neq 1$; 
\item $\hat{M}_{\FB,i,b-1}\neq  1$;
\item $\hat{M}_{p,i,b-1} \neq  1$; 
\item $\hat{{M}}_{c,b}^{(i)} \neq (1,1)$;
\item There is no pair  $(k_{1,b},k_{2,b})\in\set{K}_1\times\set{K}_2$   that satisfies \begin{IEEEeqnarray*}{rCl}
 \big(&&U_{0,b}^n(\mathbf{1}_{[4]}), U_{1,b}^n(1,k_{1,b}|\mathbf{1}_{[4]}), U_{2,b}^n(1,k_{2,b}|\mathbf{1}_{[4]})\big)\nonumber\\
 &&\hspace{4cm}\in \set{T}_{\varepsilon/16}^{(n)}(P_{U_0U_1U_2})
 \end{IEEEeqnarray*}
  \item $\big(U^n_{0,b-1}(\mathbf{1}_{[4]}), U^n_{1,b-1}(1, 1|\mathbf{1}_{[4]}), U_{2,b-1}^n(1,1,|\mathbf{1}_{[4]}),$\\ $ Y_{1,b-1}^n, Y_{2,b-1}^n \big) \notin \set{T}_{\varepsilon/12}^{(n)}(P_{U_0U_1{U_2}Y_1Y_2})$
 \item There is no pair $({m}_{\FB,i,b}, j_{i,b})\in  \set{L}_i\times \set{J}_{i}$   that satisfies\[ \big( \tilde{Y}_{i,b}^n(m_{\FB,i,b}, j_{i,b}| \mathbf{1}_{[4]}),U_{0,b}^n(\mathbf{1}_{[4]}), Y_{i,b}^n\big)\in \set{T}_{\varepsilon/4}^{(n)}(P_{\tilde{Y}_iU_0Y_i}).\]
\end{itemize}
Then, 
\begin{equation}\label{eq:proberror1b}
P_e^{(N)} \leq  \Prv{ \bigcup_{b=1}^{B+1} \Epsilon_b}  \leq {\sum_{b=2}^{B+1} \Prv{\Epsilon_b | \Epsilon_{b-1}^{c}} + \Prv{\Epsilon_{1}}}.
\end{equation}
In the following we analyze the probabilities of these events averaged over the random code construction. In particular, we shall identify conditions such that  for each {$b\in\{2,\ldots, B+1\}$}, the probability {$\Prv{\Epsilon_b | \Epsilon_{b-1}^{c}}$} tends to 0 as $n\to \infty$. Similar arguments can be used to show that under the same conditions also  {$\Prv{\Epsilon_{1}} \to 0$} as $n\to \infty$. Using standard arguments one can then conclude that there must exist a deterministic code for which the probability of error $P_e^{(N)}$ tends to 0 as $N\to \infty$ when the mentioned conditions are satisfied.

Fix {$b\in\{2,\ldots, B+1\}$} and $\varepsilon>0$, and define 
the following events. 
\begin{itemize}
\item Let $\Epsilon_{0,b}$ be the event that there is no pair  $(k_{1,b},k_{2,b})\in\set{K}_1\times\set{K}_2$  that satisfies
 \begin{IEEEeqnarray*}{rCl}
 \big(&&U_{0,b}^n(\mathbf{1}_{[4]}), U_{1,b}^n(1,k_{1,b}|\mathbf{1}_{[4]}), U_{2,b}^n(1,k_{2,b}|\mathbf{1}_{[4]})\big)\nonumber\\
 &&\quad\quad\in \set{T}_{\varepsilon/16}^{(n)}(P_{U_0U_1U_2}).
 \end{IEEEeqnarray*}
 {By the Covering Lemma \cite{book:gamal}, $\Pr(\Epsilon_{0,b})$ tends to 0 as $n\to \infty$ if}
 \begin{IEEEeqnarray}{rCl} \label{conMarton1b}
R'_1+R'_2 \geq I(U_1;U_2|U_0)+\delta(\varepsilon),
\end{IEEEeqnarray}
where throughout this section $\delta(\varepsilon)$ stands for some function that tends to 0 as $\varepsilon\to 0$.

\item Let $\Epsilon_{1,b}$ be the event that 
\begin{IEEEeqnarray*} {rCl}\lefteqn{
\big(U^n_{0,b}(\mathbf{1}_{[4]}), U^n_{1,b}(1, 1|\mathbf{1}_{[4]}), U_{2,b}^n(1,1,|\mathbf{1}_{[4]}), Y_{1,b}^n, Y_{2,b}^n \big) }\qquad \nonumber \\&&\hspace{4cm}\notin \set{T}_{\varepsilon/12}^{(n)}(P_{U_0U_1{U_2}Y_1Y_2}).
\end{IEEEeqnarray*}
{Since the channel is memoryless, by  the law of large numbers, $\Pr(\Epsilon_{1,b}|\Epsilon^{c}_{0,b})$ tends to 0 as $n\to \infty$.}

\item   Let $\Epsilon_{2,1,b}$ be the event that there is no tuple $({\hat{{m}}^{(1)}_{c,b}}, \hat{m}_{\FB,2,b-1})\in  \set{M}_c\times \set{L}_{2}$ that is not equal to $(\mathbf{1}_{[2]},1)$ and  that satisfies
\begin{IEEEeqnarray*} {rCl}
\big( U_{0,b}^n({\hat{{m}}^{(1)}_{c,b}}, 1, \hat{m}_{\FB,2,b-1}),Y_{1,b}^n\big)\in \set{T}_{\varepsilon/8}^{(n)}(P_{U_0Y_1}).
\end{IEEEeqnarray*}
{By the {Packing} Lemma, $\Pr(\Epsilon_{2,1,b}|\Epsilon^{c}_{1,b})$ tends to 0 as $n\to \infty$, if}
\begin{IEEEeqnarray}{rCl} 
\tilde{R}_2 + R_{c,1}+R_{c,2} \leq I(U_0;Y_1)+\delta(\varepsilon).
\end{IEEEeqnarray}

\item   Let $\Epsilon_{2,2,b}$ be the event that there is no tuple $({\hat{{m}}^{(2)}_{c,b}}, \hat{m}_{\FB,1,b-1})\in  \set{M}_c\times \set{L}_{1}$ with $({\hat{{m}}^{(2)}_{c,b}}, \hat{m}_{\FB,1,b-1})$ not equal to $(\mathbf{1}_{[2]},1)$  that satisfies
\begin{IEEEeqnarray*} {rCl}
\big( U_{0,b}^n({\hat{{m}}^{(2)}_{c,b}}, \hat{m}_{\FB,1,b-1},1),Y_{2,b}^n\big)\in \set{T}_{\varepsilon/8}^{(n)}(P_{U_0Y_2}).
\end{IEEEeqnarray*}
{By  the Packing Lemma, $\Pr(\Epsilon_{2,2,b}|\Epsilon^{c}_{1,b})$ tends to 0 as $n\to \infty$, if}
\begin{IEEEeqnarray}{rCl} 
\tilde{R}_1 + R_{c,1}+R_{c,2} \leq I(U_0;Y_2)+\delta(\varepsilon).
\end{IEEEeqnarray}

\item Let $\Epsilon_{3,1,b}$ be the event that 
\begin{IEEEeqnarray*}{rCl}
&&\big(U^n_{0,b-1}(\mathbf{1}_{[4]}), U^n_{1,b-1}(1, 1|\mathbf{1}_{[4]}),\nonumber\\&& \qquad\tilde{Y}_{2,b-1}^n(1,1), Y_{1,b-1}^n\big)\notin \set{T}_{{\varepsilon/2}}^{(n)}(P_{U_0U_1\tilde{Y}_2Y_1}).
\end{IEEEeqnarray*}
{By the Markov Lemma \cite{book:gamal}, $\Pr(\Epsilon_{3,1,b}|  \Epsilon^{c}_{b-1})$ tends to 0 as $n\to \infty$.}

\item Let $\Epsilon_{3,2,b}$ be the event that 
\begin{IEEEeqnarray*}{rCl}
&&\big(U^n_{0,b-1}(\mathbf{1}_{[4]}), U^n_{2,b-1}(1, 1|\mathbf{1}_{[4]}), \nonumber\\&&\qquad\tilde{Y}_{1,b-1}^n(1,1), Y_{2,b-1}^n\big)\notin \set{T}_{{\varepsilon/2}}^{(n)}(P_{U_0U_2\tilde{Y}_1Y_2}).
\end{IEEEeqnarray*}
 {By the Markov Lemma, $\Pr(\Epsilon_{3,2,b}| \Epsilon^{c}_{b-1})$ tends to 0 as $n\to \infty$. }

\item Let $\Epsilon_{4,1,b}$ be the event that there exists a tuple  $(\hat{m}_{p,1,b-1},\hat{k}_{1,b-1}, \hat{j}_{2,b-1}) \in\set{M}_{p,1}\times\set{K}_1\times\set{J}_2$ not equal to the all-one tuple  and  that satisfies
\begin{IEEEeqnarray*}{rCl}
&&\big(U^n_{0,b-1}(\mathbf{1}_{[4]}),U^n_{1,b-1}(\hat{m}_{p,1,b-1},\hat{k}_{1,b-1}|\mathbf{1}_{[4]}),\nonumber\\
&& ~\quad \tilde{Y}_{2,b-1}^n(1,\hat{j}_{2,b-1}|\mathbf{1}_{[4]}), Y_{1,b-1}^n\big)\in \set{T}_{\varepsilon}^{(n)}(P_{U_0U_1\tilde{Y}_2Y_1}).
\end{IEEEeqnarray*}
{By the Packing Lemma, $\Pr(\Epsilon_{4,1,b}|\Epsilon_{3,1,b}^c)$ tends to zero as $n\to \infty$, if}
\begin{IEEEeqnarray}{rCl}
\hat{R}_2& \leq& I(\tilde{Y}_2;U_1,Y_1|U_0)-\delta(\varepsilon)\\
R_{p,1}+R'_1&\leq& I(U_1;Y_1,\tilde{Y}_2|U_0)-\delta(\varepsilon)\\
R_{p,1}+R'_1+\hat{R}_2&\leq &I(U_1;Y_1,\tilde{Y}_2|U_0)\nonumber\\
&&+I(\tilde{Y}_2;Y_1|U_0)-\delta(\varepsilon).\label{error61np1b}
\end{IEEEeqnarray}

\item Let $\Epsilon_{4,2,b}$ be the event that there exists a tuple  $(\hat{m}_{p,2,b-1},\hat{k}_{2,b-1}, \hat{j}_{1,b-1}) \in\set{M}_{p,2}\times\set{K}_2\times\set{J}_1$ not equal to the all-one tuple  and  that satisfies
\begin{IEEEeqnarray*}{rCl}
&&\big(U^n_{0,b-1}(\mathbf{1}_{[4]}),U^n_{2,b-1}(\hat{m}_{p,2,b-1},\hat{k}_{2,b-1}|\mathbf{1}_{[4]}),\nonumber\\
&& \quad \tilde{Y}_{1,b-1}^n(1,\hat{j}_{1,b-1}|\mathbf{1}_{[4]}), Y_{2,b-1}^n\big)\in \set{T}_{{\varepsilon}}^{(n)}(P_{U_0U_2\tilde{Y}_1Y_2}).
\end{IEEEeqnarray*}
 {By  the Packing Lemma, $\Pr(\Epsilon_{4,2,b}|\Epsilon_{3,2,b}^c)$ tends to zero as $n\to \infty$, if}
\begin{IEEEeqnarray}{rCl}
\hat{R}_1& \leq& I(\tilde{Y}_1;U_2,Y_2|U_0)-\delta(\varepsilon)\\
R_{p,2}+R'_2&\leq& I(U_2;Y_2,\tilde{Y}_1|U_0)-\delta(\varepsilon)\quad \\
R_{p,2}+R'_2+\hat{R}_1&\leq &I(U_2;Y_2,\tilde{Y}_1|U_0)\nonumber\\
&&+I(\tilde{Y}_1;Y_2|U_0)-\delta(\varepsilon).\label{error62np1b}
\end{IEEEeqnarray}

\item For $i\in\{1,2\}$, let $\Epsilon_{5,i,b}$ be the event that there is no pair $({m}_{\FB,i,b}, j_{i,b})\in  \set{L}_i\times \set{J}_{i}$   that satisfies
\begin{IEEEeqnarray*} {rCl}
\big( \tilde{Y}_{i,b}^n(m_{\FB,i,b}, j_{i,b}| \mathbf{1}_{[4]}),U_{0,b}^n(\mathbf{1}_{[4]}), Y_{i,b}^n\big)\in \set{T}_{{\varepsilon/4}}^{(n)}(P_{\tilde{Y}_iU_0Y_i}).
\end{IEEEeqnarray*}
 {By   the Covering Lemma, $\Pr(\Epsilon_{5,i,b}|\Epsilon^{c}_{1,b})$ tends to 0 as $n\to \infty$, if}
\begin{IEEEeqnarray}{rCl}\label{eq:fbconSLWdm}
\tilde{R}_i +\hat{R}_i \geq I(\tilde{Y}_i;Y_i|U_0)+\delta(\varepsilon).
\end{IEEEeqnarray}

\end{itemize}
{Whenever the event $\Epsilon_{b-1}^{c}$ occurs but none of the  events $\{ \Epsilon_{0,b}, \Epsilon_{1,b}, \Epsilon_{2,i,b}, \Epsilon_{3,i,b}, \Epsilon_{3,i,b}, \Epsilon_{4,i,b},\Epsilon_{5,i,b}\}$ above, for $i=1,2$, then $\Epsilon^c_b$}. Therefore,
{\begin{IEEEeqnarray*}{rCl}
\lefteqn{
\Prv{\Epsilon_b|\Epsilon_{b-1}^c} }\nonumber \\& \leq& \textnormal{Pr}\Big[  \Epsilon_{0,b} \!\cup \!\Epsilon_{1,b}\! \cup \!\bigcup_{i=1}^2 \Big( \Epsilon_{2,i,b} \!\cup\! \Epsilon_{3,i,b}\cup \Epsilon_{4,i,b}\cup \Epsilon_{5,i,b}\Big)\Big|\Epsilon_{b-1}^c\Big] \nonumber \\
& \leq& \Prv{ \Epsilon_{0,b}|\Epsilon^c_{b-1}}+\Prv{\Epsilon_{1,b} |\Epsilon_{0,b}^c,\Epsilon_{b-1}^c}  \nonumber\\&& \quad+ \sum_{i=1}^2\Big(  \Prv{\Epsilon_{2,i,b} |\Epsilon_{1,b}^c,\Epsilon_{b-1}^c} \!+\! \Prv{ \Epsilon_{3,i,b}| \Epsilon^{c}_{b\!-\!1} }  \nonumber\\&& \quad +\Prv{ \Epsilon_{4,i,b}|\Epsilon_{3,i,b}^c, \Epsilon^{c}_{b-1}} \! +\!\Prv{ \Epsilon_{5,i,b}|\Epsilon^{c}_{1,b},\Epsilon_{b-1}^c} \Big)\nonumber\\
& =& \Prv{ \Epsilon_{0,b}}+\Prv{\Epsilon_{1,b} |\Epsilon_{0,b}^c} \nonumber\\&& \quad+ \sum_{i=1}^2\Big(  \Prv{\Epsilon_{2,i,b} |\Epsilon_{1,b}^c} \!+\! \Prv{ \Epsilon_{3,i,b}| \Epsilon^{c}_{b-1} }  \nonumber\\&& \quad\qquad +\Prv{ \Epsilon_{4,i,b}|\Epsilon_{3,i,b}^c} \! +\!\Prv{ \Epsilon_{5,i,b}|\Epsilon^{c}_{1,b}} \Big). 
\label{eq:probe1b}
\end{IEEEeqnarray*}}
{The last equality holds because the channel is memoryless and the codebooks employed in blocks $b-1$ and $b$ are drawn independently. 
As explained in the previous paragraphs, the remaining terms in the last three lines tend to 0 as $n\to \infty$, if Constraints~\eqref{conMarton1b}--\eqref{eq:fbconSLWdm} are satisfied. }
Thus, by~\eqref{eq:proberror1b} and \eqref{eq:probe1b} we conclude that the probability of error $P_e^{(N)}$ (averaged over all code constructions) vanishes as $n\to\infty$ if Constraints~\eqref{conMarton1b}--\eqref{eq:fbconSLWdm} hold. Letting $\varepsilon \to 0$, 
 we obtain that the probability of  error can be made to tend to 0 as $n\to \infty$ whenever
 \begin{subequations}\label{method1b}
\begin{IEEEeqnarray}{rCl}
R'_1+R'_2 &>& I(U_1;U_2|U_0) \\
\tilde{R}_2 \!+\! R_{c,1}\!+\!R_{c,2} &<& I(U_0;Y_1)\\
\tilde{R}_1 \!+\! R_{c,1}\!+\!R_{c,2} &< &I(U_0;Y_2)\\
\hat{R}_1 & < & I(\tilde{Y}_1;U_2,Y_2|U_0)\\
\hat{R}_2 & < & I(\tilde{Y}_2;U_1,Y_1|U_0)\\
R_{p,1}+R_1'&<&I(U_1;Y_1,\tilde{Y}_2|U_0)\\
R_{p,2}+R_2'&>&I(U_2;Y_2,\tilde{Y}_1|U_0)\\
R_{p,1}\!+\!R'_1\!+\!\hat{R}_2&< &I(U_1;Y_1,\tilde{Y}_2|U_0)\!+\!I(\tilde{Y}_2;Y_1|U_0)~\\
R_{p,2}\!+\!R'_2\!+\!\hat{R}_1&< &I(U_2;Y_2,\tilde{Y}_1|U_0)\!+\!I(\tilde{Y}_1;Y_2|U_0)~\\
\hat{R}_1+\tilde{R}_1& > &I(\tilde{Y}_1;Y_1|U_0)\\
\hat{R}_2+\tilde{R}_2 &>& I(\tilde{Y}_2;Y_2|U_0).
\end{IEEEeqnarray}
Moreover, the feedback-rate constraints~\eqref{consFB0} impose that:
\begin{IEEEeqnarray}{rCl}
\tilde{R}_1 & \leq & R_{\FB,1}\\
\tilde{R}_2 & \leq & R_{\FB,2}.
\end{IEEEeqnarray}
\end{subequations}
Applying the Fourier-Motzkin elimination algorithm to these constraints, we obtain the desired result in Theorem~\ref{theosw} with the additional constraint that 
\begin{IEEEeqnarray}{rCl}\label{eq:Martonextra1}
\Delta_1+\Delta_2 -I(U_1;U_2|U_0) &\geq& 0 
\IEEEeqnarraynumspace
\end{IEEEeqnarray}
Notice that we can ignore  Constraint \eqref{eq:Martonextra1} because for any tuple $(U_0,U_1,U_2, X, {Y}_1,Y_2, \tilde{Y}_1,\tilde{Y}_2)$ that violates~\eqref{eq:Martonextra1}, the region defined by the constraints in Theorem~\ref{theosw} is contained in the time-sharing region.  


\section{Analysis of the  Scheme~IB (Theorem~\ref{theobw})}\label{sec:analysis_relay}
An error occurs whenever
\[\hat{M}_{1,b}\neq M_{1,b}~\text{or}~ \hat{M}_{2,b}\neq M_{2,b},~\text{for some}~ b\in \{1,\ldots,B\}.\]
For each $b\in\{1,\ldots, B+1\}$, let $\Epsilon_b$ denote the event that in our coding scheme  at least one of the following holds for $i\in\{1,2\}$: 
\begin{IEEEeqnarray}{rCl}
\hat{J}_{i,b} \neq J_{i,b}\\
\hat{K}_{i,b} \neq K_{i,b}\\
\hat{M}_{\FB,i,b-1} \neq M_{\FB,i,b-1}\\
{\hat{M}_{p,i,b} \neq  M_{p,i,b}}\\
{\hat{{M}}^{(i)}_{c,b} \neq {M}^{(i)}_{c,b}}.
\end{IEEEeqnarray}
Then, 
\begin{equation}\label{eq:proberror}
P_e^{(N)} \leq  \Prv{ \bigcup_{b=1}^{B+1} \Epsilon_b}  \leq \sum_{b=1}^B \Prv{\Epsilon_b | \Epsilon_{b+1}^{c}} + \Prv{\Epsilon_{B+1}}.
\end{equation}
In the following we analyze the probabilities of these events averaged over the random code construction. In particular, we shall identify conditions such that  for each $b\in\{1,\ldots, B\}$, the probability $\Prv{\Epsilon_b | \Epsilon_{b+1}^{c}}$ tends to 0 as $n\to \infty$. Similar arguments can be used to show that under the same conditions also  $\Prv{\Epsilon_{B+1}} \to 0$ as $n\to \infty$. Using standard arguments one can then conclude that there must exist a deterministic code for which the probability of error $P_e^{(N)}$ tends to 0 as $N\to \infty$ when the mentioned conditions are satisfied.

Fix $b\in\{1,\ldots, B\}$ and $\varepsilon>0$. By the symmetry of our code construction, the probability $\Prv{\Epsilon_b | \Epsilon_{b+1}^{c}}$ does not depend on the realization of $M_{c,i,b}$, $M_{p,i,b}$, $K_{i,b}$, $J_{i,b}$, $M_{\FB,i,b}$, $M_{\FB,i,b-1}$, for $i\in\{1,2\}$. To simplify exposition we  therefore assume that  $M_{c,i,b}=M_{p,i,b}=K_{i,b}=J_{i,b}=M_{\FB,i,b}=M_{\FB,i,b-1}=1$. 

Define  the following events. 
\begin{itemize}
\item Let $\Epsilon_{0,b}$ be the event that there is no pair  $(k_{1,b},k_{2,b})\in\set{K}_1\times\set{K}_2$  that satisfies
 \begin{IEEEeqnarray*}{rCl}
 \big(&&U_{0,b}(\mathbf{1}_{[4]}), U_{1,b}^n(1,k_{1,b}|\mathbf{1}_{[4]}), U_{2,b}^n(1,k_{2,b}|\mathbf{1}_{[4]})\big)\nonumber\\
 &&\quad\quad\in \set{T}_{\varepsilon/16}^{(n)}(P_{U_0U_1U_2}).
 \end{IEEEeqnarray*}
 By  the Covering Lemma, $\Pr(\Epsilon_{0,b})$ tends to 0 as $n\to \infty$, if
 \begin{IEEEeqnarray}{rCl} \label{conMarton}
R'_1+R'_2 \geq I(U_1;U_2|U_0)+\delta(\varepsilon),
\end{IEEEeqnarray}
where throughout this section $\delta(\varepsilon)$ stands for some function that tends to 0 as $\varepsilon\to 0$.

\item Let $\varepsilon_{1,b}$ be the event that 
\begin{IEEEeqnarray*} {rCl}\lefteqn{
\big(U^n_{0,b}(\mathbf{1}_{[4]}), U^n_{1,b}(1, 1|\mathbf{1}_{[4]}, U_{2,b}^n(1,1,|\mathbf{1}_{[4]}), Y_{1,b}^n, Y_{2,b}^n \big) }\qquad \nonumber \\&&\hspace{4cm}\notin \set{T}_{\varepsilon/8}^{(n)}(P_{U_0U_1Y_2Y_1Y_2}).
\end{IEEEeqnarray*}
Since the channel is memoryless, according to the law of large numbers, $\Pr(\Epsilon_{1,b}|\Epsilon^{c}_{0,b})$ tends to 0 as $n\to \infty$.
%

\item For $i\in\{1,2\}$, let $\Epsilon_{2,i,b}$ be the event that there is no pair $({m}_{\FB,i,b}, j_{i,b})\in  \set{L}_i\times \set{J}_{i}$   that satisfies
\begin{IEEEeqnarray*} {rCl}
\big( \tilde{Y}_{i,b}^n(m_{\FB,i,b}, j_{i,b}),Y_{i,b}^n\big)\in \set{T}_{\varepsilon/4}^{(n)}(P_{\tilde{Y}_iY_i}).
\end{IEEEeqnarray*}
 By the Covering Lemma, $\Pr(\Epsilon_{2,i,b}|\Epsilon^{c}_{1,b})$ tends to 0 as $n\to \infty$ if
\begin{IEEEeqnarray}{rCl} \label{conR1fb}
\tilde{R}_i +\hat{R}_i \geq I(\tilde{Y}_i;Y_i)+\delta(\varepsilon).
\end{IEEEeqnarray}

\item Let $\Epsilon_{3,1,b}$ be the event that 
\begin{IEEEeqnarray*}{rCl}
\big(U^n_{0,b}(\mathbf{1}_{[4]}), &&U^n_{1,b}(1, 1|\mathbf{1}_{[4]}),\nonumber\\&& \tilde{Y}_{2,b}^n(1,1), Y_{1,b}^n\big)\notin \set{T}_{3\varepsilon/4}^{(n)}(P_{U_0U_1\tilde{Y}_2Y_1}).
\end{IEEEeqnarray*}
 By the Markov Lemma, $\Pr(\Epsilon_{3,1,b}|\Epsilon_{2,2,b}^c, \Epsilon_{1,b}^c)$ tends to 0 as $n\to \infty$.

\item Let $\Epsilon_{3,2,b}$ be the event that 
\begin{IEEEeqnarray*}{rCl}
\big(U^n_{0,b}(\mathbf{1}_{[4]}), &&U^n_{2,b}(1, 1|\mathbf{1}_{[4]}), \nonumber\\&&\tilde{Y}_{1,b}^n(1,1), Y_{2,b}^n\big)\notin \set{T}_{3\varepsilon/4}^{(n)}(P_{U_0U_2\tilde{Y}_1Y_2}).
\end{IEEEeqnarray*}
 By the Markov Lemma, $\Pr(\Epsilon_{3,2,b}|\Epsilon_{2,1,b}^c, \Epsilon_{1,b}^c)$ tends to 0 as $n\to \infty$.

\item Let $\Epsilon_{4,1,b}$ be the event that 
there exists a tuple  $(\hat{j}_{2,b},{{\hat{m}}^{(1)}_{c,b}},\hat{m}_{\FB,2,b-1},\hat{m}_{p,1,b},\hat{k}_{1,b})\in \set{J}_2\times\set{M}_c\times \set{L}_{2}\times \set{M}_{p,1}\times\set{K}_1$ not equal to the all-one tuple $(1, \mathbf{1}_{[2]}, 1, 1,1)$  and that satisfies
\begin{IEEEeqnarray*}{rCl}\lefteqn{
\Big(U^n_{0,b}({{\hat{m}}^{(1)}_{c,b}}, 1, \hat{m}_{\FB,2,b-1}),}\qquad \\ 
&&U^n_{1,b}(\hat{m}_{p,1,b},\hat{k}_{1,b}|{{\hat{m}}^{(1)}_{c,b}}, 1, \hat{m}_{\FB,2,b-1}),\nonumber\\
&&\hspace{1cm} \tilde{Y}_{2,b}^n(1,\hat{j}_{2,b}), Y_{1,b}^n\Big)\in \set{T}_\varepsilon^{(n)}(P_{U_0U_1\tilde{Y}_2Y_1}).
\end{IEEEeqnarray*}
By the Packing Lemma, we conclude that $\Pr(\Epsilon_{4,1,b}|\Epsilon_{3,1,b}^c)$ tends to zero as $n\to \infty$ if
\begin{IEEEeqnarray}{rCl}
&&\hat{R}_2\leq I(U_0,U_1,Y_1;\tilde{Y}_2|U_0)\!-\!\delta(\varepsilon)\nonumber\\
&&R_{p,1}+R'_1\leq I(U_1;Y_1,\tilde{Y}_2|U_0)-\delta(\varepsilon)\nonumber\\
&&R_{1}+R_{c,2}+\tilde{R}_2+R'_1\leq I(U_0,U_1;Y_1, \tilde{Y}_2)-\delta(\varepsilon)\nonumber\IEEEeqnarraynumspace\\
&&R_{1}\!+\!R_{c,2}\!+\!\tilde{R}_2\!+\!R'_1\!+\!\hat{R}_2\leq I(U_0,U_1;Y_1,\tilde{Y}_2)\nonumber\\
&&\qquad\qquad\qquad\quad\qquad\qquad+I(Y_1;\tilde{Y}_2)-\delta(\varepsilon)\nonumber\\
&&R_{p,1}+R'_1+\hat{R}_2\leq I(U_1;Y_1,\tilde{Y}_2|U_0)\nonumber\\
&&\qquad\qquad\qquad\quad\qquad\qquad+I(\tilde{Y}_2;Y_1,U_0)\!-\!\delta(\varepsilon).
\end{IEEEeqnarray}

\item Let $\Epsilon_{4,2,b}$ be the event that there exists a tuple  $(\hat{j}_{1,b},{{\hat{m}}^{(2)}_{c,b}},\hat{m}_{\FB,1,b-1},\hat{m}_{p,2,b},\hat{k}_{2,b})\in\set{J}_1\times \set{M}_c\times \set{L}_{1}\times \set{M}_{p,2}\times\set{K}_2$ not equal to the all-one tuple  and  that satisfies
\begin{IEEEeqnarray*}{rCl}\lefteqn{
\big(U^n_{0,b}({{\hat{m}}^{(2)}_{c,b}}, \hat{m}_{\FB,1, b-1}, 1),}\qquad \\ 
&&U^n_{1,b}(\hat{m}_{p,2,b},\hat{k}_{2,b}|{{\hat{m}}^{(2)}_{c,b}}, \hat{m}_{\FB,1, b-1}, 1),\nonumber\\
&&\hspace{1cm} \tilde{Y}_{1,b}^n(1,\hat{j}_{1,b}), Y_{2,b}^n\big)\in \set{T}_\varepsilon^{(n)}(P_{U_0U_2\tilde{Y}_1Y_2}).
\end{IEEEeqnarray*}
By the Packing Lemma, we conclude that $\Pr(\Epsilon_{4,2,b}|\Epsilon_{3,2,b}^c)$ tends to zero as $n\to \infty$ if
\begin{IEEEeqnarray}{rCl}
&&\hat{R}_1\leq I(U_0,U_2,Y_2;\tilde{Y}_1|U_0)\!-\!\delta(\varepsilon)\nonumber\\
&&R_{p,2}+R'_2\leq I(U_2;Y_2,\tilde{Y}_1|U_0)-\delta(\varepsilon)\nonumber\\
&&R_{2}+R_{c,1}+\tilde{R}_1+R'_2\leq I(U_0,U_2;Y_2, \tilde{Y}_1)-\delta(\varepsilon)\nonumber\\
&&R_{2}\!+\!R_{c,1}\!+\!\tilde{R}_1\!+\!R'_2\!+\!\hat{R}_1\leq I(U_0,U_2;Y_2,\tilde{Y}_1)\nonumber\\
&&\qquad\qquad\qquad\quad\qquad\qquad+I(Y_2;\tilde{Y}_1)-\delta(\varepsilon)\nonumber\\
&&R_{p,2}+R'_2+\hat{R}_1\leq I(U_2;Y_2,\tilde{Y}_1|U_0)\nonumber\\
&&\qquad\qquad\qquad\qquad\quad+I(\tilde{Y}_1;Y_2,U_0)\!-\!\delta(\varepsilon).\label{eq:e42b}
\end{IEEEeqnarray}

\end{itemize}
Whenever the event $\Epsilon_{b+1}^{c}$ occurs but none of the  events above, then $\Epsilon^c_b$. Therefore,
\begin{IEEEeqnarray}{rCl}
\lefteqn{
\Prv{\Epsilon_b|\Epsilon_{b+1}^c} }\nonumber \\& \leq&  \Prv{ \Epsilon_{0,b} \cup \Epsilon_{1,b} \cup\bigcup_{i=1}^2 \Big( \Epsilon_{2,i,b} \cup \Epsilon_{3,i,b}\cup \Epsilon_{4,i,b}\Big)\Big|\Epsilon_{b+1}^c} \nonumber \\
& \leq& \Prv{ \Epsilon_{0,b}|\Epsilon_{b+1}^c}+\Prv{ \Epsilon_{3,1,b}|\Epsilon_{1,b}^c, \!\Epsilon_{2,2,b}^c,\!\Epsilon_{b+1}^c}   \nonumber \\ &&  \quad+  \Prv{\Epsilon_{1,b} |\Epsilon_{0,b}^c,\Epsilon_{b+1}^c} \!+\!\Prv{ \Epsilon_{3,2,b}|\Epsilon_{1,b}^c,\! \Epsilon_{2,1,b}^c,\!\Epsilon_{b+1}^c } \nonumber \\ &&  \quad +\sum_{i=1}^2\Big(  \Prv{\Epsilon_{2,i,b} |\Epsilon_{1,b}^c ,\Epsilon_{b+1}^c}\! +\!\Prv{ \Epsilon_{4,i,b}|\Epsilon_{3,i,b}^c ,\Epsilon_{b+1}^c }  \Big) 
\nonumber \\
& = &  \Prv{ \Epsilon_{0,b}}+\Prv{\Epsilon_{1,b} |\Epsilon_{0,b}^c}   \nonumber \\ &&  \quad + \Prv{ \Epsilon_{3,1,b}|\Epsilon_{1,b}^c, \Epsilon_{2,2,b}^c}  +\Prv{ \Epsilon_{3,2,b}|\Epsilon_{1,b}^c, \Epsilon_{2,1,b}^c } \nonumber \\ &&  \quad +\sum_{i=1}^2\Big(  \Prv{\Epsilon_{2,i,b} |\Epsilon_{1,b}^c } +\Prv{ \Epsilon_{4,i,b}|\Epsilon_{3,i,b}^c}  \Big), \label{eq:probe}\IEEEeqnarraynumspace
\end{IEEEeqnarray} 
where the last equality follows because the channel is memoryless and the codebooks for blocks $b$ and $b+1$ have been generated independently.
As explained in the previous paragraphs, each of the terms in the last three lines tends to 0 as $n\to \infty$, if Constraints~\eqref{conMarton}--\eqref{eq:e42b} are satisfied. 
Thus, by~\eqref{eq:proberror} and \eqref{eq:probe} we conclude that the probability of error $P_e^{(N)}$ (averaged over all code constructions) vanishes as $n\to\infty$ if constraints~\eqref{conMarton}--\eqref{eq:e42b} hold. Letting $\varepsilon \to 0$, 
 we obtain that the probability of  error can be made to tend to 0 as $n\to \infty$ whenever
 \begin{subequations}\label{method2}
\begin{IEEEeqnarray}{rCl}
R'_1+R'_2 &>& I(U_1;U_2|U_0) \\
\hat{R}_1+\tilde{R}_1& > &I(\tilde{Y}_1;Y_1)\\
\hat{R}_2+\tilde{R}_2 &>& I(\tilde{Y}_2;Y_2)\\
\hat{R}_1&<& I(U_0,U_2,Y_2;\tilde{Y}_1|U_0)\\
\hat{R}_2 & < & I(U_0,U_1,Y_1;\tilde{Y}_2|U_0)\qquad\\
R_{p,1}+R_1'&<&I(U_1;Y_1,\tilde{Y}_2|U_0)\\
R_{p,2}+R_2'&<&I(U_2;Y_2,\tilde{Y}_1|U_0)\\
R_{1}+R_{c,2}+\tilde{R}_2+R'_1
&<& I(U_0,U_1;Y_1,\tilde{Y}_2)\\
R_{2}+R_{c,1}+\tilde{R}_1+R'_2
&<& I(U_0,U_2;Y_2,\tilde{Y}_1)\\
R_{1}\!+\!R_{c,2}\!+\!\tilde{R}_2\!+\!R'_1\!+\!\hat{R}_2&<& I(U_0,U_1;Y_1,\tilde{Y}_2)\nonumber\\
&&~+I(Y_1;\tilde{Y}_2) \\
R_{2}\!+\!R_{c,1}\!+\!\tilde{R}_1\!+\!R'_2\!+\!\hat{R}_1&<& I(U_0,U_2;Y_2,\tilde{Y}_1)\nonumber\\
&&~+I(Y_2;\tilde{Y}_1)\\
R_{p,1}+R'_1+\hat{R}_2&<& I(U_1;Y_1,\tilde{Y}_2|U_0)\nonumber\\
&&~+I(\tilde{Y}_2;Y_1,U_0)\\
R_{p,2}+R'_2+\hat{R}_1&<& I(U_2;Y_2,\tilde{Y}_1|U_0)\nonumber\\
&&~+I(\tilde{Y}_1;Y_2,U_0).
\end{IEEEeqnarray}
Moreover, the feedback-rate constraints~\eqref{consFB0} impose that:
\begin{IEEEeqnarray}{rCl}
\tilde{R}_1 & \leq & R_{\FB,1}\\
\tilde{R}_2 & \leq & R_{\FB,2}.
\end{IEEEeqnarray}
\end{subequations}
Applying the Fourier-Motzkin elimination algorithm to these constraints, we obtain the desired result in Theorem~\ref{theobw} with the additional constraint that 
\begin{subequations}\label{eq:Martonextra2}
\begin{IEEEeqnarray}{rCl}
\Delta_1 &\geq& 0\label{eq:Martonextra2c}\\
\Delta_2 &\geq& 0 \label{eq:Martonextra2b}\\
\Delta_1+\Delta_2  &\geq& I(U_1;U_2|U_0). \label{eq:Martonextra2a}
\IEEEeqnarraynumspace
\end{IEEEeqnarray}
\end{subequations}
We can ignore  Constraint \eqref{eq:Martonextra2a} because for any tuple $(U_0,U_1,U_2, X, {Y}_1,Y_2, \tilde{Y}_1,\tilde{Y}_2)$ that violates~\eqref{eq:Martonextra2a}, the region defined by the constraints in Theorem~\ref{theobw} is contained in the time-sharing region. Constraint \eqref{eq:Martonextra2c}   can  also be ignored because for any tuple $(U_0,U_1,U_2, X, {Y}_1,Y_2, \tilde{Y}_1,\tilde{Y}_2)$ that violates~\eqref{eq:Martonextra2c}, the region defined by the constraints in Theorem~\ref{theobw} is contained in the region in Theorem~\ref{theobw}  for the choice $\tilde{Y}_2=\textnormal{const.}$, for which \eqref{eq:Martonextra2c} is always satisfied. Constraint \eqref{eq:Martonextra2b} can be ignored by analogous arguments. 

\section{Analysis of   Scheme~II (Theorem~\ref{theo5})}\label{sec:analysis_proc}

An error occurs whenever
\[\hat{M}_{1,b}\neq M_{1,b}~\text{or}~ \hat{M}_{2,b}\neq M_{2,b},~\text{for some}~ b\in \{1,\ldots,B\}.\]
For each $b\in\{1,\ldots, B+1\}$, let $\Epsilon_b$ denote the event that in our coding scheme  at least one of the following holds for $i\in\{1,2\}$: 
\begin{IEEEeqnarray}{rCl}
\hat{J}_{i,b} \neq J_{i,b}\\
\hat{K}_{i,b} \neq K_{i,b}\\
\hat{M}_{\FB,i,b} \neq M_{\FB,i,b}\\
{\hat{M}_{p,i,b} \neq  M_{p,i,b}}\\
{\hat{{M}}^{(i)}_{c,b} \neq {M}^{(i)}_{c,b}}
\end{IEEEeqnarray}
or when 
\begin{equation}
\hat{N}_{b-1}\neq N_{b-1}.
\end{equation}
Then, 
\begin{equation}\label{eq:proberrorSch2}
P_e^{(n)} \leq  \Prv{ \bigcup_{b=1}^{B+1} \Epsilon_b}  \leq \sum_{b=1}^B \Prv{\Epsilon_b | \Epsilon_{b+1}^{c}} + \Prv{\Epsilon_{B+1}}.
\end{equation}
In the following we analyze the probabilities of these events averaged over the random code construction. In particular, we shall identify conditions such that  for each $b\in\{1,\ldots, B\}$, the probability $\Prv{\Epsilon_b | \Epsilon_{b+1}^{c}}$ tends to 0 as $n\to \infty$. Similar arguments can be used to show that under the same conditions also  $\Prv{\Epsilon_{B+1}} \to 0$ as $n\to \infty$. Using standard arguments one can then conclude that there must exist a deterministic code for which the probability of error $P_e^{(N)}$ tends to 0 as $N\to \infty$ when the mentioned conditions are satisfied.

Fix $b\in\{1,\ldots, B\}$ and $\varepsilon>0$. By the symmetry of our code construction, the probability $\Prv{\Epsilon_b | \Epsilon_{b+1}^{c}}$ does not depend on the realizations of $N_{b-1}$, $N_{b}$, or $M_{c,i,b}$, $M_{p,i,b}$, $K_{i,b}$, $J_{i,b}$, $M_{\FB,i,b}$, for $i\in\{1,2\}$. To simplify exposition we  therefore assume that  for $i\in\{1,2\}$, $M_{c,i,b}=M_{p,i,b}=K_{i,b}=J_{i,b}=M_{\FB,i,b}=1$,  and $N_{b}=N_{b-1}=1$.

Define  the following events. 
\begin{itemize}
\item Let $\Epsilon_{0,b}$ be the event that there is no pair  $(k_{1,b},k_{2,b})\in\set{K}_1\times\set{K}_2$  that satisfies
 \begin{IEEEeqnarray*}{rCl}
 \big(&&U_{0,b}(\mathbf{1}_{[3]}), U_{1,b}^n(1,k_{1,b}|\mathbf{1}_{[2}), U_{2,b}^n(1,k_{2,b}|\mathbf{1}_{[3]})\big)\nonumber\\
 &&\quad\quad\in \set{T}_{\varepsilon/64}^{(n)}(P_{U_0U_1U_2}).
 \end{IEEEeqnarray*}
 By  the Covering Lemma, $\Pr(\Epsilon_{0,b})$ tends to 0 as $n\to \infty$ if
 \begin{IEEEeqnarray}{rCl} \label{conMartonSch2}
R'_1+R'_2 \geq I(U_1;U_2|U_0)+\delta(\varepsilon),
\end{IEEEeqnarray}
where throughout this section $\delta(\varepsilon)$ stands for some function that tends to 0 as $\varepsilon\to 0$.

\item Let $\Epsilon_{1,b}$ be the event that 
\begin{IEEEeqnarray*} {rCl}\lefteqn{
\big(U^n_{0,b}(\mathbf{1}_{[3]}), U^n_{1,b}(1, 1|\mathbf{1}_{[3]}), U_{2,b}^n(1,1,|\mathbf{1}_{[3]}), Y_{1,b}^n, Y_{2,b}^n \big) }\qquad \nonumber \\&&\hspace{4cm}\notin \set{T}_{\varepsilon/32}^{(n)}(P_{U_0U_1U_2Y_1Y_2}).
\end{IEEEeqnarray*}
Since the channel is memoryless, according to the law of large numbers, $\Pr(\Epsilon_{1,b}|\Epsilon^{c}_{0,b})$ tends to 0 as $n\to \infty$.

\item For $i\in\{1,2\}$, let $\Epsilon_{2,i,b}$ be the event that there is no pair $({m}_{\FB,i,b}, j_{i,b})\in  \set{L}_i\times \set{J}_{i}$   that satisfies
\begin{IEEEeqnarray*} {rCl}\label{compressY1Mar}
\big(\tilde{Y}_{i,b}^n(m_{\FB,i,b}, j_{i,b}),Y_{i,b}^n\big)\in \set{T}_{\varepsilon/16}^{(n)}(P_{\tilde{Y}_iY_i}).
\end{IEEEeqnarray*}
 By the Covering Lemma, $\Pr(\Epsilon_{2,i,b}|\Epsilon_{1,b}^c)$ tends to 0 as $n\to \infty$ if
\begin{IEEEeqnarray}{rCl} \label{conR1fb2}
\tilde{R}_i +\hat{R}_i \geq I(\tilde{Y}_i;Y_i)+\delta(\Epsilon).
\end{IEEEeqnarray}

\item Let $\Epsilon_{3,b}$ be the event that 
\begin{IEEEeqnarray*}{rCl}
\lefteqn{
\big(U^n_{0,b}(\mathbf{1}_{[3]}), U^n_{1,b}(1, 1|\mathbf{1}_{[3]}), U^n_{2,b}(1, 1|\mathbf{1}_{[3]}),}\nonumber \\
 & & \qquad  \hspace{2cm} \tilde{Y}_{1,b}^n(1,1), \tilde{Y}_{2,b}^n(1,1), Y_{1,b}^n, Y_{2,b}^n\big)\nonumber \\ & & \hspace{4cm}\notin \set{T}_{\varepsilon/6}^{(n)}(P_{U_0U_1U_2\tilde{Y}_1\tilde Y_2Y_1Y_2}).
\end{IEEEeqnarray*}
 By the Markov Lemma, $\Pr(\Epsilon_{3,b}|\Epsilon_{2,1,b}^c, \Epsilon_{2,2,b}^c,\Epsilon_{1,b}^c)$ tends to 0 as $n\to \infty$.

\item {Let $\Epsilon_{4,b}$ be the event that there is a pair of indices $\hat j_{1,b}\in\set{J}_1$ and $\hat j_{2,b}\in\set{J}_2$ not equal to the all-one pair $(1,1)$ and that satisfies
\begin{IEEEeqnarray*}{rCl}
\lefteqn{
\big(U^n_{0,b}(\mathbf{1}_{[3]}), U^n_{1,b}(1, 1|\mathbf{1}_{[3]}), U^n_{2,b}(1, 1|\mathbf{1}_{[3]}),}\nonumber \\
 & & \quad  \tilde{Y}_{1,b}^n(1,\hat{j}_{1,b}), \tilde{Y}_{2,b}^n(1,\hat{j}_{2,b})\big)\!\in\! \set{T}_{\varepsilon/4}^{(n)}(P_{U_0U_1U_2\tilde{Y}_1\tilde Y_2}).
\end{IEEEeqnarray*}
By the Packing Lemma, $\Pr(\Epsilon_{4,b}|\Epsilon_{3,b}^c)$ tends to 0 as $n
 \to \infty$, if}
 \begin{IEEEeqnarray}{rCl}
 \hat{R}_1 & \leq & I(U_0,U_1,U_2,\tilde{Y}_2;\tilde{Y}_1)-\delta(\varepsilon)\\
 \hat{R}_2 & \leq & I(U_0,U_1,U_2,\tilde{Y}_1;\tilde{Y}_2)-\delta(\varepsilon)\\ 
  \hat{R}_1 +\hat{R}_2& \leq & I(U_0,U_1,U_2;\tilde{Y}_1,\tilde{Y}_2)\!+\! I(\tilde{Y}_1;\tilde{Y}_2)\!-\!\delta(\varepsilon).\nonumber \\
 \end{IEEEeqnarray}
 
\item {Let $\Epsilon_{5,b}$ be the event that there is no index ${n}_{b}\in\set{N}$ that satisfies 
\begin{IEEEeqnarray*}{rCl}
\lefteqn{
\big(U^n_{0,b}(\mathbf{1}_{[3]}), U^n_{1,b}(1, 1|\mathbf{1}_{[3]}), U^n_{2,b}(1, 1|\mathbf{1}_{[3]}),}\nonumber \\
 & & \qquad \qquad \qquad  \tilde{Y}_{1,b}^n(1,1), \tilde{Y}_{2,b}^n(1,1), V_{b}^n({n}_{b}|1)\big)\nonumber \\ & &\hspace{4cm} \in \set{T}_{\varepsilon/2}^{(n)}(P_{U_0U_1U_2\tilde{Y}_1\tilde Y_2V}).
\end{IEEEeqnarray*}
By the Covering Lemma, $\Pr(\Epsilon_{5,b}|\Epsilon_{3,b}^c)$ tends to 0 as $n\to \infty$, if}
\begin{equation}
\tilde{R}_v \geq I(U_0,U_1,U_2,\tilde{Y}_1, \tilde{Y}_2;V)+\delta(\varepsilon).
\end{equation}


\item Let $\Epsilon_{6,1,b}$ be the event that 
\begin{IEEEeqnarray*}{rCl}\lefteqn{
\Big(U^n_{0,b}(\mathbf{1}_{[3]},  1),U^n_{1,b}(1,1|\mathbf{1}_{[3]},  1),}\qquad \\ 
&&\quad V_{b}^n(1|1), Y_{1,b}^n, \tilde{Y}_{1,b}^n(1,1)\Big)   \in \set{T}_\varepsilon^{(n)}(P_{U_0U_1VY_1\tilde{Y}_1}).
\end{IEEEeqnarray*}
By the Markov Lemma $\Pr(\Epsilon_{6,1,b}|\Epsilon_{3,b}^c,\Epsilon_{5,b}^c)$ tends to zero as $n\to \infty$.

\item Let $\Epsilon_{6,2,b}$ be the event that 
\begin{IEEEeqnarray*}{rCl}\lefteqn{
\Big(U^n_{0,b}(\mathbf{1}_{[3]},  1),U^n_{2,b}(1,1|\mathbf{1}_{[3]},  1),}\qquad \\ 
&&\quad V_{b}^n(1|1), Y_{2,b}^n, \tilde{Y}_{2,b}^n(1,1)\Big)  \in \set{T}_\varepsilon^{(n)}(P_{U_0U_2VY_2\tilde{Y}_2}).
\end{IEEEeqnarray*}
By the Markov Lemma $\Pr(\Epsilon_{6,2,b}|\Epsilon_{3,b}^c,\Epsilon_{5,b}^c)$ tends to zero as $n\to \infty$.

\item Let $\Epsilon_{7,1,b}$ be the event that 
there is a tuple  $({{\hat{m}}^{(1)}_{c,b}},\hat{n}_{b-1},\hat{m}_{p,1,b},\hat{k}_{1,b}) \in\set{M}_c\times \set{N}\times \set{M}_{p,1}\times \set{K}_{1}$ that is not equal to the all-one tuple $(\mathbf{1}_{[3]}, 1, 1, 1)$ and that satisfies
\begin{IEEEeqnarray*}{rCl}\lefteqn{
\Big(U^n_{0,b}({{\hat{m}}^{(1)}_{c,b}},  \hat{n}_{b-1}),U^n_{1,b}(\hat{m}_{p,1,b},\hat{k}_{1,b}|{{\hat{m}}^{(1)}_{c,b}}, \hat{n}_{b-1}),}\qquad \\ 
&&\hspace{1cm} V_{b}^n(1|\hat{n}_{b-1}), Y_{1,b}^n, \tilde{Y}_{1,b}^n(1,1)\Big)\\ 
&&\hspace{3cm}  \in \set{T}_\varepsilon^{(n)}(P_{U_0U_1VY_1\tilde{Y}_1}).
\end{IEEEeqnarray*}
By the  Packing Lemma, we conclude that $\Pr(\Epsilon_{7,1,b}|\Epsilon_{6,1,b}^c)$ tends to zero as $n\to \infty$ if
\begin{IEEEeqnarray}{rCl}
&&R_{1}+R_{c,2}+R'_1\leq I(U_0,U_1;Y_1, \tilde{Y}_1,V)-\delta(\varepsilon)\IEEEeqnarraynumspace\\
&&R_{1}\!+\!R_{c,2}\!+\!\tilde{R}_v\!+\!R'_1\leq I(U_0,U_1;Y_1, \tilde{Y}_1,V)\nonumber \\ 
& &\hspace{3.8cm}+ I(V;\tilde{Y}_1,Y_1)-\delta(\varepsilon)\IEEEeqnarraynumspace\\
&&R_{p,1}+R'_1\leq I(U_1;Y_1,\tilde{Y}_1,V|U_0)-\delta(\varepsilon).\IEEEeqnarraynumspace
\end{IEEEeqnarray}
\item Let $\Epsilon_{7,2,b}$ be the event that 
there is a tuple  $({{\hat{m}}^{(2)}_{c,b}},\hat{n}_{b-1},\hat{m}_{p,2,b},\hat{k}_{2,b}) \in\set{M}_c\times \set{N}\times \set{M}_{p,2}\times \set{K}_{2}$ that is not equal to the all-one tuple $(\mathbf{1}_{[3]}, 1, 1, 1)$ and  that satisfies
\begin{IEEEeqnarray*}{rCl}\lefteqn{
\Big(U^n_{0,b}({{\hat{m}}^{(2)}_{c,b}},  \hat{n}_{b-1}),U^n_{2,b}(\hat{m}_{p,2,b},\hat{k}_{2,b}|{{\hat{m}}^{(2)}_{c,b}}, \hat{n}_{b-1}),}\qquad \\ 
&&\hspace{1cm} V_{b}^n(1|\hat{n}_{b-1}), Y_{2,b}^n, \tilde{Y}_{2,b}^n(1,1)\Big) \nonumber \\ 
& &\hspace{4cm} \in \set{T}_{\varepsilon}^{(n)}(P_{U_0U_2VY_2\tilde{Y}_2}).
\end{IEEEeqnarray*}
By the Markov Lemma and the Packing Lemma, we conclude that $\Pr(\Epsilon_{7,2,b}|\Epsilon_{6,2,b}^c)$ tends to zero as $n\to \infty$, if
\begin{IEEEeqnarray}{rCl}
&&R_{2}+R_{c,1}+R'_2\leq I(U_0,U_2;Y_2, \tilde{Y}_2,V)-\delta(\varepsilon)\IEEEeqnarraynumspace\\
&&R_{2}\!+\!R_{c,1}\!+\!\tilde{R}_v\!+\!R'_2\leq I(U_0,U_2;Y_2, \tilde{Y}_2,V)\nonumber \\ 
& &\hspace{3.6cm}+ I(V;\tilde{Y}_2,Y_2)-\delta(\varepsilon)\IEEEeqnarraynumspace\\
&&R_{p,2}+R'_2\leq I(U_2;Y_2,\tilde{Y}_2,V|U_0)-\delta(\varepsilon).\IEEEeqnarraynumspace \label{eq:lastcon}
\end{IEEEeqnarray}
\end{itemize}

Whenever the event $\Epsilon_{b+1}^{c}$ occurs but none of the  events above, then $\Epsilon^c_b$. Therefore,
\begin{IEEEeqnarray}{rCl}
\lefteqn{
\Prv{\Epsilon_b|\Epsilon_{b+1}^c} }\nonumber \\& \leq&  \textnormal{Pr}\Big[ \Epsilon_{0,b} \cup \Epsilon_{1,b} \cup \Epsilon_{2,1,b} \cup \Epsilon_{2,2,b} \cup \Epsilon_{3,b} \nonumber \\ & & 
 \qquad\qquad \cup \Epsilon_{4,b}\cup \Epsilon_{5,b} \cup \Epsilon_{6,1,b}\cup \Epsilon_{6,2,b}\big|\Epsilon_{b+1}^c \Big] \nonumber \\
 & \leq& \Prv{ \Epsilon_{0,b}\big|\Epsilon_{b+1}^c }+\Prv{\Epsilon_{1,b} |\Epsilon_{0,b}^c,\Epsilon_{b+1}^c} \nonumber \\ & &+ \sum_{i=1}^2  \Prv{\Epsilon_{2,i,b} |\Epsilon_{1,b}^c,\Epsilon_{b+1}^c}  \nonumber \\ & &
 + \Prv{ \Epsilon_{3,b}|\Epsilon_{1,b}^c, \Epsilon_{2,1,b}^c, \Epsilon_{2,2,b}^c,\Epsilon_{b+1}^c}+\Prv{ \Epsilon_{4,b}|\Epsilon_{3,b}^c,\Epsilon_{b+1}^c } \nonumber \\
  & &+\Prv{ \Epsilon_{5,b}|\Epsilon_{3,b}^c,\Epsilon_{b+1}^c} + \sum_{i=1}^2 \Prv{\Epsilon_{6,i,b}| \Epsilon_{3,b}^c,\Epsilon_{b+1}^c}\nonumber \\
& =& \Prv{ \Epsilon_{0,b}}+\Prv{\Epsilon_{1,b} |\Epsilon_{0,b}^c}+ \sum_{i=1}^2  \Prv{\Epsilon_{2,i,b} |\Epsilon_{1,b}^c}  \nonumber \\ & &
 + \Prv{ \Epsilon_{3,b}|\Epsilon_{1,b}^c, \Epsilon_{2,1,b}^c, \Epsilon_{2,2,b}^c}+\Prv{ \Epsilon_{4,b}|\Epsilon_{3,b}^c } \nonumber \\ 
& &+\Prv{ \Epsilon_{5,b}|\Epsilon_{3,b}^c}+ \sum_{i=1}^2 \Prv{\Epsilon_{6,i,b}| \Epsilon_{3,b}^c},\label{eq:probeSch2}\IEEEeqnarraynumspace
\end{IEEEeqnarray} 
where the last equality follows because the channel is memoryless and the codebooks in blocks $b$ and $b+1$ have been chosen independently. 
As explained in the previous paragraphs, each of the terms in the last five lines tends to 0 as $n\to \infty$, if Constraints~\eqref{conMartonSch2}--\eqref{eq:lastcon} are satisfied. 
Thus, by~\eqref{eq:proberrorSch2} and \eqref{eq:probeSch2} we conclude that the probability of error $P_e^{(N)}$ (averaged over all code constructions) vanishes as $n\to\infty$ if Constraints~\eqref{conMartonSch2}--\eqref{eq:lastcon} hold. Letting $\varepsilon \to 0$, 
 we obtain that the probability of  error can be made to tend to 0 as $n\to \infty$ whenever
 \begin{subequations}\label{method4}
\begin{IEEEeqnarray}{rCl}
R'_1+R'_2 &>& I(U_1;U_2|U_0) \\
\hat{R}_1+\tilde{R}_1& > &I(\tilde{Y}_1;Y_1)\\
\hat{R}_2+\tilde{R}_2 &>& I(\tilde{Y}_2;Y_2)\\
\hat{R}_1 & < & I(U_0,U_1,U_2,\tilde{Y}_2;\tilde{Y}_1)\\
\hat{R}_2 & < & I(U_0,U_1,U_2,\tilde{Y}_1;\tilde{Y}_2)\\
\hat{R}_1 +\hat{R}_2& < & I(U_0,U_1,U_2;\tilde{Y}_1,\tilde{Y}_2)\nonumber \\ &&+I(\tilde{Y}_1;\tilde{Y}_2)\\
\tilde{R}_v & > & I(U_0,U_1,U_2,\tilde{Y}_1,\!\tilde{Y}_2;V\!)\qquad \\
R_{1}+R_{c,2}+\tilde{R}_v+R'_1
&<& I(U_0,U_1;Y_1,\tilde{Y}_1,V)\nonumber \\ &&+I(V;\tilde{Y}_1,Y_1)\\
R_{1}+R_{c,2}+R'_1
&<& I(U_0,U_1;Y_1,\tilde{Y}_1,V)\\
R_{c,1}+R_{2}+R'_2
&<& I(U_0,U_2;Y_2,\tilde{Y}_2,V)\\
R_{c,1}+R_{2}+\tilde{R}_v+R'_2
&<& I(U_0,U_2;Y_2,\tilde{Y}_2,V)\nonumber \\ &&+I(V;\tilde{Y}_2,Y_2)\\
R_{p,1}+R_1'&<&I(U_1;Y_1,\tilde{Y}_1,V|U_0)\\
R_{p,2}+R_2'&<&I(U_2;Y_2,\tilde{Y}_2,V|U_0).
\end{IEEEeqnarray}
Moreover, the feedback-rate constraints~\eqref{consFB0} impose that:
\begin{IEEEeqnarray}{rCl}
\tilde{R}_1 & \leq & R_{\FB,1}\\
\tilde{R}_2 & \leq & R_{\FB,2}.
\end{IEEEeqnarray}
\end{subequations}
Eliminating the auxiliaries $\tilde{R}_1, \tilde{R}_2, \hat{R}_1, \hat{R}_2, \tilde{R}_v$ from the above (using the Fourier-Motzkin algorithm), we obtain: 
 \begin{subequations}\label{method5}
\begin{IEEEeqnarray}{rCl}
R'_1+R'_2 &>& I(U_1;U_2|U_0) \\
R_{1}+R_{c,2}+R'_1
&<& I(U_0,U_1;Y_1,\tilde{Y}_1,V)\nonumber \\ &&-I(V;U_0,U_1,U_2, \tilde{Y}_2|\tilde{Y}_1,Y_1)\\
R_{c,1}+R_{2}+R'_2
&<& I(U_0,U_2;Y_2,\tilde{Y}_2,V)\nonumber \\ &&-I(V;U_0,U_1,U_2,\tilde{Y}_1|\tilde{Y}_2,Y_2)\qquad \\
R_{p,1}+R_1'&<&I(U_1;Y_1,\tilde{Y}_1,V|U_0)\\
R_{p,2}+R_2'&<&(U_2;Y_2,\tilde{Y}_2,V|U_0)
\end{IEEEeqnarray}
\end{subequations}
\begin{subequations}\label{eq:method6}
where the feedback-rate constraints have to satisfy
\begin{IEEEeqnarray}{rCl}
I(Y_1;\tilde{Y}_1|U_0,U_1,U_2,\tilde{Y}_2) & \leq & R_{\FB,1}\\
I(Y_2;\tilde{Y}_2|U_0,U_1,U_2,\tilde{Y}_1)  & \leq & R_{\FB,2}\\
I(Y_1,Y_2;\tilde{Y}_1,\tilde{Y}_2|U_0,U_1,U_2)  & \leq & R_{\FB,1}+R_{\FB,2}.\quad
\end{IEEEeqnarray}
\end{subequations}
Applying again the Fourier-Motzkin elimination algorithm to Constraints~\eqref{method5} and keeping Constraints~\eqref{eq:method6}, we obtain the desired result in Theorem~\ref{theo5} with the additional constraint that 
\begin{equation}\label{eq:Martonextra3}
I(U_1;U_2|U_0)  \leq I(U_1;Y_1,\!\tilde{Y}_1,\!V|U_0)\! +\!I(U_2;Y_2,\!\tilde{Y}_2,\!V|U_0).
\end{equation}
Finally, this last constraint can be ignored because for any tuple $(U_0,U_1,U_2, X, {Y}_1,Y_2, \tilde{Y}_1,\tilde{Y}_2)$ that violates~\eqref{eq:Martonextra3}, the region defined by the constraints in Theorem~\ref{theo5} is contained in the time-sharing region. 

\section{Proof of Theorem~\ref{thm:useful}}\label{sec:prooftheorem}
Let $R_{\FB,1}>0$. 
Fix   a tuple $(U_0^{\M},U_1^{\M}, U_2^{\M}, X^{\M})$ and rate pairs $(R_1^{\M},R_2^\M)$ and   $(R_1^{\En}, R_2^{\En}) \in \set{C}_{\textnormal{Enh}}^{(1)}$ as stated in the theorem. Then, by  the assumptions in the theorem,
\begin{subequations}
\begin{IEEEeqnarray}{rCl}
R_1^{\M}\!\!&\leq& I(U_0^{\M}\!\!,U_1^{\M};Y_1^{\M}) \label{eq:M1}\\
R_2^{\M}\!\!&<& I(U_0^{\M}\!\!,U_2^{\M};Y_2^{\M})\label{eq:M2}\\
R_1^{\M}\!\!+\!R_2^{\M}\!&\leq&I(U_0^{\M}\!\!,U_1^{\M};Y_1^{\M})\!+\!I(U_2^{\M};\!Y_2^{\M}|U_0^{\M})\nonumber \\ & & -I(U_1^{\M};U_2^{\M}|U_0^{\M}),\label{eq:M3}
\end{IEEEeqnarray}
\end{subequations}
where $Y_{1}^{\M}$ and $Y_{2}^{\M}$ denote the outputs of the considered DMBC corresponding to input $X^{\M}$. (Notice the strict inequality of the second constraint.)

 By the definition of $\set{C}_{\textnormal{Enh}}^{(1)}$ we can  identify random variables $U_0^{\En}$ and $X^{\En}$ such that 
 \begin{subequations}
 \begin{IEEEeqnarray}{rCl}
R_1^{\En}&\leq&  I(U_0^{\En};Y_1^{\En})\label{eq:E1} \\
R_2^{\En}&\leq&  I(X^{\En};Y_1^{\En}, Y_2^{\En}|U_0^{\En}),\label{eq:E2}
 \end{IEEEeqnarray}
 \end{subequations}
where $Y_{1}^{\En}$ and $Y_{2}^{\En}$ denote the outputs of the considered DMBC corresponding to input $X^{\En}$. 

Define further $U_{1}^{\En}=\textnormal{const.}$, $U_2^{\En}=X^{\En}$, $\tilde{Y}_1^{\En}=Y_1^{\En}$,  $\tilde{Y}_1^{M}=$const,  and  a binary random variable $Q$ independent of all  previously defined random variables and of pmf
\begin{IEEEeqnarray}{rCl}
  P_Q(q)=\left\{
  \begin{array}{ll}
\gamma, & q=\textnormal{Enh}\\
  1-\gamma,  & q=\textnormal{M}.
  \end{array} \right.
 \end{IEEEeqnarray}

We show that when $\gamma$ is sufficiently small, then the random variables  
\begin{IEEEeqnarray}{rCl}\label{eq:choice33}
\lefteqn{U_0:=U_0^{(Q)}, \; U_1:=U_1^{(Q)},
\;U_2:=U_2^{(Q)}\;}\qquad \nonumber \\ && X:=X^{(Q)}, \textnormal{ and } \tilde{Y}_1:=\tilde{Y}_1^{(Q)}
\end{IEEEeqnarray}
 satisfy the feedback rate constraints~\eqref{eq:fbRateThm3} and  the rate pair $(R_1', R_2')$,
\begin{subequations}\label{eq:pointout}
\begin{IEEEeqnarray}{rClCl}\label{barRQ}
{R}'_1&:=&(1-\gamma) R^{\M}_1+ \gamma R^{\En}_1 \label{eq:Rt1}\\
{R}'_2&:=&(1-\gamma) R^{\M}_2+\gamma R^{\En}_2  \label{eq:Rt2},\end{IEEEeqnarray}
 \end{subequations}
 satisfies the constraints in~\eqref{eq:region_relays} for the choice in \eqref{eq:choice33}. The two imply that the rate pair $(R_1',R_2')$ lies in $\set{R}_{\textnormal{relay,hb}}^{(1)}$ and concludes our proof.

Notice that the pmf of the tuple $U_0,U_1,U_2, X, Y_1,Y_2, \tilde{Y}_1$ has the desired form 
\begin{equation}\label{fbMC}
P_Q  P_{U_0U_1U_2|Q} P_{X|U_0U_1U_2Q} P_{Y_1Y_2|X} P_{\tilde{Y}_1| Y_1Q}.
\end{equation}
where $P_{Y_1Y_2|X}$ denotes the channel law.

For the  described choice of random variables~\eqref{eq:choice33}, the feedback-rate constraint~\eqref{eq:fbRateThm3} specializes to
\begin{equation}
\gamma H(Y_1^{\En}|Y_2^{\En}, {X^{\En}}) \leq R_{\FB,1},
\end{equation}
which is satisfied for all sufficiently small $\gamma\in(0,1)$. 
Moreover, for this choice the constraints in~\eqref{eq:region_relays} specialize to
\begin{subequations}
\begin{IEEEeqnarray}{rCl}
R_1 & \leq &(1- \gamma) I(U_0^{\M}, U_1^{\M}; Y_1^{\M}) \nonumber \\ &&  + \gamma I(U^{\En}_0;Y_1^{\En})\label{eq:one}\\
R_2 & \leq & (1-\gamma) I(U_0^{\M}, U_2^{\M}; Y_2^{\M})\nonumber \\ &&  + \gamma \big(I(X^{\En};Y_1^{\En}, Y_2^{\En}) \nonumber \\
&& \qquad \qquad\qquad - H(Y_1^{\En}|Y_2^{\En})\big)\label{eq:two}\\
R_1+R_2 & \leq & (1-\gamma) \big( I(U_0^{\M}, U_1^{\M}; Y_1^{\M}) \nonumber \\ 
&&\qquad \qquad \qquad + I(U_2^{\M};Y_2^{\M}|U_0^\M)\nonumber \\ &&  \hspace{3.3cm}- I(U_1^{\M};U_2^{\M}|U_0^{\M})\big) \nonumber \\\label{eq:third}
 & & + \gamma\big( I(X^{\En};Y_1^{\En},Y_2^{\En}|U_0^{\En})\nonumber \\ & &  \qquad \qquad\qquad + I(U_0^{\En};Y_1^{\En}) \big)\\
 R_1+R_2 & \leq & (1-\gamma) \big( I(U_1^{\M}; Y_1^{\M}|U_0^\M)\nonumber \\ & & \qquad \qquad  \quad + I(U_0^{\M},U_2^{\M};Y_2^{\M}) \nonumber \\ &&  \hspace{3.3cm}- I(U_1^{\M};U_2^{\M}|U_0^{\M})\big) \nonumber \\
 & & + \gamma\big(  I(X^{\En};Y_1^{\En},Y_2^{\En})\nonumber \\ 
 && \qquad \qquad\qquad-H(Y_1^{\En}|Y_2^{\En})\big).\label{eq:last}
\end{IEEEeqnarray}
\end{subequations}
{We argue in the following that the rate pair $(R_1=R_1', R_2=R_2')$ defined in \eqref{eq:pointout} satisfies these constraints for all sufficiently small $\gamma>0$. }
Comparing \eqref{eq:M1}, \eqref{eq:E1}, and \eqref{eq:Rt1}, we see that the first constraint~\eqref{eq:one} is satisfied for any choice of $\gamma\in[0,1]$.  Similarly, comparing~\eqref{eq:M3}, \eqref{eq:E1}, \eqref{eq:E2}, and \eqref{eq:Rt1} and \eqref{eq:Rt2}, we note that also the third constraint~\eqref{eq:third} is satisfied for any $\gamma\in[0,1]$.  
The second constraint~\eqref{eq:two} is  satisfied when $\gamma$ is sufficiently small. This can be seen by comparing~\eqref{eq:M2}, \eqref{eq:E2}, and \eqref{eq:Rt2}, and because Constraint~\eqref{eq:M2} holds with strict inequality. 
The last constraint~\eqref{eq:last} is not active in view of Constraint~\eqref{eq:third} whenever
\begin{equation}\label{eq:gamma1}
\gamma H(Y_1^{\En}|Y_2^{\En}) \leq (1-\gamma) {\Gamma^{\M}},
\end{equation}
{where $\Gamma^{\M}$ is defined in \eqref{eq:delta}.}
Thus, also this last constraint is satisfied when $\gamma$ is sufficiently small. This concludes our proof.

%
%

\section{Proof of Proposition~\ref{prop:noisy_fb}}\label{sec:proof_noisy}

{Let $\Epsilon_{\text{Fb},i,b}$, for $i=1,2$,  denote the event that   during block~$b$ there is an error in the  feedback communication    from Receiver $i$  to the transmitter, and let $\varepsilon$ denote the event that $\hat{M}_1\neq M_1~\text{or}~\hat{M}_2\neq M_2$.
Then,
\begin{IEEEeqnarray}{rCl}
\Pr [\hat{M}_1&&\neq M_1~\text{or}~\hat{M}_2\neq M_2]\nonumber\\
&&\leq \Prv{ \Epsilon  \cup \left(\bigcup^B_{b=1} \Epsilon_{\text{Fb},1,b} \right)\cup \left(\bigcup^B_{b=1} \Epsilon_{\text{Fb},2,b} \right)}\nonumber\\
&&\leq \Prv{ \Epsilon\Big| \left(\bigcup^B_{b=1} \Epsilon_{\text{Fb},1,b} \right)^c \cap  \left(\bigcup^B_{b=1} \Epsilon_{\text{Fb},2,b} \right)^c} \nonumber\\
&&\;+\Prv{ \bigcup^B_{b=1} \Epsilon_{\text{Fb},1,b}}+\Prv{\bigcup^B_{b=1} \Epsilon_{\text{Fb},2,b}}\nonumber\\
&&\leq\Prv{ \Epsilon\Big| \left(\bigcup^B_{b=1} \Epsilon_{\text{Fb},1,b} \right)^c \cap  \left(\bigcup^B_{b=1} \Epsilon_{\text{Fb},2,b} \right)^c }\nonumber\\
&&\;+\! \sum^B_{b=1}\! \Prv{ \Epsilon_{\text{Fb},1,b}}+\Prv{ \Epsilon_{\text{Fb},2,b} }.\label{eq:setin}
\end{IEEEeqnarray}
Since we use capacity-achieving codes on the feedback links,  the probabilities $\Prv{ \Epsilon_{\text{Fb},1,b}}$ and $\Prv{ \Epsilon_{\text{Fb},2,b}}$ vanish as the blocklength increases. When the feedback communications in  all the blocks are error-free, then the probability of error  in the setup with noisy feedback is no larger than {that} in the setup with noise-free feedback. Thus, under the corresponding rate constraints, also the probability $\Prv{ \Epsilon\Big| \left(\bigcup^B_{b=1} \Epsilon_{\text{Fb},1,b} \right)^c \cap  \left(\bigcup^B_{b=1} \Epsilon_{\text{Fb},2,b} \right)^c }$ vanishes as the blocklength increases.

}

\begin{IEEEbiographynophoto}{Youlong WU} (S'13-M'15) Youlong Wu obtained his B.S. degree in electrical engineering from Wuhan University, Wuhan, China, in 2007. He received the M.S. degree in electrical engineering from Shanghai Jiaotong University, Shanghai, China, in 2011. In 2014, he received the Ph.D. degree at Telecom ParisTech, in Paris, France. Since December 2014, he has been a postdoc at the Institute for Communication Engineering, Technische Universit\"at M\" unchen (TUM), Munich, Germany. His research interests include information theory and wireless communication.
\end{IEEEbiographynophoto}

\begin{IEEEbiographynophoto}{Mich\`ele Wigger}
(S'05-M'09) received the M.Sc. degree in electrical
engineering (with distinction) and the Ph.D. degree in electrical 
engineering
both from ETH Zurich in 2003 and 2008, respectively. In 2009 she was
a postdoctoral researcher at the ITA center at the University of California,
San Diego. Since December 2009 she is an Assistant Professor at Telecom
ParisTech, in Paris, France.
Her research interests are in information and
communications theory.
\end{IEEEbiographynophoto}

\end{document}